\documentclass[a4paper, 10pt]{article}

\usepackage{arxiv}

\usepackage[utf8]{inputenc}
\usepackage[LSF, T1]{fontenc}
\usepackage{chessfss}
\usepackage{dkTensor}
\usepackage{graphicx}
\usepackage[
    style = authoryear,
    isbn = false,
    maxbibnames = 50,
    maxcitenames = 2,
    autocite = inline,
    block = space,
    backrefstyle = three+,
    date = year,
    giveninits = false,
    clearlang = false,
    uniquename = false,
    sortcites = false,
    natbib = false,
    dashed = true,
    url = true,
    doi = true,
    bibencoding = utf8
]{biblatex}

\usepackage[pdftex, colorlinks, allcolors=blue]{hyperref}
\usepackage[noabbrev, capitalize, nameinlink]{cleveref}

\title{Generalized Multi-Linear Models for Sufficient Dimension Reduction on Tensor Valued Predictors}
\author{
    \textbf{Daniel Kapla} \\
    Institute of Statistics and Mathematical Methods in Economics \\
    Faculty of Mathematics and Geoinformation \\
    TU Wien, Vienna, Austria
    \And \\
    \textbf{Efstathia Bura} \\
    \hspace{.2cm}
    Institute of Statistics and Mathematical Methods in Economics \\
    Faculty of Mathematics and Geoinformation \\
    TU Wien, Vienna, Austria
}
\date{\today}
\AtBeginDocument{
    \hypersetup{
        pdftitle = {GMLM SDR},
        pdfauthor = {Daniel Kapla},
        pdfcreator = {\pdftexbanner}
    }
}

\theoremstyle{plain}
\newtheorem{theorem}{Theorem}
\newtheorem{lemma}{Lemma}

\newtheorem{example}{Example}
\theoremstyle{definition}
\newtheorem{definition}{Definition}
\newtheorem{condition}{Condition}
\theoremstyle{remark}
\newtheorem{remark}{Remark}

\crefname{condition}{Condition}{Conditions}
\Crefname{condition}{Condition}{Conditions}
\crefrangelabelformat{condition}{#3#1#4-#5#2#6}

\renewcommand*{\t}[1]{{#1^{T}}}
\renewcommand*{\d}{\mathrm{d}}

\begin{document}

\maketitle

\begin{abstract}
    We consider supervised learning (regression/classification) problems with tensor-valued input. We derive multi-linear sufficient reductions for the regression or classification problem by modeling the conditional distribution of the predictors given the response as a member of the quadratic exponential family. We develop estimation procedures of sufficient reductions for both continuous and binary tensor-valued predictors. We prove the consistency and asymptotic normality of the estimated sufficient reduction using manifold theory. For continuous predictors, the estimation algorithm is highly computationally efficient and is also applicable to situations where the dimension of the reduction exceeds the sample size. We demonstrate the superior performance of our approach in simulations and real-world data examples for both continuous and binary tensor-valued predictors.
\end{abstract}

\section{Introduction}\label{sec:introduction}
Complex data are collected at different times and/or under conditions often involving a large number of multi-indexed variables represented as tensor-valued data \cite[]{KoldaBader2009}. \textit{Tensors} are a mathematical tool to represent data of complex structure in that they are a generalization of matrices to higher dimensions: A tensor is a multi-dimensional array of numbers. They occur in large-scale longitudinal studies \cite[e.g.][]{Hoff2015}, in agricultural experiments and chemometrics and spectroscopy \cite[e.g.][]{LeurgansRoss1992,Burdick1995}, in signal and video processing where sensors produce multi-indexed data, e.g. over spatial, frequency, and temporal dimensions \cite[e.g.][]{DeLathauwerCastaing2007,KofidisRegalia2005}, and in telecommunications \cite[e.g.][]{DeAlmeidaEtAl2007}. Other examples of multiway data include 3D images of the brain, where the modes are the 3 spatial dimensions, and spatio-temporal weather imaging data, a set of image sequences represented as 2 spatial modes and 1 temporal mode.

Tensor regression models have been proposed to leverage the structure inherent in tensor-valued data. For instance, \cite{HaoEtAl2021,ZhouLiZhu2013} focus on tensor covariates, while \cite{RabusseauKadri2016,LiZhang2017,ZhouLiZhu2013} focus on tensor responses, and \cite{Hoff2015,Lock2018} consider tensor on tensor regression. \cite{HaoEtAl2021} modeled a scalar response as a flexible nonparametric function of tensor covariates. \cite{ZhouLiZhu2013} assume the scalar response has a distribution in the exponential family given the tensor-valued predictors with the link modeled as a multi-linear function of the predictors. \cite{RabusseauKadri2016} model the tensor-valued response as a linear model with tensor-valued regression coefficients subject to a multi-linear rank constraint. \cite{LiZhang2017} approach the problem with a similar linear model but instead of a low-rank constraint the error term is assumed to have a separable Kronecker product structure while using a generalization of the envelope model \cite[]{CookLiChiaromonte2010}. \cite{ZhouEtAl2023} focus on partially observed tensor response given vector-valued predictors with mode-wise sparsity constraints in the regression coefficients. \cite{Hoff2015} extends an existing bilinear regression model to a tensor on tensor of conformable modes and dimensions regression model based on a Tucker product. \cite{Lock2018} uses a tensor contraction to build a penalized least squares model for a tensor with an arbitrary number of modes and dimensions.

We consider the general regression problem of fitting a response of general form (univariate, multivariate, tensor-valued) on a tensor-value predictor. We operate in the context of sufficient dimension reduction (SDR) \cite[e.g.][]{Cook1998,Li2018} based on inverse regression, which leads us to regressing the tensor-valued predictor on tensor-valued functions of the response (tensor on tensor regression). We assume the conditional distribution of the tensor-valued predictors given the response belongs to the quadratic exponential family for which the first and second moments admit a separable Kronecker product structure. The quadratic exponential family contains the multi-linear normal and the multi-linear Ising distributions, for continuous and binary tensor-valued random variables, respectively. 

Our tensor-to-tensor inverse regression model is a generalized multi-linear model similar to the generalized linear model of \cite{ZhouLiZhu2013, Hoff2015} and \cite{SunLi2017}. Our approach differs in that it is a model-based \emph{Sufficient Dimension Reduction} (SDR) method for tensor-valued data; that is, we obtain the maximum dimension reduction of the tensor-valued predictor without losing any information about the response. Our estimates of the sufficient reduction are maximum likelihood based and enjoy the attendant optimality properties, such as consistency, efficiency and asymptotic normality. The main challenge in estimation in multi-linear tensor regression models is the non-identifiability of the parameters, as has been acknowledged by researchers in this field (e.g., \cite{LiZhang2017, Lock2018}). Our approach  does not require any penalty terms and/or sparse approaches in order to address this issue, in contrast to, e.g., \cite{ZhouLiZhu2013,Hoff2015, SunLi2017, LiZhang2017, LlosaMaitra2023}. Instead, we consider the parameter space as a smooth embedded manifold, which not only leads to identifiability but also allows great modeling flexibility.

Our main contributions are (i) formulating the dimension reduction problem via the quadratic exponential family modeling up to two-way interaction that allows us to derive the minimal sufficient dimension reduction in closed form, (ii) great flexibility in modeling by defining the parameter space as a smooth embedded manifold, (iii) deriving the maximum likelihood estimator of the sufficient reduction subject to multi-linear constraints and overcoming parameter non-identifiability, (iv) establishing the consistency and asymptotic normality of the estimators, and (v) developing computationally efficient estimation algorithms (very fast in the case of multi-linear normal predictors).

We illustrate our proposed methodology and contributions using a simple example where the predictor is matrix-valued (tensor of order $2$) and the response is univariate (binary). The electroencephalography (EEG) data set of $77$ alcoholic and $45$ control subjects is commonly used in EEG signal analysis.\footnote{\textsc{Begleiter, H.} (1999). EEG Database. \textit{Neurodynamics Laboratory, State University of New York Health Center}. Donated by Lester Ingber. \url{http://kdd.ics.uci.edu/databases/eeg/eeg.data.html}}. EEG is a noninvasive neuroimaging technique that involves the placement of electrodes on the scalp to record electrical activity of the brain. Each subject data point consists of a $p_1\times p_2 = 256\times 64$ matrix, with each row representing a time point and each column a channel. The measurements were obtained by exposing each individual to visual stimuli and measuring voltage values from $64$ electrodes placed on the subjects' scalps sampled at $256$ time points over $1$ second ($256$ Hz). Different stimulus conditions were used, and for each condition, $120$ trials were measured. As typically done in the analysis of this data set (e.g., \cite{LiKimAltman2010,PfeifferForzaniBura2012,DingCook2015,PfeifferKaplaBura2021}), we use a single stimulus condition (S1), and for each subject, we average over all the trials under this condition. That is, $(\ten{X}_i, y_i)$, $i = 1, \ldots, 122$, where $\ten{X}_i$ is a $256\times 64$ matrix (two-mode tensor), with each entry representing the mean voltage value of subject $i$ at a combination of a time point and a channel, averaged over all trials under the S1 stimulus condition, and $Y$ is a binary outcome variable with $Y_i = 1$ for an alcoholic and $Y_i = 0$ for a control subject.

We assume $\ten{X} \mid Y$ follows the bilinear model
\begin{equation}\label{eq:bilinear}
    \ten{X} = \bar{\mat{\eta}} + \mat{\alpha}_1\mat{F_Y}\t{\mat{\alpha}_2} + \mat{\epsilon}
\end{equation}
where $\mat{F}_Y\equiv Y$ is a $1\times 1$ matrix (in general, $\mat{F}_Y$ is a tensor-valued function of $Y$), $\vec{\mat{\epsilon}}\sim\mathcal{N}(\mat{0}, \mat{\Sigma}_2\otimes\mat{\Sigma}_1)$, and $\otimes$ signifies the Kronecker product operator. Model \eqref{eq:bilinear} expresses that $\E(\ten{X}\mid Y) = \bar{\mat{\eta}} + \mat{\alpha}_1\mat{F}_Y\t{\mat{\alpha}_2}$ solely contains the information in $Y$ about $\ten{X}$. Alternatively, \eqref{eq:bilinear} can be written as $\ten{X} = \bar{\mat{\eta}} + \mat{\alpha}_1 \mat{F_Y}\t{\mat{\alpha}_2} + \mat{\Sigma}_1^{-1/2}\mat{U}\mat{\Sigma}_2^{-1/2}$, where $\vec{\mat{U}}\sim\mathcal{N}(\mat{0}, \mat{I}_{p_2}\otimes\mat{I}_{p_1})$. Thus, $\ten{X}$ follows a matrix normal distribution with mean $\bar{\mat{\eta}} + \mat{\alpha}_1 \mat{F_Y}\t{\mat{\alpha}_2}$ and variance-covariance structure $\mat{\Sigma}_2\otimes\mat{\Sigma}_1$. The \textit{minimal sufficient reduction} of the predictor $\ten{X}$ for the classification (regression) of $Y$ on $\ten{X}$ is
\begin{displaymath}
    \ten{R}(\ten{X})
    = \t{(\mat{\Sigma}_1^{-1}\mat{\alpha}_1)}(\ten{X}-\bar{\mat{\eta}})(\mat{\Sigma}_2^{-1}\mat{\alpha}_2)
    = \t{\mat{\beta}_1}(\ten{X}-\bar{\mat{\eta}})\mat{\ten{\beta}_2}.
\end{displaymath}
Our approach involves optimization of the likelihood of $\ten{X}$ conditional on $Y$ subject to the non-convex nonlinear constraint that both the mean and covariance parameter spaces have Kronecker product structure, where the components of the Kronecker product are non-identifiable. This is resolved by letting the parameter space be a smooth embedded submanifold. We derive the maximum likelihood estimator of the minimal sufficient reduction and use it to predict the binary response. Our method exhibits uniformly better classification accuracy when compared to competing existing techniques (see Table~\ref{tab:eeg}). An important aspect of our approach is that it bypasses the dimensionality issue ($p_1 \times p_2 = 256 \times 64 = 16384 \gg 122 = n$) without having to artificially reduce the dimension by preprocessing the data, as in \cite{LiKimAltman2010, PfeifferForzaniBura2012, PfeifferKaplaBura2021, DingCook2015}, or imposing simplifying assumptions and/or sparsity or regularization constraints \cite[]{PfeifferKaplaBura2021,LiZhang2017,SunLi2017, Lock2018}.

Even though our motivation is rooted in the SDR perspective, our proposal concerns inference on any regression model with a tensor-valued response and predictors of any type. Thus, our approach can be used as a stand-alone model for such data regardless of whether one is interested in deriving sufficient reductions or simply regressing tensor-valued variables on any type of predictor. In effect, our model can be viewed as an extension of reduced rank regression for tensor-valued response regression.

The structure of this paper is as follows. We introduce our notation in \cref{sec:notation}, define the problem in \Cref{sec:problem-formulation} and introduce our model in \cref{sec:gmlm-model}. We present a general maximum likelihood estimation procedure and derive specialized methods for the multi-linear normal and the multi-linear Ising distributions in \cref{sec:ml-estimation}. \Cref{sec:manifolds} is a short introduction to manifolds and provides the basis for the consistency and asymptotic normality results in \cref{sec:statprop}. Simulations for continuous and binary predictors are carried out in \cref{sec:simulations}. We apply our model to EEG data and perform a proof of concept data analysis where a chess board is interpreted as a collection of binary $8\times 8$ matrices in \cref{sec:data-analysis}. We summarize our contributions and discuss directions for future research in \cref{sec:discussion}.

\section{Notation}\label{sec:notation}

Capital Latin letters in \LaTeX{} fraktur font will denote tensors, and Latin and Greek letters in boldface will denote vectors and matrices throughout the paper. A multiway array $\ten{A}$ of dimension $q_1\times \ldots\times q_r$ is a tensor of order\footnote{Also referred to as rank, hence the variable name $r$. We refrain from using this term to avoid confusion with the rank of a matrix.} $r$, in short $r$-tensor, where $r\in\mathbb{N}$ is the number of its modes or axes (dimensions), and we denote it by $\ten{A}\in\mathbb{R}^{q_1\times \ldots\times q_r}$ where $\ten{A}_{i_1,...,i_r} \in \mathbb{R}$ is its $(i_1, \ldots, i_r)$th entry. For example, a $p \times q$ matrix $\mat{B}$ is a tensor of \textit{order} 2 as it has two modes, the rows and columns. Given a tensor $\ten{A}\in\mathbb{R}^{q_1\times \ldots\times q_r}$ and matrices $\mat{B}_k\in\mathbb{R}^{p_k\times q_k}$, $k\in[r] = \{1, 2, \ldots, r\}$, the \emph{multi-linear multiplication}, or \emph{Tucker operator} \cite[]{Kolda2006}, is defined element-wise as
\begin{displaymath}
    (\ten{A}\times\{\mat{B}_1, \ldots, \mat{B}_r\})_{j_1, \ldots, j_r} = \sum_{i_1, \ldots, i_r = 1}^{q_1, \ldots, q_r} \ten{A}_{i_1, \ldots, i_r}(\mat{B}_{1})_{j_1, i_1} \cdots (\mat{B}_{r})_{j_r, i_r},
\end{displaymath}
and is itself an order $r$ tensor of dimension $p_1\times ...\times p_k$. The \emph{$k$-mode product} of the tensor $\ten{A}$ with the matrix $\mat{B}_k$ is 
\begin{displaymath}
    \ten{A}\times_k\mat{B}_k = \ten{A}\times\{\mat{I}_{q_1}, \ldots, \mat{I}_{q_{k-1}}, \mat{B}_{k}, \mat{I}_{q_{k+1}}, \ldots, \mat{I}_{q_r}\}.
\end{displaymath}
The notation $\ten{A}\mlm_{k\in S}\mat{B}_k$ is shorthand for the iterative application of the mode product for all indices in $S\subseteq[r]$. For example $\ten{A}\mlm_{k\in\{2, 5\}}\mat{B}_k = \ten{A}\times_2\mat{B}_2\times_5\mat{B}_5$. By only allowing $S$ to be a set, this notation is unambiguous because the mode product commutes for different modes; i.e., $\ten{A}\times_j\mat{B}_j\times_k\mat{B}_k = \ten{A}\times_k\mat{B}_k\times_j\mat{B}_j$ for $j\neq k$. For example, let $\mat{A}, \mat{B}_1, \mat{B}_2$ be matrices of matching dimensions. The bilinear mode-product and multi-linear multiplication relate to the well-known matrix-matrix multiplication as
\begin{displaymath}
    \mat{A}\times_1\mat{B}_1 = \mat{B}_1\mat{A}, \qquad
    \mat{A}\times_2\mat{B}_2 = \mat{A}\t{\mat{B}_2}, \qquad
    \mat{A}\mlm_{k = 1}^2\mat{B}_k = \mat{A}\mlm_{k \in \{1, 2\}}\mat{B}_k = \mat{B}_1\mat{A}\t{\mat{B}_2}.
\end{displaymath}
The operator $\vec$ maps an array to a vector. For a tensor $\ten{A}$ of order $r$ and dimensions $q_1, \ldots, q_r$, $\vec(\ten{A})$ is the $q_1 q_2 \ldots q_r \times 1$ vector resulting from stacking the modes of $\ten{A}$ one after the other in order $r$, then $r-1$, and so on. For example, if $\ten{A}$ is a 3-dimensional array, $\vec(\ten{A})=\t{(\t{\vec(\ten{A}_{:,:,1})},\t{\vec(\ten{A}_{:,:,2})},\ldots,\t{\vec(\ten{A}_{:,:,q_3})})}$. We use the notation $\ten{A}\equiv \ten{B}$ if and only if $\vec(\ten{A}) = \vec(\ten{B})$, and similarly for matrices $\mat{A},\mat{B}$. For tensors of order at least $2$, the \emph{flattening} (or \emph{unfolding} or \emph{matricization}) is a reshaping of the tensor into a matrix along a particular mode. For a tensor $\ten{A}$ of order $r$ and dimensions $q_1, \ldots, q_r$, the $k$-mode unfolding $\ten{A}_{(k)}$ is a $q_k\times \prod_{l=1, l\neq k}q_l$ matrix with
\begin{displaymath}
    (\ten{A}_{(k)})_{i_k, j} = \ten{A}_{i_1, \ldots, i_r}\quad\text{ with }\quad j = 1 + \sum_{\substack{l = 1\\l \neq k}}^r (i_l - 1) \prod_{\substack{m = 1\\m\neq k}}^{l - 1}q_m.
\end{displaymath}
Illustration examples of vectorization and matricization are given in \cref{app:examples}.

The \emph{inner product} between two tensors of the same number of elements is $\langle\ten{A}, \ten{B}\rangle = \t{\vec(\ten{A})}\vec(\ten{B})\in\mathbb{R}$, and the \emph{Frobenius norm} for tensors is $\|\ten{A}\|_F = \sqrt{\langle\ten{A}, \ten{A}\rangle}$. The \emph{outer product} between two tensors $\ten{A}$ of dimensions $q_1, \ldots, q_r$ and $\ten{B}$ of dimensions $p_1, \ldots, p_l$ is a tensor $\ten{A}\circ\ten{B}$ of order $r + l$ and dimensions $q_1, \ldots, q_r, p_1, \ldots, p_l$, such that $\ten{A}\circ\ten{B} \equiv (\vec\ten{A})\t{(\vec{\ten{B}})}$. The iterated outer and iterated Kronecker products are written as
\begin{displaymath}
    \bigouter_{k = 1}^{r} \mat{A}_k = \mat{A}_1\circ\ldots \circ\mat{A}_r,
    \qquad
    \bigkron_{k = 1}^{r} \mat{A}_k = \mat{A}_1\otimes\ldots \otimes\mat{A}_r
\end{displaymath}
where the order of iteration is important. Similar to the outer product, we extend the Kronecker product to $r$-tensors $\ten{A}, \ten{B}$ giving a $2r$-tensor $\ten{A}\otimes\ten{B}$.

Finally, the gradient of a function $\ten{F}(\ten{X})$ of any shape, univariate, multivariate, or tensor-valued, with argument $\ten{X}$ of any shape is defined to be the $p\times q$ matrix 
\begin{displaymath}
    \nabla_{\ten{X}}\ten{F} = \frac{\partial\t{(\vec\ten{F}(\ten{X}))}}{\partial(\vec\ten{X})},
\end{displaymath}
where the vectorized quantities $\vec{\ten{X}}\in\mathbb{R}^p$ and $\vec\ten{F}(\ten{X})\in\mathbb{R}^q$. This is consistent with the gradient of a real-valued function $f(\mat{x})$ where $\mat{x}\in\mathbb{R}^p$ as $\nabla_{\mat{x}}f\in\mathbb{R}^{p\times 1}$ \cite[][Ch.~15]{Harville1997}.

\section{Problem Formulation}\label{sec:problem-formulation}

Our goal is to reduce the complexity of inferring the cumulative distribution function (cdf) $F$ of $Y\mid \ten{X}$, where $\ten{X}$ is assumed to admit $r$-tensor structure of dimension $p_1\times ... \times p_r$ with continuous or discrete entries, and the response $Y$ is unconstrained. To this end, we assume there exists a tensor-valued function of lower dimension $\ten{R}:\ten{X}\mapsto \ten{R}(\ten{X})$ such that 
\begin{displaymath}
    F(Y\mid \ten{X}) = F(Y\mid \ten{R}(\ten{X})).
\end{displaymath}
where the tensor-valued $\ten{R}(\ten{X})$ has dimension $q_1\times...\times q_r$ with $q_j\leq p_j$, $j = 1, ..., r$, which represents a dimension reduction along every mode of $\ten{X}$. Since $\ten{R}(\ten{X})$ replaces the predictors without any effect in the conditional cdf of $Y\mid \ten{X}$, it is a \emph{sufficient reduction} for the regression $Y\mid\ten{X}$. This formulation is flexible as it allows, for example, to select ``important'' modes by reducing ``unimportant'' modes to be $1$ dimensional.

To find such a reduction $\ten{R}$, we leverage the equivalence relation pointed out in \cite{Cook2007},
\begin{equation}\label{eq:inverse-regression-sdr}
    Y\mid\ten{X} \sim Y\mid \ten{R}(\ten{X})
        \quad\Longleftrightarrow\quad
    \ten{X}\mid(Y, \ten{R}(\ten{X})) \sim \ten{X}\mid\ten{R}(\ten{X}).
\end{equation}
According to \eqref{eq:inverse-regression-sdr}, a \textit{sufficient statistic} $\ten{R}(\ten{X})$ for $Y$ in the inverse regression $\ten{X}\mid Y$, where $Y$ is considered as a parameter indexing the model, is also a \textit{sufficient reduction} for $\ten{X}$ in the forward regression $Y\mid\ten{X}$.

The factorization theorem is the usual tool to identify sufficient statistics and requires a distributional model. We assume that $\ten{X}\mid Y$ is a full rank quadratic exponential family with density
\begin{align}
    f_{\mat{\eta}_y}(\ten{X}\mid Y = y)&= h(\ten{X})\exp(\t{\mat{\eta}_y}\mat{t}(\ten{X}) - b(\mat{\eta}_y)) \notag\\
    &= h(\ten{X})\exp(\langle \mat{t}_1(\ten{X}), \mat{\eta}_{1y} \rangle + \langle \mat{t}_2(\ten{X}), \mat{\eta}_{2y} \rangle - b(\mat{\eta}_{y})) \label{eq:quad-density}
\end{align}
where $\mat{t}_1(\ten{X})=\vec \ten{X}$ and $\mat{t}_2(\ten{X})$ is linear in $\ten{X}\circ\ten{X}$. The dependence of $\ten{X}$ on $Y$ is fully captured in the natural parameter $\mat{\eta}_y$.\footnote{$\mat{\eta}_y$ is a function of the response $Y$, so it is not a parameter in the formal statistical sense. We view it as a parameter in order to leverage \eqref{eq:inverse-regression-sdr} and derive a sufficient reduction from the inverse regression.} The function $h$ is non-negative real-valued and $b$ is assumed to be at least twice continuously differentiable and strictly convex. An important feature of the \emph{quadratic exponential family} is that the distribution of its members is fully characterized by their first two moments. Distributions within the quadratic exponential family include the \emph{multi-linear normal} (\cref{sec:tensor-normal-estimation}) and \emph{multi-linear Ising model} (\cref{sec:ising_estimation}, a generalization of the (inverse) Ising model which is a multi-variate Bernoulli with up to second-order interactions) and mixtures of these two.

\section{The Generalized Multi-Linear Model}\label{sec:gmlm-model}

In model \eqref{eq:quad-density}, the dependence of $\ten{X}$ and $Y$ is absorbed in $\mat{\eta}_y$, and $\mat{t}(\ten{X})$ is the minimal sufficient statistic for $\mat{\eta}_y = (\mat{\eta}_{1y}, \mat{\eta}_{2y})$ with
\begin{displaymath}\label{eq:t-stat}
    \mat{t}(\ten{X}) = (\mat{t}_1(\ten{X}),\mat{t}_2(\ten{X}))=(\vec{\ten{X}}, \mat{T}_2\vech((\vec\ten{X})\t{(\vec\ten{X})})),
\end{displaymath}
where the $d\times p(p + 1) / 2$ dimensional matrix $\mat{T}_2$ with $p = \prod_{i = 1}^r p_i$ ensures that $\mat{\eta}_{2y}$ is of minimal dimension $d$. The matrix $\mat{T}_2$ is of full rank $d$ and unique to different members of the quadratic exponential family.
We can reexpress the exponent in \eqref{eq:quad-density} as
\begin{align*}
    \t{\mat{\eta}_y} \mat{t}(\ten{X})
        &= \langle \vec \ten{X}, \mat{\eta}_{1y} \rangle + \langle \mat{T}_2\vech(\ten{X}\circ\ten{X}), \mat{\eta}_{2y} \rangle \\
        &= \langle \vec \ten{X}, \mat{\eta}_{1y} \rangle + \langle \vec(\ten{X}\circ\ten{X}), \t{(\mat{T}_2\pinv{\mat{D}_p})}\mat{\eta}_{2y} \rangle
\end{align*}
where $\mat{D}_p$ is the \emph{duplication matrix} \cite[][Ch.~11]{AbadirMagnus2005}, defined so that $\mat{D}_p\vech \mat{A} = \vec \mat{A}$ for every symmetric $p\times p$ matrix $\mat{A}$, and $\pinv{\mat{D}_p}$ is its Moore-Penrose pseudo-inverse. The first natural parameter component, $\mat{\eta}_{1y}$, captures the first order, and $\mat{\eta}_{2y}$, the second-order relationship of $Y$ and $\ten{X}$. The quadratic exponential density of $\ten{X} \mid Y$ can then be expressed as
\begin{equation}\label{eq:quadratic-exp-fam}
    f_{\mat{\eta}_y}(\ten{X}\mid Y = y) = h(\ten{X})\exp\left(\langle \vec \ten{X}, \mat{\eta}_{1y} \rangle + \langle \vec(\ten{X}\circ\ten{X}), \t{(\mat{T}_2\pinv{\mat{D}_p})}\mat{\eta}_{2y} \rangle - b(\mat{\eta}_y)\right)
\end{equation}
The exponential family in \eqref{eq:quadratic-exp-fam} is easily generalizable to any order. This, though, would result in the number of parameters becoming prohibitive to estimate, which is also the reason why we opted for the second-order exponential family in our formulation.

By the equivalence in \eqref{eq:inverse-regression-sdr}, in order to find the sufficient reduction $\ten{R}(\ten{X})$ we need to infer $\mat{\eta}_{1y}$, and $\mat{\eta}_{2y}$. This is reminiscent of generalized linear modeling, which we extend to a multi-linear formulation next.
Suppose $\ten{F}_y$ is a known mapping of $y$ with zero expectation, $\E_Y\ten{F}_Y = 0$. We assume the dependence of $\ten{X}$ and $Y$ is reflected only in $\mat{\eta}_{1y}$ and let
\begin{align}
    \mat{\eta}_{1y} &= \vec{\overline{\ten{\eta}}} + \mat{B}\vec\ten{F}_y, \label{eq:eta1-manifold} \\
    \mat{\eta}_{2}  &= \t{(\pinv{(\mat{T}_2\pinv{\mat{D}_p})})}\vec(c\,\mat{\Omega}), \label{eq:eta2-manifold}
\end{align}
where $\overline{\ten{\eta}}\in\mathbb{R}^{p_1\times\ldots\times p_r}$, $\mat{\Omega} \in \mathbb{R}^{p \times p}$ is positive definite with $p = \prod_{j = 1}^{r} p_j$, and $c\in\mathbb{R}$ is a known constant determined by the distribution to ease modeling. That is, we assume that only $\mat{\eta}_{1y}$ depends on $Y$ through $\mat{B}$. The second parameter $\mat{\eta}_2$ captures the second-order interaction structure of $\ten{X}$, which we assume not to depend on the response $Y$. To relate individual modes of $\ten{X}$ to the response, allowing flexibility in modeling, we assume $\ten{F}_y$ takes values in $\mathbb{R}^{q_1\times ...\times q_r}$; that is, $\ten{F}_y$ is a tensor-valued independent variable. This, in turn, leads to imposing a corresponding tensor structure to the regression parameter $\mat{B}$ and \eqref{eq:eta1-manifold} becomes
\begin{align}
    \mat{\eta}_{1y} &=
    \vec\biggl(\overline{\ten{\eta}} + \ten{F}_y\mlm_{j = 1}^{r}\mat{\beta}_j\biggr), \label{eq:eta1}
\end{align}
with $\mat{B} = \bigotimes_{j = r}^{1}\mat{\beta}_j$ and the component matrices $\mat{\beta}_j\in\mathbb{R}^{p_j\times q_j}$ are of known rank for $j = 1, \ldots, r$. 

As the bilinear form of the matrix normal requires its covariance be separable, the multi-linear structure of $\ten{X}$ also induces separability on its covariance structure (see, e.g., \cite{Hoff2011}). Therefore, we further assume that
\begin{align}
\t{(\mat{T}_2\pinv{\mat{D}_p})}\mat{\eta}_{2y}= \t{(\mat{T}_2\pinv{\mat{D}_p})}\mat{\eta}_{2} &= \vec\biggl(c\bigotimes_{j = r}^{1}\mat{\Omega}_j\biggr), \label{eq:eta2}
\end{align}
where $\mat{\Omega}_j\in\mathbb{R}^{p_j\times p_j}$ are symmetric positive definite matrices for $j = 1, \ldots, r$. Requiring $\mat{\Omega} = \bigotimes_{j = r}^{1}\mat{\Omega}_j$ substantially reduces the number of parameters in $\mat{\Omega}$. The requirement the $\mat{\Omega}_j$s be positive definite is possible due to the constant $c$.

Equation \eqref{eq:eta2} is underdetermined since $\t{(\mat{T}_2\pinv{\mat{D}_p})}$ has full column rank $d < p^2$ (with a non-strict inequality if $\ten{X}$ is univariate) but $\mat{\eta}_2$ is uniquely determined given any $\mat{\Omega}$ as $\t{(\pinv{(\mat{T}_2\pinv{\mat{D}_p})})}$ has full row rank. We let $\mat{\xi} = (\vec{\overline{\ten{\eta}}}, \vec{\mat{B}}, \vech{\mat{\Omega}})$ be the unconstrained $p(p + 2 q + 3) / 2$-parameter vector and $\mat{\theta} = (\vec{\overline{\ten{\eta}}}, \vec{\mat{B}}, \vech{\mat{\Omega}})$ be the constrained parameter vector, where $\mat{B}=\bigotimes_{j = r}^{1}\mat{\beta}_j$ and $\mat{\Omega} = \bigotimes_{j = r}^{1}\mat{\Omega}_j$. We also let $\Xi$ and $\Theta$ denote the unconstrained and constrained parameter spaces, with $\mat{\xi}$ and $\mat{\theta}$ varying in $\Xi$ and $\Theta$, respectively. The parameter space $\Xi$ is an open subset of $\mathbb{R}^{p(p + 2 q + 3) / 2}$ so that \eqref{eq:quadratic-exp-fam} is a proper density. We relax the full rank and positive definite assumptions for $\mat{\beta}_k$ and $\mat{\Omega}_k$ later, as a consequence of \cref{thm:param-manifold} in \cref{sec:kron-manifolds}.

In a classical \emph{generalized linear model} (GLM), the link function connecting the natural parameters to the expectation of the sufficient statistic $\mat{\eta}_y = \mat{g}(\E[\mat{t}(\ten{X}) \mid Y = y])$ is invertible. Such a link may not exist in our setting, but for our purpose what we call the ``inverse'' link suffices. As in the non-minimal formulation \eqref{eq:quadratic-exp-fam}, we define the ``inverse'' link through its tensor-valued components
\begin{align}
    \ten{g}_1(\mat{\eta}_y) &= \E[\ten{X} \mid Y = y], \label{eq:inv-link1}\\
    \ten{g}_2(\mat{\eta}_y) &= \E[\ten{X}\circ\ten{X} \mid Y = y] \label{eq:inv-link2}
\end{align}
as $\widetilde{\mat{g}}(\mat{\eta}_y) = (\vec\ten{g}_1(\mat{\eta}_y), \vec\ten{g}_2(\mat{\eta}_y))$.
Under the quadratic exponential family model \eqref{eq:quadratic-exp-fam}, a sufficient reduction for the regression of $Y$ on $\ten{X}$ is given in \cref{thm:sdr}.
\begin{theorem}[SDR]\label{thm:sdr}
    A sufficient reduction for the regression $Y\mid \ten{X}$ under the quadratic exponential family inverse regression model \eqref{eq:quadratic-exp-fam} with natural parameters \eqref{eq:eta1} and \eqref{eq:eta2} is given by
    \begin{equation}\label{eq:sdr}
        \ten{R}(\ten{X})
            = (\ten{X} - \E\ten{X})\mlm_{k = 1}^{r}\t{\mat{\beta}_j}.
    \end{equation}
    The reduction \eqref{eq:sdr} is minimal if $\mat{\beta}_j$ are full rank for all $j=1,\ldots,r$.
\end{theorem}

The reduction \eqref{eq:sdr} in vectorized form is $\vec\ten{R}(\ten{X}) = \t{\mat{B}}\vec(\ten{X} - \E\ten{X})$, where $\mat{B} = \bigotimes_{k = r}^{1}\mat{\beta}_k$ with $\Span(\mat{B}) = \Span(\{\mat{\eta}_{1y} - \E_{Y}\mat{\eta}_{1Y} : y\in\mathcal{S}_Y\})$, where $\mathcal{S}_Y$ denotes the sample space of the random variable $Y$. \cref{thm:sdr} obtains that the \emph{sufficient reduction} $\ten{R}(\ten{X})$ reduces $\ten{X}$ along each mode (dimension) linearly. The graph in \cref{fig:SDRvisual} is a visual representation of the sufficient reduction for a 3-dimensional tensor-valued predictor. 

\begin{figure}
    \centering
    \includegraphics{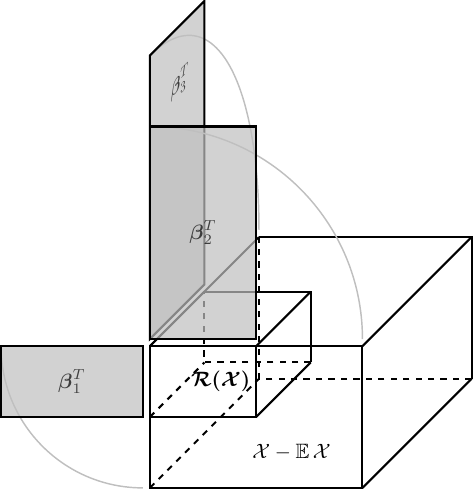}
    \caption{\label{fig:SDRvisual}Visual depiction of the sufficient reduction in \cref{thm:sdr}.}
\end{figure}

\begin{example}[Vector valued $\mat{x}$ ($r = 1$)]\label{ex:vector-dist}
    Given a vector-valued predictor $\mat{X}\in\mathbb{R}^p$, the tensor order is $r = 1$, and the collection of parameters is $\mat{\theta} = (\overline{\mat{\eta}}, \mat{\beta}, \mat{\Omega})$ with $\overline{\mat{\eta}}\in\mathbb{R}^p$, $\mat{\beta}\in\StiefelNonCompact{p}{q}$ and $\mat{\Omega}\in\SymPosDefMat{p}$ where $\mat{f}_y\in\mathbb{R}^q$ are known functions of the response $Y$. The conditional density of $\mat{X}\mid Y = y$ is given by
    \begin{align*}
        f_{\theta}(\mat{x}\mid Y = y)
            &= h(\mat{x})\exp(\langle\mat{x}, \mat{\eta}_{1y}(\mat{\theta})\rangle + \langle\vec(\mat{x}\circ\mat{x}), \mat{\eta}_2(\mat{\theta})\rangle - b(\mat{\eta}_y(\mat{\theta}))) \\
            &= h(\mat{x})\exp(\t{(\overline{\mat{\eta}} + \mat{\beta}\mat{f}_y)}\mat{x} + c\,\t{\mat{x}}\mat{\Omega}\,\mat{x} - b(\mat{\eta}_y(\mat{\theta}))).
    \end{align*}
    using $\mat{\eta}_{1y}(\mat{\theta}) = \overline{\mat{\eta}} + \mat{\beta}\mat{f}_y$ and $\mat{\eta}_2(\mat{\theta}) = c\,\mat{\Omega}$.
\end{example}

\begin{example}[Matrix-valued $\mat{X}$ ($r = 2$)]
    Assuming $\mat{X}$ is matrix-valued, $\mat{\theta} = (\overline{\mat{\eta}}, \mat{\beta}_1, \mat{\beta}_2, \mat{\Omega}_1, \mat{\Omega}_2)$, where the intercept term $\overline{\mat{\eta}}\in\mathbb{R}^{p_1\times p_2}$ is now matrix-valued. Now $\mat{F}_Y\in\mathbb{R}^{q_1\times q_2}$ is matrix-valued, and the conditional density of $\mat{X}\mid Y = y$ is
    \begin{align*}
        f_{\mat{\theta}}(\mat{X}\mid Y = y)
            &= h(\mat{X})\exp(\langle\vec{\mat{X}}, \mat{\eta}_{1y}(\mat{\theta})\rangle + \langle\vec(\mat{X}\circ\mat{X}), \mat{\eta}_2(\mat{\theta})\rangle - b(\mat{\eta}_y(\mat{\theta}))) \\
            &= h(\mat{X})\exp(\tr((\overline{\mat{\eta}} + \mat{\beta}_1\mat{F}_y\t{\mat{\beta}_2})\t{\mat{X}}) + c \tr(\mat{\Omega}_1\mat{X}\mat{\Omega}_2\t{\mat{X}}) - b(\mat{\eta}_y(\mat{\theta}))).
    \end{align*}
\end{example}

\section{Maximum Likelihood Estimation}\label{sec:ml-estimation}

Suppose $(\ten{X}_i, Y_i)$ are independently and identically distributed with joint cdf $F(\ten{X}, Y)$, for $i = 1, \ldots, n$. The empirical log-likelihood function of \eqref{eq:quadratic-exp-fam} under \eqref{eq:eta1} and \eqref{eq:eta2}, ignoring terms not depending on the parameters, is
\begin{equation}\label{eq:log-likelihood}
    l_n(\mat{\theta}) = \frac{1}{n}\sum_{i = 1}^n \biggl(\Bigl\langle\overline{\ten{\eta}} + \ten{F}_{y_i}\mlm_{k = 1}^{r}\mat{\beta}_k, \ten{X}_i \Bigr\rangle + c\Bigl\langle\ten{X}_i\mlm_{k = 1}^{r}\mat{\Omega}_k, \ten{X}_i \Bigr\rangle - b(\mat{\eta}_{y_i})\biggr).
\end{equation}
The maximum likelihood estimate of $\mat{\theta}_0$ is the solution to the optimization problem
\begin{equation}\label{eq:mle}
    \hat{\mat{\theta}}_n = \argmax_{\mat{\theta}\in\Theta}l_n(\mat{\theta})
\end{equation}
with $\hat{\mat{\theta}}_n = (\vec\widehat{\overline{\ten{\eta}}}, \vec\widehat{\mat{B}}, \vech\widetilde{\mat{\Omega}})$ where $\widehat{\mat{B}} = \bigkron_{k = r}^{1}\widehat{\mat{\beta}}_k$ and $\widehat{\mat{\Omega}} = \bigkron_{k = r}^{1}\widehat{\mat{\Omega}}_k$.

A straightforward general method for parameter estimation is \emph{gradient descent}. To apply gradient-based optimization, we compute the gradients of $l_n$ in \cref{thm:grad}.

\begin{theorem}\label{thm:grad}
    Suppose $(\ten{X}_i, y_i), i = 1, ..., n$, are i.i.d. with conditional log-likelihood of the form \eqref{eq:log-likelihood}, where $\mat{\theta}$ denotes the collection of all GMLM parameters $\overline{\ten{\eta}}$, ${\mat{B}} = \bigkron_{k = r}^{1}{\mat{\beta}}_k$ and ${\mat{\Omega}} = \bigkron_{k = r}^{1}{\mat{\Omega}}_k$ for $k = 1, ..., r$. Let $\ten{G}_2(\mat{\eta}_y)$ be a tensor of dimension $p_1, \ldots, p_r$ such that
    \begin{displaymath}
        \vec{\ten{G}_2(\mat{\eta}_y)} = \pinv{(\mat{T}_2\pinv{\mat{D}_p})}\mat{T}_2\pinv{\mat{D}_p}\vec{\ten{g}_2(\mat{\eta}_y)}.
    \end{displaymath}
    Then, the partial gradients with respect to $\overline{\ten{\eta}}, \mat{\beta}_1, \ldots, \mat{\beta}_r, \mat{\Omega}_1, \ldots, \mat{\Omega}_r$ are given by
    \begin{align*}
        \nabla_{\overline{\ten{\eta}}}l_n &= \vec\frac{1}{n}\sum_{i = 1}^n (\ten{X}_i - \ten{g}_1(\mat{\eta}_{y_i})), \\
        \nabla_{\mat{\beta}_j}l_n &= \vec\frac{1}{n}\sum_{i = 1}^n (\ten{X}_i - \ten{g}_1(\mat{\eta}_{y_i}))_{(j)}\t{\Big(\ten{F}_{y_i}\mlm_{k\in[r]\backslash j}\mat{\beta}_k\Big)_{(j)}}, \\
        \nabla_{\mat{\Omega}_j}l_n &= \vec\frac{c}{n}\sum_{i = 1}^n (\ten{X}_i\otimes\ten{X}_i - \K(\ten{G}_2(\mat{\eta}_{y_i})))\mlm_{k\in[r]\backslash j}\t{(\vec{\mat{\Omega}_k})},
    \end{align*}
    so that $\nabla l_n = (\nabla_{\overline{\ten{\eta}}}l_n, \nabla_{\mat{\beta}_1}l_n, \ldots, \nabla_{\mat{\beta}_r}l_n, \nabla_{\mat{\Omega}_1}l_n, \ldots, \nabla_{\mat{\Omega}_r}l_n)$. If $\mat{T}_2$ is the identity matrix $\mat{I}_{p(p + 1) / 2}$, then $\ten{G}_2(\mat{\eta}_y) = \ten{g}_2(\mat{\eta}_y)$.
\end{theorem}

Although the general case of any GMLM model can be fitted via gradient descent using \cref{thm:grad}, this may be very inefficient. In \cref{thm:grad}, $\mat{T}_2$ can be used to introduce flexible second moment structures. For example, it allows modeling effects differently for predictor components, as described in \cref{sec:ising_estimation} after Eqn. \eqref{eq:ising-cond-prob}. In the remainder, we focus on $\mat{T}_2$ being the identity matrix, which simplifies the estimation algorithm and the speed of the numerical calculation in the case of multi-linear normal predictors. An iterative cyclic updating scheme that converges fast, is stable, and does not require hyperparameters is derived in \cref{sec:tensor-normal-estimation}, as discussed later. On the other hand, the Ising model does not allow such a scheme. There we need to use a gradient-based method, which is the subject of \cref{sec:ising_estimation}.

\subsection{Multi-Linear Normal}\label{sec:tensor-normal-estimation}

The \emph{multi-linear normal} is the extension of the matrix normal to tensor-valued random variables and a member of the quadratic exponential family \eqref{eq:quadratic-exp-fam} under \eqref{eq:eta2}. \cite{Dawid1981} and \cite{Arnold1981} introduced the term matrix normal and, in particular, \cite{Arnold1981} provided several theoretical results, such as its density, moments and conditional distributions of its components. The matrix normal distribution is a bilinear normal distribution; a distribution of a two-way (two-component) array, each component representing a vector of observations \cite[]{OhlsonEtAl2013}. \cite{KolloVonRosen2005,Hoff2011,OhlsonEtAl2013} presented the extension of the bilinear to the multi-linear normal distribution, using a parallel extension of bilinear matrices to multi-linear tensors \cite[]{Comon2009}.

The defining feature of the matrix normal distribution, and its multi-linear extension, is the Kronecker product structure of its covariance. This formulation, where the vectorized variables are normal with multiway covariance structure modeled as a Kronecker product of matrices of much lower dimension, aims to overcome the significant modeling and computational challenges arising from the high computational complexity of manipulating tensor representations \cite[see, e.g.,][]{HillarLim2013,WangEtAl2022}.

Suppose the conditional distribution of tensor $\ten{X}$ given $Y$ is multi-linear normal with mean $\ten{\mu}_y$ and covariance $\mat{\Sigma} = \bigkron_{k = r}^{1}\mat{\Sigma}_k$. We assume the distribution is non-degenerate which means that the covariances $\mat{\Sigma}_k$ are symmetric positive definite matrices. Its density is given by
\begin{displaymath}
    f_{\mat{\theta}}(\ten{X}\mid Y = y) = (2\pi)^{-p / 2}\prod_{k = 1}^{r}\det(\mat{\Sigma}_k)^{-p / 2 p_k}\exp\left( -\frac{1}{2}\left\langle\ten{X} - \ten{\mu}_y, (\ten{X} - \ten{\mu}_y)\mlm_{k = 1}^{r}\mat{\Sigma}_k^{-1} \right\rangle \right).
\end{displaymath}
For the sake of simplicity and w.l.o.g., we assume $\ten{X}$ has 0 marginal expectation; i.e., $\E\ten{X} = 0$. Rewriting this in the quadratic exponential family form \eqref{eq:quadratic-exp-fam} determines the scaling constant $c = -1/2$. The relation to the GMLM parameters $\overline{\ten{\eta}}, \mat{\beta}_k$ and $\mat{\Omega}_k$, for $k = 1, \ldots, r$ is
\begin{equation}\label{eq:tnormal_cond_params}
    \ten{\mu}_y = \ten{F}_y\mlm_{k = 1}^{r}\mat{\Omega}_k^{-1}\mat{\beta}_k, \qquad
    \mat{\Omega}_k = \mat{\Sigma}_k^{-1},
\end{equation}
where we used that $\overline{\ten{\eta}} = 0$ due to $0 = \E\ten{X} = \E\E[\ten{X}\mid Y] = \E\ten{\mu}_Y$ in combination with $\E\ten{F}_Y = 0$. Additionally, the positive definiteness of the $\mat{\Sigma}_k$'s renders all the $\mat{\Omega}_k$'s symmetric positive definite. This leads to another simplification since then $\mat{T}_2$ in \eqref{eq:t-stat} equals the identity, which also means that the gradients of the log-likelihood $l_n$ in \cref{thm:grad} are simpler to compute. We obtain $\ten{g}_1(\mat{\eta}_y) = \E[\ten{X}\mid Y = y] = \ten{\mu}_y$, and 
\begin{align*}
    \ten{G}_2(\mat{\eta}_y) = \ten{g}_2(\mat{\eta}_y) &= \E[\ten{X}\circ\ten{X}\mid Y = y] \equiv \bigkron_{k = r}^1\mat{\Sigma}_k + (\vec{\ten{\mu}}_y)\t{(\vec{\ten{\mu}}_y)}.
\end{align*}
In practice, we assume we have a random sample of $n$ observations $(\ten{X}_i, \ten{F}_{y_i})$ from the joint distribution. We start the estimation process by demeaning the data. Then, only the reduction matrices $\mat{\beta}_k$ and the scatter matrices $\mat{\Omega}_k$ need to be estimated. To solve the optimization problem \eqref{eq:mle}, with $\overline{\ten{\eta}} = 0$ we initialize the parameters using a simple heuristic approach. First, we compute moment based mode-wise marginal covariance estimates $\widehat{\mat{\Sigma}}_k(\ten{X})= \sum_{i = 1}^{n} (\ten{X}_i)_{(k)}\t{(\ten{X}_i)_{(k)}}/n$ and $\widehat{\mat{\Sigma}}_k(\ten{F}_Y)= \sum_{i = 1}^{n} (\ten{F}_{y_i})_{(k)}\t{(\ten{F}_{y_i})_{(k)}}/n$. Then, for every mode $k = 1, \ldots, r$, we compute the eigenvectors $\mat{v}_j(\widehat{\mat{\Sigma}}_k(\ten{X}))$, $\mat{v}_j(\widehat{\mat{\Sigma}}_k(\ten{F}_Y))$ that correspond to the leading eigenvalues $\lambda_j(\widehat{\mat{\Sigma}}_k(\ten{X}))$, $\lambda_j(\widehat{\mat{\Sigma}}_k(\ten{X}))$ of the marginal covariance estimates, for $j = 1, \ldots, q_k$. We set 
\begin{align*}
    \mat{U}_k &= (\mat{v}_1(\widehat{\mat{\Sigma}}_1(\ten{X})), \ldots, \mat{v}_{q_k}(\widehat{\mat{\Sigma}}_{q_k}(\ten{X}))), \, \mat{V}_k = (\mat{v}_1(\widehat{\mat{\Sigma}}_1(\ten{F}_Y), \ldots, \mat{v}_{q_k}(\widehat{\mat{\Sigma}}_{q_k}(\ten{F}_Y)) \\
    \mat{D}_k &= \diag(\mat{v}_1(\widehat{\mat{\Sigma}}_1(\ten{X}))\mat{v}_1(\widehat{\mat{\Sigma}}_1(\ten{F}_{Y})), \ldots, \mat{v}_{q_k}(\widehat{\mat{\Sigma}}_{q_k}(\ten{X}))\mat{v}_{q_k}(\widehat{\mat{\Sigma}}_k(\ten{F}_{Y}))).
\end{align*}
The initial value of $\mat{\beta}_k$ is $\hat{\mat{\beta}}_k^{(0)} = \mat{U}_k\sqrt{\mat{D}_k}\t{\mat{V}_k},$ and the initial value of $\mat{\Omega}_k$ is set to the identity $\mat{\Omega}_k^{(0)} = \mat{I}_{p_k}$, for $k=1,\ldots,r$.
Given $\hat{\mat{\beta}}_1, \ldots, \hat{\mat{\beta}}_r, \hat{\mat{\Omega}}_1, \ldots, \hat{\mat{\Omega}}_r$, we take the gradient $\nabla_{\mat{\beta}_j}l_n$ of the tensor normal log-likelihood $l_n$ in \eqref{eq:log-likelihood} applying \cref{thm:grad} and keep all other parameters except $\mat{\beta}_j$ fixed. Then, $\nabla_{\mat{\beta}_j}l_n = 0$ has the closed form solution
\begin{equation}\label{eq:tensor_normal_beta_solution}
    \t{\mat{\beta}_j} = \biggl(
        \sum_{i = 1}^{n}
        \Bigl( \ten{F}_{y_i}\mlm_{k \neq j}\hat{\mat{\Omega}}_k^{-1}\hat{\mat{\beta}}_k \Bigr)_{(j)}
        \t{\Bigl( \ten{F}_{y_i}\mlm_{k \neq j}\hat{\mat{\beta}}_k \Bigr)_{(j)}}
    \biggr)^{-1}
    \biggl(
        \sum_{i = 1}^{n}
        \Bigl( \ten{F}_{y_i}\mlm_{k \neq j}\hat{\mat{\beta}}_k \Bigr)_{(j)}
        \t{(\ten{X}_{i})_{(j)}}
    \biggr)
        \hat{\mat{\Omega}}_j.
\end{equation}
Equating the partial gradient of the $j$th scatter matrix $\mat{\Omega}_j$ in \cref{thm:grad} to zero ( $\nabla_{\mat{\Omega}_j}l_n = 0$) gives a quadratic matrix equation due to the dependence of $\ten{\mu}_y$ on $\mat{\Omega}_j$. In practice though, it is faster, more stable, and equally accurate to use mode-wise covariance estimates via the residuals
\begin{displaymath}
    \hat{\ten{R}}_i = \ten{X}_i - \hat{\ten{\mu}}_{y_i} = \ten{X}_i - \ten{F}_{y_i}\mlm_{k = 1}^{r}\hat{\mat{\Omega}}_k^{-1}\hat{\mat{\beta}}_k.
\end{displaymath} 
The estimates are computed via $\tilde{\mat{\Sigma}}_j = \sum_{i = 1}^n (\hat{\ten{R}}_i)_{(j)} \t{(\hat{\ten{R}}_i)_{(j)}},$ where $\tilde{s}\tilde{\mat{\Sigma}}_j = \hat{\mat{\Omega}}_j^{-1}$. To decide on the scaling factor $\tilde{s}$ we use that the mean squared error has to be equal to the trace of the covariance estimate,
\begin{displaymath}
    \frac{1}{n}\sum_{i = 1}^n \langle \hat{\ten{R}}_i, \hat{\ten{R}}_i \rangle = \tr\bigkron_{k = r}^{1}\hat{\mat{\Omega}}_k^{-1} = \prod_{k = 1}^{r}\tr{\hat{\mat{\Omega}}_k^{-1}} = \tilde{s}^r\prod_{k = 1}^{r}\tr{\tilde{\mat{\Sigma}}_k},
\end{displaymath}
so that
\begin{displaymath}
    \tilde{s} = \biggl(\Bigl(\prod_{k = 1}^{r}\tr{\tilde{\mat{\Sigma}}_k}\Bigr)^{-1}\frac{1}{n}\sum_{i = 1}^n \langle \hat{\ten{R}}_i, \hat{\ten{R}}_i \rangle\biggr)^{1 / r}
\end{displaymath}
resulting in the estimates $\hat{\mat{\Omega}}_j = (\tilde{s}\tilde{\mat{\Sigma}}_j)^{-1}$. Estimation is performed by updating the estimates $\hat{\mat{\beta}}_j$ via \eqref{eq:tensor_normal_beta_solution} for $j = 1, \ldots, r$, and then recompute the $\hat{\mat{\Omega}}_j$ estimates simultaneously keeping the $\hat{\mat{\beta}}_j$'s fixed. This procedure is repeated until convergence.

A technical detail for numerical stability is to ensure that the scaled values $\tilde{s}\tilde{\mat{\Sigma}}_j$, assumed to be symmetric and positive definite, are well conditioned. Thus, we estimate the condition number of $\tilde{s}\tilde{\mat{\Sigma}}_j$ before computing the inverse. In case of ill-conditioning, we use the regularized $\hat{\mat{\Omega}}_j = (\tilde{s}\tilde{\mat{\Sigma}}_j + 0.2 \lambda_{1}(\tilde{s}\tilde{\mat{\Sigma}}_j)\mat{I}_{p_j})^{-1}$ instead, where $\lambda_{1}(\tilde{s}\tilde{\mat{\Sigma}}_j)$ is the first (maximum) eigenvalue. Experiments showed that this regularization is usually only required in the first few iterations.

Furthermore, if the parameter space follows a more general setting as in \cref{thm:param-manifold}, updating may produce estimates outside the parameter space. A simple and efficient method is to project every updated estimate onto the corresponding manifold.

A common algorithm to calculate the MLE of a Kronecker product is block-coordinate descent, proposed independently by \cite{MardiaGoodall1993} and \cite{Dutilleul1999}. It was later called ``flip-flop'' algorithm by \cite{LuZimmerman2005} for the computation of the maximum likelihood estimators of the components of a separable covariance matrix. \cite{ManceurDutilleul2013} extended the ``flip-flop'' algorithm for the computation of the MLE of the separable covariance structure of a 3-way and 4-way normal distribution and obtained a lower bound for the sample size required for its existence. The same issue was also studied by \cite{DrtonEtAl2020} in the case of a two-way array (matrix). Our algorithm uses a
similar ``flip-flop'' approach by iteratively updating the $\mat{\beta}_k$'s and $\mat{\Omega}_k$'s.

\subsection{Multi-Linear Ising Model}\label{sec:ising_estimation}

The Ising\footnote{Also known as the \emph{Lenz-Ising} model as the physical assumptions of the model where developed by both Lenz and Ising \cite[]{Niss2005} where Ising gave a closed form solution for the 1D lattice \cite[]{Ising1925}.} model \cite[]{Lenz1920,Ising1925,Niss2005} is a mathematical model originating in statistical physics to study ferromagnetism in a thermodynamic setting. It describes magnetic dipoles (atomic ``spins'' with values $\pm 1$) under an external magnetic field (first moments) while allowing two-way interactions (second moments) between direct neighbors on a lattice, a discrete grid. The Ising model is a member of the discrete quadratic exponential family \cite[]{CoxWermuth1994,JohnsonEtAl1997} for multivariate binary outcomes where the interaction structure (non-zero correlations) is determined by the lattice. The $p$-dimensional Ising model is a discrete probability distribution on the set of $p$-dimensional binary vectors $\mat{x}\in\{0, 1\}^p$ with probability mass function (pmf) given by
\begin{displaymath}
    P_{\mat{\gamma}}(\mat{x}) = p_0(\mat{\gamma})\exp(\t{\vech(\mat{x}\t{\mat{x}})}\mat{\gamma}).
\end{displaymath}
The scaling factor $p_0(\mat{\gamma})\in\mathbb{R}_{+}$ ensures that $P_{\mat{\gamma}}$ is a pmf. It is equal to the probability of the zero event $P(X = \mat{0}) = p_0(\mat{\gamma})$. More commonly known as the \emph{partition function}, the reciprocal of $p_0$, is given by
\begin{equation}\label{eq:ising-partition-function}
    p_0(\mat{\gamma})^{-1} = \sum_{\mat{x}\in\{0, 1\}^p}\exp(\t{\vech(\mat{x}\t{\mat{x}})}\mat{\gamma}).
\end{equation}
Abusing notation, we let $\mat{\gamma}_{j l}$ denote the element of $\mat{\gamma}$ corresponding to $\mat{x}_j\mat{x}_l$ in $\vech(\mat{x}\t{\mat{x}})$.\footnote{Specifically, the element $\mat{\gamma}_{j l}$ of $\mat{\gamma}$ is a short hand for $\mat{\gamma}_{\iota(j, l)}$ with $\iota(j, l) = (\min(j, l) - 1)(2 p - \min(j, l)) / 2 + \max(j, l)$ mapping the matrix row index $j$ and column index $l$ to the corresponding half vectorization indices $\iota(j, l)$.} The ``diagonal'' parameter $\mat{\gamma}_{j j}$ expresses the conditional log odds of $X_j = 1\mid X_{-j} = \mat{0}$, where the negative subscript in $X_{-j}$ describes the $p - 1$ dimensional vector $X$ with the $j$th element removed. The off diagonal entries $\mat{\gamma}_{j l}$, $j\neq l$, are equal to the conditional log odds of simultaneous occurrence $X_j = 1, X_l = 1 \mid X_{-j, -l} = \mat{0}$. More precisely, the conditional probabilities $\pi_j(\mat{\gamma}) = P_{\mat{\gamma}}(X_j = 1\mid X_{-j} = \mat{0})$ and $\pi_{j, l}(\mat{\gamma}) = P_{\mat{\gamma}}(X_j = 1, X_l = 1\mid X_{-j, -l} = \mat{0})$ are related to the natural parameters via
\begin{equation}\label{eq:ising-two-way-log-odds}
    \mat{\gamma}_{j j} = \log\frac{\pi_j(\mat{\gamma})}{1 - \pi_j(\mat{\gamma})}, \qquad
    \mat{\gamma}_{j l} = \log\frac{1 - \pi_j(\mat{\gamma})\pi_l(\mat{\gamma})}{\pi_j(\mat{\gamma})\pi_l(\mat{\gamma})}\frac{\pi_{j l}(\mat{\gamma})}{1 - \pi_{j l}(\mat{\gamma})}.
\end{equation}

Conditional Ising models, incorporating the information of covariates $Y$ into the model, were considered by \cite{ChengEtAl2014,BuraEtAl2022}. The direct way is to parameterize $\mat{\gamma} = \mat{\gamma}_y$ by the covariate $Y = y$ to model a conditional distribution $P_{\mat{\gamma}_y}(\mat{x}\mid Y = y)$.

We extend the conditional pmf by allowing the binary variables to be tensor-valued; that is, we set $\mat{x} = \vec{\ten{X}}$, with dimension $p = \prod_{k = 1}^{r}p_k$ for $\ten{X}\in\{ 0, 1 \}^{p_1\times\cdots\times p_r}$. The tensor structure of $\ten{X}$ is accommodated by assuming Kronecker product constraints to the parameter vector $\mat{\gamma}_y$ in a similar fashion as in the multi-linear normal model. This means that we compare the pmf $P_{\mat{\gamma}_y}(\vec{\ten{X}} | Y = y)$ with the quadratic exponential family \eqref{eq:quadratic-exp-fam} with the natural parameters modeled by \eqref{eq:eta1} and \eqref{eq:eta2}. The diagonal of $(\vec{\ten{X}})\t{(\vec{\ten{X}})}$ is equal to $\vec{\ten{X}}$, which results in the GMLM being expressed as
\begin{align}
    P_{\mat{\gamma}_y}(\ten{X} \mid Y = y)
        &= p_0(\mat{\gamma}_y)\exp(\t{\vech((\vec{\ten{X}})\t{(\vec{\ten{X}})})}\mat{\gamma}_y) \label{eq:ising-cond-prob} \\
        &= p_0(\mat{\gamma}_y)\exp\Bigl(\Bigl\langle \ten{X}, \ten{F}_y\mlm_{k = 1}^{r}\mat{\beta}_k \Bigr\rangle + \Bigl\langle\ten{X}\mlm_{k = 1}^{r}\mat{\Omega}_k, \ten{X}\Bigr\rangle\Bigr) \nonumber
\end{align}
where we set $\overline{\ten{\eta}} = 0$ and $\mat{T}_2$ to the identity. This imposes an additional constraint on the model, the reason is that the diagonal elements of $\mat{\Omega} = \bigkron_{k = r}^{1}\mat{\Omega}_k$ take the role of $\overline{\ten{\eta}}$, although not fully. Having the diagonal of $\mat{\Omega}$ and $\overline{\ten{\eta}}$ accounting for the self-interaction effects might lead to interference in the optimization routine. Another approach would be to use the $\mat{T}_2$ matrix to set the corresponding diagonal elements of $\mat{\Omega}$ to zero and let $\overline{\ten{\eta}}$ handle the self-interaction effect. All of these approaches, namely setting $\overline{\ten{\eta}} = 0$, keeping $\overline{\ten{\eta}}$ or using $\mat{T}_2$, are theoretically solid and compatible with \cref{thm:grad,thm:param-manifold,thm:asymptotic-normality-gmlm}, assuming all axis dimensions $p_k$ are non-degenerate, that is $p_k > 1$ for all $k = 1, \ldots, r$. Regardless, under our modeling choice, the relation between the natural parameters $\mat{\gamma}_y$ of the conditional Ising model and the GMLM parameters $\mat{\beta}_k$ and $\mat{\Omega}_k$ is
\begin{equation}\label{eq:ising-natural-params}
    \mat{\gamma}_y
        = \t{\mat{D}_p}\vec(\mat{\Omega} + \diag(\mat{B}\vec{\ten{F}_y}))
        = \t{\mat{D}_p}\vec\Biggl(\bigkron_{k = r}^{1}\mat{\Omega}_k + \diag\biggl(\vec\Bigl(\ten{F}_y\mlm_{k = 1}^{r}\mat{\beta}_k\Bigr)\biggr)\Biggr).
\end{equation}
In contrast to the multi-linear normal GMLM, the matrices $\mat{\Omega}_k$ are only required to be symmetric. More specifically, we require $\mat{\Omega}_k$, for $k = 1, \ldots, r$, to be elements of an embedded submanifold of $\SymMat{p_k}$ (see \cref{sec:kron-manifolds,sec:matrix-manifolds}). The mode-wise reduction matrices $\mat{\beta}_k$ are elements of an embedded submanifold of $\mathbb{R}^{p_k\times q_k}$. Common choices are listed in \cref{sec:matrix-manifolds}.

To solve the optimization problem \eqref{eq:mle}, given a data set $(\ten{X}_i, y_i)$, $i = 1, \ldots, n$, we use a variation of gradient descent.

\subsubsection{Initial Values}

The first step is to get reasonable starting values. Experiments showed that a good starting value of $\mat{\beta}_k$ is to use the multi-linear normal estimates from \cref{sec:tensor-normal-estimation} for $k = 1, \ldots, r$, considering $\ten{X}_i$ as continuous. For initial values of $\mat{\Omega}_k$, a different approach is required. Setting everything to the uninformative initial value results in $\mat{\Omega}_k = \mat{0}$ as this corresponds to the conditional log odds to be $1:1$ for every component and pairwise interaction. This is not possible, since $\mat{0}$ is a stationary point of the log-likelihood, as can be directly observed by considering the partial gradients of the log-likelihood in \cref{thm:grad}. Instead, we use a crude heuristic that threads every mode separately and ignores any relation to the covariates. It is computationally cheap and better than any of the alternatives we considered. For every $k = 1, \ldots, r$, let the $k$th mode second moment estimate be
\begin{equation}\label{eq:ising-mode-moments}
    \hat{\mat{M}}_{2(k)} = \frac{p_k}{n p}\sum_{i = 1}^n (\ten{X}_i)_{(k)}\t{(\ten{X}_i)_{(k)}}
\end{equation}
which contains the $k$th mode first moment estimate in its diagonal $\hat{\mat{M}}_{1(k)} = \diag\hat{\mat{M}}_{2(k)}$. Considering every column of the matricized observation $(\ten{X}_i)_{(k)}$ as a $p_k$ dimensional observation. The number of those artificially generated observations is $n \prod_{j\neq k}p_j$. Let $Z_k$ denote the random variable those artificial observations are realization of. Then, we can interpret the elements $(\hat{\mat{M}}_{1(k)})_{j}$ as the estimates of the marginal probability of the $j$th element of $Z_k$ being $1$, $P((Z_k)_j = 1)$. Similarly, for $l \neq j$, $(\hat{\mat{M}}_{2(k)})_{j l}$ estimates the marginal probability of two-way interactions, $P((Z_k)_j = 1, (Z_k)_l = 1)$. Now, we set the diagonal elements of $\mat{\Omega}_k$ to zero. For the off diagonal elements of $\mat{\Omega}_k$, we equate the conditional probabilities $P((Z_k)_j = 1 \mid (Z_k)_{-j} = \mat{0})$ and $P((Z_k)_j = 1, (Z_k)_l = 1\mid (Z_k)_{-j, -l} = \mat{0})$ with the marginal probability estimates $(\hat{\mat{M}}_{1(k)})_{j}$ and $(\hat{\mat{M}}_{2(k)})_{j l}$, respectively. Applying \eqref{eq:ising-two-way-log-odds} gives the initial component-wise estimates $\hat{\mat{\Omega}}_k^{(0)}$, 
\begin{equation}\label{eq:ising-init-Omegas}
    (\hat{\mat{\Omega}}_k^{(0)})_{j j} = 0,
    \qquad
    (\hat{\mat{\Omega}}_k^{(0)})_{j l} = \log\frac{1 - (\hat{\mat{M}}_{1(k)})_{j}(\hat{\mat{M}}_{1(k)})_{l}}{(\hat{\mat{M}}_{1(k)})_{j}(\hat{\mat{M}}_{1(k)})_{l}}\frac{(\hat{\mat{M}}_{2(k)})_{j l}}{1 - (\hat{\mat{M}}_{2(k)})_{j l}}, \, j \neq l.
\end{equation}

\subsubsection{Gradient Optimization}

Given initial values, the gradients derived in \cref{thm:grad} can be evaluated for the Ising model. The first step therefore is to determine the values of the inverse link components $\ten{g}_1(\mat{\gamma}_y) = \E[\ten{X} \mid Y = y]$ and $\ten{G}_2(\mat{\gamma}_y) = \ten{g}_2(\mat{\gamma}_y) = \E[\ten{X}\circ\ten{X} \mid Y = y]$. An immediate simplification is that the first moment is part of the second moment. Its values are determined via $\vec(\E[\ten{X} \mid Y = y]) = \diag(\E[\ten{X}\circ\ten{X} \mid Y = y]_{(1, \ldots, r)})$; i.e., only the second moment needs to be computed, or estimated (see \cref{sec:ising-bigger-dim}) in the case of slightly bigger $p$. For the Ising model, the conditional second moment with parameters $\mat{\gamma}_y$ is given by the matricized relation
\begin{equation}\label{eq:ising-m2}
    \ten{g}_2(\ten{\gamma}_y)_{(1, \ldots, r)} = \E\left[(\vec{\ten{X}})\t{(\vec{\ten{X}})}\mid Y = y\right] = p_0(\mat{\gamma}_y)\sum_{\mat{x}\in\{0, 1\}^{p}}\mat{x}\t{\mat{x}}\exp(\t{\vech(\mat{x}\t{\mat{x}})}\mat{\gamma}_y).
\end{equation}
The natural parameter $\mat{\gamma}_y$ is evaluated via \eqref{eq:ising-natural-params} enabling us to compute the partial gradients of the log-likelihood $l_n$ \eqref{eq:log-likelihood} for the Ising model by \cref{thm:grad} for the GMLM parameters $\mat{\beta}_k$ and $\mat{\Omega}_k$, $k = 1, \ldots, r$, at the current iterate $\mat{\theta}^{(I)} = (\mat{\beta}_1^{(I)}, \ldots, \mat{\beta}_r^{(I)}, \mat{\Omega}_1^{(I)}, \ldots, \mat{\Omega}_r^{(I)})$. Using classic gradient ascent for maximizing the log-likelihood, we have to specify a learning rate $\lambda\in\mathbb{R}_{+}$, usually a value close to $10^{-3}$. The update rule is
\begin{displaymath}
    \mat{\theta}^{(I + 1)} = \mat{\theta}^{(I)} + \lambda\nabla_{\mat{\theta}} l_n(\mat{\theta})\bigr|_{\mat{\theta} = \mat{\theta}^{(I)}},
\end{displaymath}
which is iterated till convergence. In practice, iteration is performed until either a maximum number of iterations is exhausted and/or some break condition is satisfied. A proper choice of the learning rate is needed as a large learning rate $\lambda$ may cause instability, while a very low learning rate requires an enormous amount of iterations. Generically, there are two approaches to avoid the need to determine a proper learning rate. First, \emph{line search methods} determine an appropriate step size for every iteration. This works well if the evaluation of the object function (the log-likelihood) is cheap. This is not the case in our setting, see \cref{sec:ising-bigger-dim}. The second approach is an \emph{adaptive learning rate}, where one tracks specific statistics while optimizing and dynamically adapting the learning rate via well-tested heuristics using the gathered knowledge from past iterations. We opted to use an adaptive learning rate approach, which not only removes the need to determine an appropriate learning rate but also accelerates learning.

Our method of choice is \emph{root mean squared propagation} (RMSprop) \cite[]{Hinton2012}. This is a well-known method in machine learning for training neural networks. It is a variation of gradient descent with a per scalar parameter adaptive learning rate. It tracks a moving average of the element-wise squared gradient $\mat{g}_2^{(I)}$, which is then used to scale (element-wise) the gradient in the update rule (see \cite{Hinton2012} and \cite{GoodfellowEtAl2016} among others). The update rule using RMSprop for maximization\footnote{Instead of the more common minimization, therefore $+$ in the update of $\mat{\theta}$.} is
\begin{align*}
    \mat{g}_2^{(I + 1)} &= \nu \mat{g}_2^{(I)} + (1 - \nu)\nabla l_n(\mat{\theta}^{(I)})\odot\nabla l_n(\mat{\theta}^{(I)}), \\
    \mat{\theta}^{(I + 1)} &= \mat{\theta}^{(I)} + \frac{\lambda}{\sqrt{\mat{g}_2^{(I + 1)}} + \epsilon}\odot\nabla l_n(\mat{\theta}^{(I)}).
\end{align*}
The parameters $\nu = 0.9$, $\lambda = 10^{-3}$ and $\epsilon\approx 1.49\cdot 10^{-8}$ are fixed. The initial value of $\mat{g}_2^{(0)} = \mat{0}$, where the symbol $\odot$ denotes the Hadamard product, or element-wise multiplication. The division and square root operations are performed element-wise as well. According to our experiments, RMSprop requires iterations in the range of $50$ till $1000$ till convergence while gradient ascent with a learning rate of $10^{-3}$ is in the range of $1000$ till $10000$.

\subsubsection{Small Data Sets}\label{sec:ising-small-data-sets}

In the case of a finite number of observations, specifically in data sets with a small number of observations $n$, the situation where one component is always either zero or one can occur. It is also possible to observe two exclusive components. In practice, this situation of a ``degenerate'' data set should be protected against. Working with parameters on a log scale, gives estimates of $\pm\infty$, which is outside the parameter space and breaks our optimization algorithm.

The first situation where this needs to be addressed is in \eqref{eq:ising-init-Omegas}, where we set initial estimates for $\mat{\Omega}_k$. To avoid division by zero as well as evaluating the log of zero, we adapt \eqref{eq:ising-mode-moments}, the mode-wise moment estimates $\hat{\mat{M}}_{2(k)}$. A simple method is to replace the ``degenerate'' components, that are entries with value $0$ or $1$, with the smallest positive estimate of exactly one occurrence $p_k / n p$, or all but one occurrence $1 - p_k / n p$, respectively.

The same problem is present in gradient optimization. Therefore, before starting the optimization, we detect degenerate combinations. We compute upper and lower bounds for the ``degenerate'' element in the Kronecker product $\hat{\mat{\Omega}} = \bigkron_{k = r}^{1}\hat{\mat{\Omega}}_k$. After every gradient update, we check if any of the ``degenerate'' elements fall outside of the bounds. In that case, we adjust all the elements of the Kronecker component estimates $\hat{\mat{\Omega}}_k$, corresponding to the ``degenerate'' element of their Kronecker product, to fall inside the precomputed bounds. While doing so, we try to alter every component as little as possible to ensure that the non-degenerate elements in $\hat{\mat{\Omega}}$, affected by this change due to its Kronecker structure, are altered as little as possible. The exact details are technically cumbersome while providing little insight.

\subsubsection{Slightly Bigger Dimensions}\label{sec:ising-bigger-dim}

A big challenge for the Ising model is its high computational complexity as it involves summing over all binary vectors of length $p = \prod_{k = 1}^r p_k$ in the partition function \eqref{eq:ising-partition-function}. Exact computation of the partition function requires summing all $2^p$ binary vectors. For small dimensions, say $p\approx 10$, this is easily computed. Increasing the dimension beyond $20$ becomes extremely expensive and impossible for a dimension bigger than $30$. Trying to avoid the evaluation of the log-likelihood and only computing its partial gradients via \cref{thm:grad} does not resolve the issue. The gradients require the inverse link, that is the second moment \eqref{eq:ising-m2}, which still involves summing $2^p$ terms if the scaling factor $p_0$ is dropped. Basically, with our model, this means that the optimization of the Ising model using exactly computed gradients is impossible for moderately sized problems.

When $p = \prod_{i = 1}^r p_i > 20$, we use a Monte-Carlo method to estimate the second moment \eqref{eq:ising-m2}, required to compute the partial gradients of the log-likelihood. Specifically, we use a Gibbs-Sampler to sample from the conditional distribution and approximate the second moment in an importance sampling framework. This can be implemented quite efficiently and the estimation accuracy for the second moment is evaluated experimentally. Simultaneously, we use the same approach to estimate the partition function. This, though, is inaccurate and may only be used to get a rough idea of the log-likelihood. Regardless, for our method, we only need the gradient for optimization where appropriate break conditions, not based on the likelihood, lead to a working method for MLE estimation.

\begin{figure}
    \centering
    \includegraphics{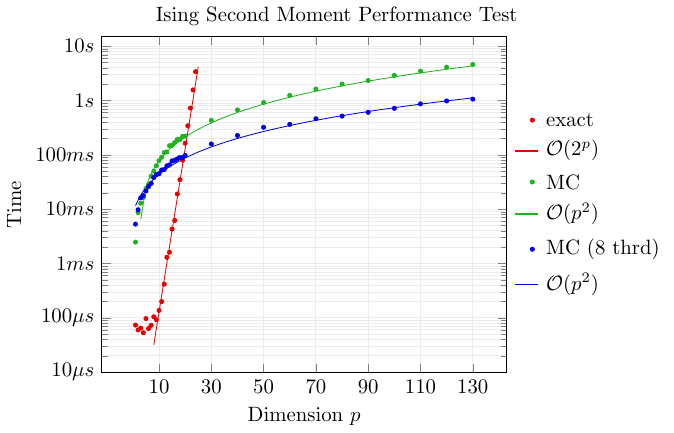}
    \caption{\label{fig:ising-m2-perft}Performance test for computing/estimating the second moment of the Ising model of dimension $p$ using either the exact method or a Monte-Carlo (MC) simulation.}
\end{figure}

\section{Manifolds}\label{sec:manifolds}

\cref{thm:sdr} finds the sufficient reduction for the regression of $Y$ on $\ten{X}$ in the population. Any estimation of the sufficient reduction requires application of some optimality criterion. As we operate within the framework of the exponential family, we opted for maximum likelihood estimation (MLE). For the unconstrained problem, where the parameters are simply $\mat{B}$ and $\mat{\Omega}$ in \eqref{eq:eta1-manifold}, maximizing the likelihood of $\ten{X} \mid Y$ is straightforward and yields well-defined MLEs of both parameters. Our setting, though, requires the constrained optimization of the $\ten{X} \mid Y$ likelihood subject to $\mat{B} = \bigotimes_{j = r}^{1}\mat{\beta}_j$ and $\mat{\Omega}=\bigkron_{j = r}^{1}\mat{\Omega}_j$. \Cref{thm:kron-manifolds,thm:param-manifold} provide the setting for which the MLE of the constrained parameter $\mat{\theta}$ is well-defined, which in turn leads to the derivation of its asymptotic normality.

The main problem in obtaining asymptotic results for the MLE of the constrained parameter $\mat{\theta} = (\overline{\ten{\eta}}, \vec\mat{B}, \vech\mat{\Omega})$ stems from the nature of the constraint. We assumed that $\mat{B} = \bigkron_{k = r}^{1}\mat{\beta}_k$, where the parameter $\mat{B}$ is identifiable. This means that different values of $\mat{B}$ lead to different densities $f_{\mat{\theta}}(\ten{X}\mid Y = y)$, a basic property needed to ensure consistency of parameter estimates, which in turn is needed for asymptotic normality. On the other hand, the components $\mat{\beta}_j$, $j = 1, \ldots, r$, are \emph{not} identifiable, which is a direct consequence of the equality $\mat{\beta}_2\otimes\mat{\beta}_1 = (c\mat{\beta}_2)\otimes (c^{-1}\mat{\beta}_1)$ for every $c\neq 0$. This is the reason we considered $\Theta$ as a constrained parameter space instead of parameterizing the densities of $\ten{X}\mid Y$ with $\mat{\beta}_1, \ldots, \mat{\beta}_r$. The same is true for $\mat{\Omega} = \bigkron_{k = r}^{1}\mat{\Omega}_k$.

In addition to identifiable parameters, the asymptotic normality obtained in \cref{thm:asymptotic-normality-gmlm} requires differentiation. Therefore, the space itself must admit defining differentiation, which is usually a vector space. This is too strong an assumption for our purposes. To weaken the vector space assumption, we consider \emph{smooth manifolds}. The latter are spaces that look like Euclidean spaces locally and allow the notion of differentiation. The more general \emph{topological} manifolds are too weak for differentiation. To make matters worse, a smooth manifold only allows for first derivatives. Without going into details, the solution is a \emph{Riemannian manifold}. Similar to an abstract \emph{smooth manifold}, Riemannian manifolds are detached from our usual intuition as well as complicated to handle in an already complicated setting. This is where an \emph{embedded (sub)manifold} comes to the rescue. Simply speaking, an embedded manifold is a manifold that is a subset of a manifold from which it inherits its properties. If a manifold is embedded in a Euclidean space, almost all the complications of abstract manifold theory simplify drastically. Moreover, since an Euclidean space is itself a Riemannian manifold, we inherit the means for higher derivatives. Finally, a smooth embedded submanifold structure for the parameter space maintains consistency with existing approaches and results for parameter sets with linear subspace structure. These reasons justify the constraint that the parameter space $\Theta$ be a \emph{smooth embedded submanifold} in an open subset $\Xi$ of a Euclidean space.

Now, we define a \emph{smooth manifold} embedded in $\mathbb{R}^p$ without detours to the more general theory. See, for example, \cite{Lee2012,Lee2018,AbsilEtAl2007,Kaltenbaeck2021} among others.
\begin{definition}[Manifolds]\label{def:manifold}
A set $\manifold{A}\subseteq\mathbb{R}^p$ is an \emph{embedded smooth manifold} of dimension $d$ if for every $\mat{x}\in\manifold{A}$ there exists a smooth\footnote{Here \emph{smooth} means infinitely differentiable or $C^{\infty}$.} bi-continuous map $\varphi:U\cap\manifold{A}\to V$, called a \emph{chart}, with $\mat{x}\in U\subseteq\mathbb{R}^p$ open and $V\subseteq\mathbb{R}^d$ open.
\end{definition}

We also need the concept of a \emph{tangent space} to formulate asymptotic normality in a way that is independent of a particular coordinate representation. Intuitively, the tangent space at a point $\mat{x}\in\manifold{A}$ of the manifold $\manifold{A}$ is the hyperspace of all velocity vectors $\t{\nabla\gamma(0)}$ of any curve $\gamma:(-1, 1)\to\manifold{A}$ passing through $\mat{x} = \gamma(0)$, see \cref{fig:torus}. Locally, at $\mat{x} = \gamma(0)$ with a chart $\varphi$ we can write $\gamma(t) = \varphi^{-1}(\varphi(\gamma(t)))$ that gives $\Span\t{\nabla\gamma(0)} \subseteq \Span\t{\nabla\varphi^{-1}(\varphi(\mat{x}))}$. Taking the union over all smooth curves through $\mat{x}$ gives equality. The following definition leverages the simplified setup of smooth manifolds in Euclidean space.

\begin{definition}[Tangent Space]\label{def:tangent-space}
    Let $\manifold{A}\subseteq\mathbb{R}^p$ be an embedded smooth manifold and $\mat{x}\in\manifold{A}$. The \emph{tangent space} at $\mat{x}$ of $\manifold{A}$ is defined as
    \begin{displaymath}
        T_{\mat{x}}\manifold{A} := \Span\t{\nabla\varphi^{-1}(\varphi(\mat{x}))}
    \end{displaymath}
    for any chart $\varphi$ with $\mat{x}$ in the pre-image of $\varphi$.
\end{definition}

\Cref{def:tangent-space} is consistent since it can be shown that two different charts at the same point have identical span.

\begin{figure}
    \centering
    \includegraphics[width = 0.5\textwidth]{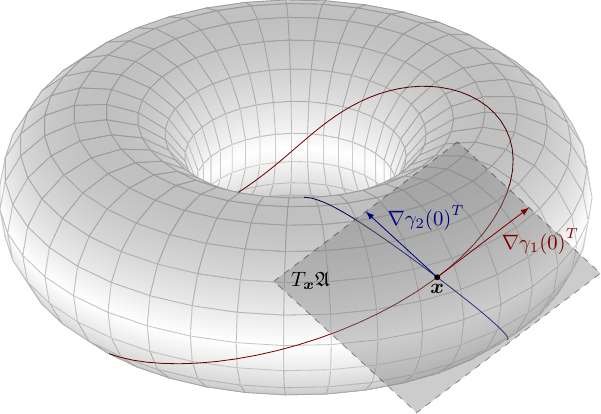}
    \caption{\label{fig:torus}Visualization of the tangent space $T_{\mat{x}}\manifold{A}$ at $\mat{x}$ of the torus $\manifold{A}$. The torus $\manifold{A}$ is a 2-dimensional embedded manifold in $\mathbb{R}^3$. The tangent space $T_{\mat{x}}\manifold{A}\subset\mathbb{R}^3$ is a 2-dimensional hyperplane visualized with its origin $\mat{0}$ shifted to $\mat{x}$. Two curves $\gamma_1, \gamma_2$ with $\mat{x} = \gamma_1(0) = \gamma_2(0)$ are drawn on the torus. The curve velocity vectors $\t{\nabla\gamma_1(0)}$ and $\t{\nabla\gamma_2(0)}$ are drawn as tangent vectors with root $\mat{x}$.}
\end{figure}

\subsection{Kronecker Product Manifolds}\label{sec:kron-manifolds}

As a basis to ensure that the constrained parameter space $\Theta$ is a manifold, which is a requirement of \cref{thm:param-manifold}, we need \cref{thm:kron-manifolds}. Therefore, we need the notion of a \emph{spherical} set, which is a set $\manifold{A}$, on which the Frobenius norm is constant. That is, $\|\,.\,\|_F:\manifold{A}\to\mathbb{R}$ is constant. Forthermore, we call a scale invariant set $\manifold{A}$ a \emph{cone}, that is $\manifold{A} = \{ c \mat{A} : \mat{A}\in\manifold{A} \}$ for all $c > 0$.

\begin{theorem}[Kronecker Product Manifolds]\label{thm:kron-manifolds}
    Let $\manifold{A}\subseteq\mathbb{R}^{p_1\times q_1}\backslash\{\mat{0}\}, \manifold{B}\subseteq\mathbb{R}^{p_2\times q_2}\backslash\{\mat{0}\}$ be smooth embedded submanifolds. Assume one of the following conditions holds.
    \begin{itemize}
        \item[-] ``sphere condition'':
            At least one of $\manifold{A}$ or $\manifold{B}$ is \emph{spherical} and let $d = \dim\manifold{A} + \dim\manifold{B}$.
        \item[-] ``cone condition'':
            Both $\manifold{A}$ and $\manifold{B}$ are \emph{cones} and let $d = \dim\manifold{A} + \dim\manifold{B} - 1$.
    \end{itemize}
    Then, $\{ \mat{A}\otimes \mat{B} : \mat{A}\in\manifold{A}, \mat{B}\in\manifold{B} \}\subset\mathbb{R}^{p_1 p_2\times q_1 q_2}$ is a smooth embedded $d$-manifold.
\end{theorem}

\begin{theorem}[Parameter Manifold]\label{thm:param-manifold}
    Let
    \begin{displaymath}
        \manifold{K}_{\mat{B}} = \Bigl\{ \bigkron_{k = r}^{1}\mat{\beta}_k : \mat{\beta}_k\in\manifold{B}_k \Bigr\}
        \quad\text{and}\quad
        \manifold{K}_{\mat{\Omega}} = \Bigl\{ \bigkron_{k = r}^{1}\mat{\Omega}_k : \mat{\Omega}_k\in\manifold{O}_k \Bigr\}
    \end{displaymath}
    where $\manifold{B}_k\subseteq\mathbb{R}^{p_k\times q_k}\backslash\{\mat{0}\}$ and $\manifold{O}_k\subseteq\mathbb{R}^{p_k\times p_k}\backslash\{\mat{0}\}$ are smooth embedded manifolds which are either spheres or cones, for $k = 1, ..., r$. Also, let
    \begin{displaymath}
        \manifold{CK}_{\mat{\Omega}} = \{ \vech{\mat{\Omega}} : \mat{\Omega}\in\manifold{K}_{\mat{\Omega}} \land \pinv{(\mat{T}_2\pinv{\mat{D}_p})}\mat{T}_2\pinv{\mat{D}_p}\vec{\mat{\Omega}} = \vec{\mat{\Omega}} \}.
    \end{displaymath}
Then, the constrained parameter space $\Theta = \mathbb{R}^p \times \manifold{K}_{\mat{B}}\times\manifold{CK}_{\mat{\Omega}}\subset\mathbb{R}^{p(p + 2 q + 3) / 2}$ is a smooth embedded manifold.
\end{theorem}

\subsection{Matrix Manifolds}\label{sec:matrix-manifolds}

A powerful feature of \cref{thm:param-manifold} is the modeling flexibility it provides. For example, we can perform low-rank regression. Or, we may constrain two-way interactions between direct axis neighbors by using band matrices for the $\mat{\Omega}_k$'s, among others.

This flexibility derives from many different matrix manifolds that can be used as building blocks $\manifold{B}_k$ and $\manifold{O}_k$ of the parameter space $\Theta$ in \cref{thm:param-manifold}. A list of possible choices, among others, is given in \cref{tab:matrix-manifolds}. As long as parameters in $\Theta$ are a valid parameterization of a density (or probability mass function) of \eqref{eq:quadratic-exp-fam} subject to \eqref{eq:eta1-manifold} and \eqref{eq:eta2-manifold}, one may choose any of the manifolds listed in \cref{tab:matrix-manifolds} which are either cones or spherical. We also included an example which is neither a sphere nor a cone. They may also be valid building blocks but require more work as they are not directly leading to a parameter manifold by \cref{thm:param-manifold}. In case one can show the resulting parameter space $\Theta$ is an embedded manifold, the asymptotic theory of \cref{sec:asymtotics} is applicable.

\begin{table}
    \centering
    \begin{tabular}{l l c c c}
        \hline
        Symbol & Description & C & S & Dimension\\
        \hline
        $\mathbb{R}^{p\times q}$ & All matrices of dimension $p\times q$ &
            \checkmark & \xmark     & $p q$ \\
            $\mathbb{R}_{*}^{p\times q}$ & Full rank $p\times q$ matrices &
            \checkmark & \xmark     & $p q$                  \\
        $\Stiefel{p}{q}$ & \emph{Stiefel Manifold}, $\{ \mat{U}\in\mathbb{R}^{p\times q} : \t{\mat{U}}\mat{U} = \mat{I}_q \}$ for $q\leq p$ &
            \xmark     & \checkmark & $p q - q (q + 1) / 2$  \\
        $\mathcal{S}^{p - 1}$ & Unit sphere in $\mathbb{R}^p$, special case $\Stiefel{p}{1}$ &
            \xmark     & \checkmark & $p - 1$                \\
        $\OrthoGrp{p}$ & Orthogonal Group, special case $\Stiefel{p}{p}$ &
            \xmark     & \checkmark & $p (p - 1) / 2$        \\
        $\SpecialOrthoGrp{p}$ & Special Orthogonal Group $\{ \mat{U}\in U(p) : \det{\mat{U}} = 1 \}$ &
            \xmark     & \checkmark & $p (p - 1) / 2$        \\
        $\mathbb{R}_{r}^{p\times q}$ & Matrices of known rank $r > 0$, generalizes $\StiefelNonCompact{p}{q}$ &
            \checkmark & \xmark     & $r(p + q - r)$         \\
        & Symmetric matrice &
            \checkmark & \xmark     & $p (p + 1) / 2$    \\
        $\SymPosDefMat{p}$ & Symmetric Positive Definite matrices &
            \checkmark & \xmark     & $p (p + 1) / 2$    \\
        & Scaled Identity $\{ a\mat{I}_p : a\in\mathbb{R}_{+} \}$ &
            \checkmark & \xmark     & $1$    \\
        & Symmetric $r$-band matrices (includes diagonal) &
            \checkmark & \xmark     & $(2 p - r) (r + 1) / 2$    \\
        & Auto correlation $\{ \mat{A}\in\mathbb{R}^{p\times p} : \mat{A}_{i j} = \rho^{|i - j|}, \rho\in(0, 1) \}$ &
            \xmark     & \xmark     & $1$ \\
        \hline
    \end{tabular}
    \caption{\label{tab:matrix-manifolds}Examples of embedded matrix manifolds. ``Symbol'' a (more or less) common notation for the matrix manifold, if at all. ``C'' stands for \emph{cone}, meaning it is scale invariant. ``S'' means \emph{spherical}, that is, constant Frobenius norm.}
\end{table}

\begin{remark}
    The \emph{Grassmann Manifold} of $q$ dimensional subspaces in $\mathbb{R}^p$ is not listed in \cref{tab:matrix-manifolds} since it is not embedded in $\mathbb{R}^{p \times q}$.
\end{remark}

\section{Statistical Properties}\label{sec:statprop}

\subsection{Asymptotics}\label{sec:asymtotics}

Let $Z$ be a random variable with density $f_{\mat{\theta_0}}\in\{ f_{\mat{\theta}}: \mat{\theta}\in\Theta \}$, where $\Theta$ is a subset of a Euclidean space. We want to estimate the parameter that indexes the pdf, ${\mat{\theta}}_0$, using $n$ i.i.d. (independent and identically distributed) copies of $Z$. We assume a known, real-valued and measurable function $z\mapsto m_{\mat{\theta}}(z)$ for every $\mat{\theta}\in\Theta$ and that ${\mat{\theta}}_0$ is the unique maximizer of the map $\mat{\theta}\mapsto M(\mat{\theta}) = \E m_{\mat{\theta}}(Z)$. For the estimation we maximize the empirical version
\begin{align}\label{eq:Mn}
    M_n(\mat{\theta}) &= \frac{1}{n}\sum_{i = 1}^n m_{\mat{\theta}}(Z_i).
\end{align}
An \emph{M-estimator} $\hat{\mat{\theta}}_n = \hat{\mat{\theta}}_n(Z_1, \ldots, Z_n)$ is a maximizer for the objective function $M_n$ over the parameter space $\Theta$ defined as
\begin{displaymath}
    \hat{\mat{\theta}}_n = \argmax_{\mat{\theta}\in\Theta} M_n(\mat{\theta}).
\end{displaymath}
It is not necessary to have a perfect maximizer, as long as the objective has finite supremum, it is sufficient to take an \emph{almost maximizer} $\hat{\mat{\theta}}_n$ as defined in the following;

\begin{definition}[weak and strong M-estimators]
    An estimator $\hat{\mat{\theta}}_n$ for the objective function $M_n$ in \eqref{eq:Mn} with $\sup_{\mat{\theta}\in\Theta}M_n(\mat{\theta}) < \infty$ such that
    \begin{displaymath}
        M_n(\hat{\mat{\theta}}_n) \geq \sup_{\mat{\theta}\in\Theta}M_n(\mat{\theta}) - o_P(n^{-1})
    \end{displaymath}
    is called a \emph{strong M-estimator} over $\Theta$. Replacing $o_P(n^{-1})$ by $o_P(1)$ gives a \emph{weak M-estimator}.
\end{definition}

\begin{theorem}[Asymptotic Normality]\label{thm:asymptotic-normality-gmlm}
    Assume $Z = (\ten{X}, Y)$ satisfies model \eqref{eq:quadratic-exp-fam} subject to \eqref{eq:eta1-manifold} and \eqref{eq:eta2-manifold} with true constrained parameter $\mat{\theta}_0 = (\overline{\eta}_0, \mat{B}_0, \mat{\Omega}_0)\in\Theta$, where $\Theta$ is defined in \cref{thm:param-manifold}. Under the regularity \crefrange{cond:differentiable-and-convex}{cond:finite-sup-on-compacta} in \cref{app:proofs}, there exists a strong M-estimator sequence $\hat{\mat{\theta}}_n$ deriving from $l_n$ in \eqref{eq:log-likelihood} over $\Theta$. Furthermore, any strong M-estimator $\hat{\mat{\theta}}_n$ converges in probability to the true parameter $\mat{\theta}_0$, $\hat{\mat{\theta}}_n\xrightarrow{p}\mat{\theta}_0$, over $\Theta$. Moreover, every strong M-estimator $\hat{\mat{\theta}}_n$ is asymptotically normal,
    \begin{displaymath}
        \sqrt{n}(\hat{\mat{\theta}}_n - \mat{\theta}_0) \xrightarrow{d} \mathcal{N}_{p(p + 2 q + 3) / 2}(0, \mat{\Sigma}_{\mat{\theta}_0})
    \end{displaymath}
    with asymptotic variance-covariance structure $\mat{\Sigma}_{\mat{\theta}_0}$ given in \eqref{eq:asymptotic-covariance-gmlm}.
\end{theorem}

\subsection{Asymptotic Normality}

The following is a reformulation of Lemma~2.3 from \cite{BuraEtAl2018} which assumes Condition~2.2 to hold. The existence of a mapping in Condition~2.2 is not needed for Lemma~2.3. It suffices that the restricted parameter space $\Theta$ is a subset of the unrestricted parameter space $\Xi$, which is trivially satisfied in our setting. Under this, \cref{thm:exists-strong-M-estimator-on-subsets} follows directly from Lemma~2.3 in \cite{BuraEtAl2018}.

\begin{theorem}[Existence of strong M-estimators on Subsets]\label{thm:exists-strong-M-estimator-on-subsets}
    Assume there exists a (weak/strong) M-estimator $\hat{\mat{\xi}}_n$ for $M_n$ over $\Xi$. Then, there exists a strong M-estimator $\hat{\mat{\theta}}_n$ for $M_n$ over any non-empty $\Theta\subseteq\Xi$.
\end{theorem}

\begin{theorem}[Existence and Consistency of M-estimators on Subsets]\label{thm:M-estimator-consistency-on-subsets}
    Let $\Xi$ be a convex open subset of a Euclidean space and $\Theta\subseteq\Xi$ non-empty. Assume $\mat{\xi}\mapsto m_{\mat{\xi}}(z)$ is a strictly concave function on $\Xi$ for almost all $z$ and $z\mapsto m_{\mat{\xi}}(z)$ is measurable for all $\mat{\xi}\in\Xi$. Let $M(\mat{\xi}) = \E m_{\mat{\xi}}(Z)$ be a well defined function with a unique maximizer $\mat{\theta}_0\in\Theta\subseteq\Xi$; that is, $M(\mat{\theta}_0) > M(\mat{\xi})$ for all $\mat{\xi}\neq\mat{\theta}_0$. Also, assume 
    \begin{displaymath}
        \E\sup_{\mat{\xi}\in K}|m_{\mat{\xi}}(Z)| < \infty,
    \end{displaymath}
    for every non-empty compact $K\subset\Xi$. Then, there exists a strong M-estimator $\hat{\mat{\theta}}_n$ of $M_n(\mat{\theta}) = \frac{1}{n}\sum_{i = 1}^{n} m_{\mat{\theta}}(Z_i)$ over the subset $\Theta$. Moreover, any strong M-estimator $\hat{\mat{\theta}}_n$ of $M_n$ over $\Theta$ converges in probability to $\mat{\theta}_0$, that is $\hat{\mat{\theta}}_n\xrightarrow{p}\mat{\theta}_0$.
\end{theorem}

\begin{theorem}[Asymptotic Normality for M-estimators on Manifolds]\label{thm:M-estimator-asym-normal-on-manifolds}
    Let $\Theta\subseteq\mathbb{R}^p$ be a smooth embedded manifold. For each $\mat{\theta}$ in a neighborhood in $\mathbb{R}^p$ of the true parameter $\mat{\theta}_0\in\Theta$ let $z\mapsto m_{\mat{\theta}}(z)$ be measurable and $\mat{\theta}\mapsto m_{\mat{\theta}}(z)$ be differentiable at $\mat{\theta}_0$ for almost all $z$. Assume also that there exists a measurable function $u$ such that $\E[u(Z)^2] < \infty$, and for almost all $z$ as well as all $\mat{\theta}_1, \mat{\theta}_2$ in a neighborhood of $\mat{\theta}_0$ such that
    \begin{displaymath}
        | m_{\mat{\theta}_1}\!(z) - m_{\mat{\theta}_2}\!(z) | \leq u(z) \| \mat{\theta}_1 - \mat{\theta}_2 \|_2.
    \end{displaymath}
    Moreover, assume that $\mat{\theta}\mapsto\E[m_{\mat{\theta}}(Z)]$ admits a second-order Taylor expansion at $\mat{\theta}_0$ in a neighborhood of $\mat{\theta}_0$ in $\mathbb{R}^p$ with a non-singular Hessian $\mat{H}_{\mat{\theta}_0} = \nabla^2_{\mat{\theta}}\E[m_{\mat{\theta}}(Z)]|_{\mat{\theta} = \mat{\theta}_0}\in\mathbb{R}^{p\times p}$.

    If $\hat{\mat{\theta}}_n$ is a strong M-estimator of $\mat{\theta}_0$ in $\Theta$, then $\hat{\mat{\theta}}_n$ is asymptotically normal
    \begin{displaymath}
        \sqrt{n}(\hat{\mat{\theta}}_n - \mat{\theta}_0) \xrightarrow{d} \mathcal{N}_p\left(\mat{0}, \mat{\Pi}_{\mat{\theta}_0} \E\left[\nabla_{\mat{\theta}} m_{\mat{\theta}_0}(Z)\t{(\nabla_{\mat{\theta}} m_{\mat{\theta}_0}(Z))}\right] \mat{\Pi}_{\mat{\theta}_0}\right)
    \end{displaymath}
    where $\mat{\Pi}_{\mat{\theta}_0} = \mat{P}_{\mat{\theta}_0}\pinv{(\t{\mat{P}_{\mat{\theta}_0}}\mat{H}_{\mat{\theta}_0}\mat{P}_{\mat{\theta}_0})}\t{\mat{P}_{\mat{\theta}_0}}$ and $\mat{P}_{\mat{\theta}_0}$ is any matrix whose span is the tangent space $T_{\mat{\theta}_0}\Theta$ of $\Theta$ at $\mat{\theta}_0$.
\end{theorem}

\begin{remark}
    \cref{thm:M-estimator-asym-normal-on-manifolds} has as special case Theorem~5.23 in \cite{vanderVaart1998}, when $\Theta$ is an open subset of a Euclidean space, which is the simplest form of an embedded manifold.
\end{remark}

\section{Simulations}\label{sec:simulations}

In this section, we report simulation results for the multi-linear normal and the multi-linear Ising model where different aspects of the GMLM model are compared against other methods. The comparison methods are Tensor Sliced Inverse Regression (TSIR) \cite[]{DingCook2015}, Multiway Generalized Canonical Correlation Analysis (MGCCA) \cite[]{ChenEtAl2021,GirkaEtAl2024} and the Tucker decomposition that is a higher-order form of principal component analysis (HOPCA) \cite[]{KoldaBader2009}, for both continuous and binary data. For the latter, the binary values are treated as continuous. As part of our baseline analysis, we also incorporate traditional Principal Component Analysis (PCA) on vectorized observations. In the case of the Ising model, we also compare with LPCA (Logistic PCA) and CLPCA (Convex Logistic PCA), both introduced in \cite{LandgrafLee2020}. All experiments are performed at sample sizes $n = 100, 200, 300, 500$ and $750$. Every experiment is repeated $100$ times.

To assess the accuracy of the estimation of $\ten{R}(\ten{X})$ in \cref{thm:sdr}, we compare the estimate with the true vectorized reduction matrix $\mat{B} = \bigkron_{k = r}^{1}\mat{\beta}_k$, as it is compatible with any linear reduction method. We compute the \emph{subspace distance}, $d(\mat{B}, \hat{\mat{B}})$, between $\mat{B}\in\mathbb{R}^{p\times q}$ and an estimate $\hat{\mat{B}}\in\mathbb{R}^{p\times \tilde{q}}$, which satisfies 
\begin{displaymath}
    d(\mat{B}, \hat{\mat{B}}) \propto \| \mat{B}\pinv{(\t{\mat{B}}\mat{B})}\t{\mat{B}} - \hat{\mat{B}}\pinv{(\t{\hat{\mat{B}}}\hat{\mat{B}})}\t{\hat{\mat{B}}} \|_F,
\end{displaymath}
where $\propto$ signifies proportional to. The proportionality constant\footnote{Depends on row dimension $p$ and the ranks of $\mat{B}$ and $\hat{\mat{B}}$ given by $(\min(\rank\mat{B} + \rank\hat{\mat{B}}, 2 p - (\rank\mat{B} + \rank\hat{\mat{B}})))^{-1/2}$.} ensures $d(\mat{B}, \hat{\mat{B}}) \in [0, 1]$. A distance of zero implies space overlap and a distance of one implies orthogonality of the subspaces.

\subsection{Multi-Linear Normal}\label{sec:sim-tensor-normal}

We generate a random sample $y_i$, $i=1,\ldots, n$, from the standard normal distribution. We then draw i.i.d. samples $\ten{X}_i$ for $i = 1, ..., n$ from the conditional multi-linear normal distribution of $\ten{X}\mid Y = y_i$. The conditional distribution $\ten{X}\mid Y = y_i$ depends on the choice of the GMLM parameters $\overline{\ten{\eta}}$, $\mat{\beta}_1, ..., \mat{\beta}_r$, $\mat{\Omega}_1, ..., \mat{\Omega}_r$, and the function $\ten{F}_y$ of $y$. In all experiments we set $\overline{\ten{\eta}} = \mat{0}$. The other parameters and $\ten{F}_y$ are described per experiment. With the true GMLM parameters and $\ten{F}_y$ given, we compute the conditional multi-linear normal mean $\ten{\mu}_y = \ten{F}_y\mlm_{k = 1}^{r}\mat{\Omega}_k^{-1}\mat{\beta}_k$ and covariances $\mat{\Sigma}_k = \mat{\Omega}_k^{-1}$ as in \eqref{eq:tnormal_cond_params}.

We consider the following settings:
\begin{itemize}
    \item[1a)] $\ten{X}$ is a three-way ($r = 3$) array of dimension $2\times 3\times 5$, and $\ten{F}_y\equiv y$ is a $1\times 1\times 1$ tensor. The true $\mat{\beta}_k$'s are all equal to $\mat{e}_1\in\mathbb{R}^{p_k}$, the first unit vector, for $k \in \{1, 2, 3\}$. The matrices $\mat{\Omega}_k = \mathrm{AR}(0.5)$ follow an auto-regression like structure. That is, the elements are given by $(\mat{\Omega}_k)_{i j} = 0.5^{|i - j|}$.
    \item[1b)] $\ten{X}$ is a three-way ($r = 3$) array of dimension $2\times 3\times 5$, and relates to the response $y$ via a qubic polynomial. This is modeled via $\ten{F}_y$ of dimension $2\times 2\times 2$ by the twice iterated outer product of the vector $(1, y)$. Element-wise this reads $(\ten{F}_y)_{i j k} = y^{i + j + k - 3}$. All $\mat{\beta}_k$'s are set to $(\mat{e}_1, \mat{e}_2)\in\mathbb{R}^{p_k\times 2}$ with $\mat{e}_i$ the $i$th unit vector and the $\mat{\Omega}_k$'s are $\mathrm{AR}(0.5)$.
    \item[1c)] Same as 1b), except that the GMLM parameters $\mat{\beta}_k$ are rank $1$ given by
    \begin{displaymath}
        \mat{\beta}_1 = \begin{pmatrix} 1 & -1 \\ -1 & 1 \end{pmatrix},\quad
        \mat{\beta}_2 = \begin{pmatrix} 1 & -1 \\ -1 & 1 \\ 1 & -1 \end{pmatrix},\quad
        \mat{\beta}_3 = \begin{pmatrix} 1 & -1 \\ -1 & 1 \\ 1 & -1 \\ -1 & 1 \\ 1 & -1 \end{pmatrix}.
    \end{displaymath}
    \item[1d)] Same as 1b), but the true $\mat{\Omega}_k$ is tri-diagonal, for $k = 1, 2, 3$. Their elements are given by $(\mat{\Omega}_k)_{i j} = \delta_{0, |i - j|} + 0.5\delta_{1, |i - j|}$ with $\delta_{i, j}$ being the Kronecker delta.
    \item[1e)] For the misspecification model we let $\ten{X}\mid Y$ be multivariate but \emph{not} multi-linear normal. Let $\ten{X}$ be a $5\times 5$ random matrix with normal entries, $Y$ univariate standard normal and $\mat{f}_y$ a $4$ dimensional vector given by $\mat{f}_y = (1, \sin(y), \cos(y), \sin(y)\cos(y))$. The true vectorized reduction matrix $\mat{B}$ is $25\times 4$ consisting of the first $4$ columns of the identity; i.e., $\mat{B}_{i j} = \delta_{i j}$. The variance-covariance matrix $\mat{\Sigma}$ has elements $\mat{\Sigma}_{i j} = 0.5^{|i - j|}$. Both, $\mat{B}$ and $\mat{\Omega} = \mat{\Sigma}^{-1}$ violate the Kronecker product assumptions \eqref{eq:eta1} and \eqref{eq:eta2} of the GMLM model. Then, we set
    \begin{displaymath}
        \vec{\ten{X}}\mid (Y = y) = \mat{B}\mat{f}_y + \mathcal{N}_{25}(\mat{0}, \mat{\Sigma}).
    \end{displaymath}
    Furthermore, we fit the model with the wrong ``known'' function $\ten{F}_y$. We set $\ten{F}_y$ to be a $2\times 2$ matrix with $(\ten{F}_y)_{i j} = y^{i + j - 2}$, $i,j=1,2$.
\end{itemize}

The final multi-linear normal experiment 1e) is a misspecified model to explore the robustness of our approach. The true model does \emph{not} have a Kronecker structure and the ``known'' function $\ten{F}_y$ of $y$ is misspecified as well.

The results are visualized in \cref{fig:sim-normal}. Simulation 1a), given a 1D linear relation between the response $Y$ and $\ten{X}$, TSIR and GMLM are equivalent. This is expected as \cite{DingCook2015} already established that TSIR gives the MLE estimate under a multi-linear (matrix) normal distributed setting. For the other methods, MGCCA is only a bit better than PCA which, unexpectedly, beats HOPCA. But none of them are close to the performance of TSIR or GMLM. Continuing with 1b), where we introduced a cubic relation between $Y$ and $\ten{X}$, we observe a bigger deviation in the performance of GMLM and TSIR. This is caused mainly because we are estimating an $8$ dimensional subspace now, which amplifies the small performance boost, in the subspace distance, we gain by avoiding slicing. The GMLM model in 1c) behaves as expected, clearly being the best. The other results are surprising. First, PCA, HOPCA and MGCCA are visually indistinguishable. This is explained by a high signal-to-noise ratio in this particular example. But the biggest surprise is the failure of TSIR. Even more surprising is that the conditional distribution $\ten{X}\mid Y$ is multi-linear normal distributed which, in conjunction with $\cov(\vec\ten{X})$ having a Kronecker structure, should give the MLE estimate. The low-rank assumption is also not an issue, this simply relates to TSIR estimating a 1D linear reduction which fulfills all the requirements. Finally, a common known issue of slicing, used in TSIR, is that conditional multi-modal distributions can cause estimation problems due to the different distribution modes leading to vanishing slice means. Again, this is not the case in simulation 1c). An investigation into this behavior revealed the problem in the estimation of the mode covariance matrices $\mat{O}_k = \E[(\ten{X} - \E\ten{X})_{(k)}\t{(\ten{X} - \E\ten{X})_{(k)}}]$. The mode-wise reductions provided by TSIR are computed as $\hat{\mat{O}}_k^{-1}\hat{\mat{\Gamma}}_k$ where the poor estimation of $\hat{\mat{O}}_k$ causes the failure of TSIR. The poor estimate of $\mat{O}_k$ is rooted in the high signal-to-noise ratio in this particular simulation. GMLM does not have degenerate behavior for high signal-to-noise ratios but it is less robust in low signal-to-noise ratio setting where TSIR performs better in this specific example. Simulation 1d), incorporating information about the covariance structure behaves similarly to 1b), except that GMLM gains a statistically significant lead in estimation performance. The last simulation, 1e), where the model was misspecified for GMLM. GMLM, TSIR, as well as MGCCA, are on par where GMLM has a slight lead in the small sample size setting and MGCCA overtakes in higher sample scenarios. The PCA and HOPCA methods both still outperformed.

\begin{figure}[ht!]
    \centering
    \includegraphics{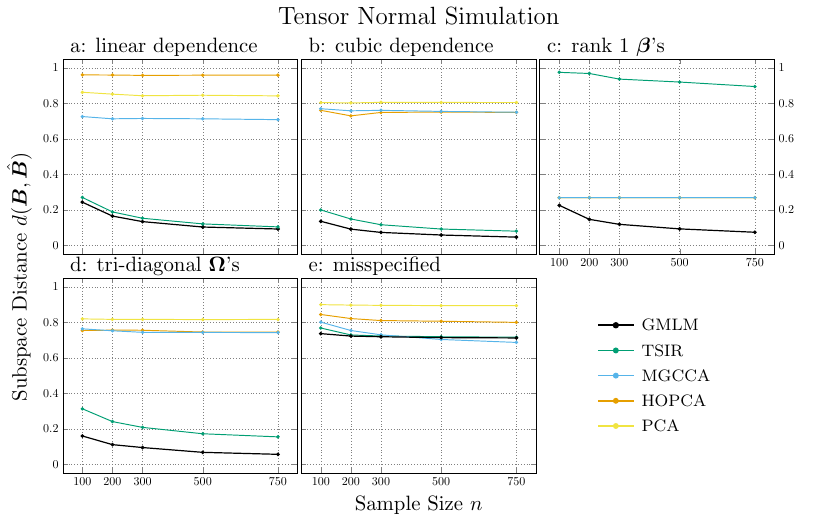}
    \caption{\label{fig:sim-normal}Multi-linear normal GMLM. The mean subspace distance $d(\mat{B}, \hat{\mat{B}})$ over $100$ replications is plotted versus sample size on the $x$-axis. The simulation settings are described in \cref{sec:sim-tensor-normal}.}
\end{figure}

\subsection{Ising Model}\label{sec:sim-ising}

We let $Y_i$ be i.i.d. uniform on $[-1,1]$, $i = 1, \ldots, n$ and $\ten{X}$ be $2\times 3$ with conditional matrix (multi-linear) Ising distribution $\ten{X}\mid Y$ as in \cref{sec:ising_estimation}. Unless otherwise specified, we set the GMLM parameters to be
\begin{displaymath}
    \mat{\beta}_1 = \begin{pmatrix}
        1 & 0 \\ 0 & 1
    \end{pmatrix}, \quad \mat{\beta}_2 = \begin{pmatrix}
        1 & 0 \\ 0 & 1 \\ 0 & 0
    \end{pmatrix}, \quad \mat{\Omega}_1 = \begin{pmatrix}
        0 & -2 \\ -2 & 0
    \end{pmatrix}, \quad \mat{\Omega}_2 = \begin{pmatrix}
        1 & 0.5 & 0 \\
        0.5 & 1 & 0.5 \\
        0 & 0.5 & 1
    \end{pmatrix}
\end{displaymath}
and
\begin{displaymath}
    \ten{F}_y = \begin{pmatrix}
        \sin(\pi y) & -\cos(\pi y) \\
        \cos(\pi y) & \sin(\pi y)
    \end{pmatrix}.
\end{displaymath}

\begin{itemize}
    \item[2a)] A purely linear relation between $\mat{X}$ and the response with $\ten{F}_y\equiv y:1\times 1$ and $\t{\mat{\beta}_1} = (1, 0)$ and $\t{\mat{\beta}_2} = (1, 0, 0)$.
    \item[2b)] The ``base'' simulation with all parameters as described above.
    \item[2c)] Low rank regression with both $\mat{\beta}_1$ and $\mat{\beta}_2$ of rank $1$,
    \begin{displaymath}
        \mat{\beta}_1 = \begin{pmatrix}
            1 & 0 \\ 1 & 0
        \end{pmatrix}, \qquad \mat{\beta}_2 = \begin{pmatrix}
            0 & 0 \\ 1 & -1 \\ 0 & 0
        \end{pmatrix}.
    \end{displaymath}
    \item[2d)] The original design of the Ising model is a mathematical model of the behavior of Ferromagnetism \cite{Ising1925} in a thermodynamic setting, modeling the interaction effects of elementary magnets (spin up/down relating to $0$ and $1$). The model assumes all elementary magnets to be the same, which translates to all having the same coupling strength (two-way interactions) governed by a single parameter relating to the temperature of the system. Assuming the magnets to be arranged in a 2D grid (matrix-valued $\ten{X}$), their interactions are constrained to direct neighbors, which is modeled by choosing the true $\mat{\Omega}_k$'s to be tri-diagonal matrices with zero diagonal entries and all non-zero entries identical. We set
    \begin{displaymath}
        \mat{\Omega}_1 = \frac{1}{2}\begin{pmatrix}
            0 & 1 \\ 1 & 0
        \end{pmatrix}, \qquad \mat{\Omega}_2 = \begin{pmatrix}
            0 & 1 & 0 \\
            1 & 0 & 1 \\
            0 & 1 & 0
        \end{pmatrix}
    \end{displaymath}
    where $1 / 2$ corresponds to arbitrary temperature. The mean effect depending on $\ten{F}_y$ can be interpreted as an external magnetic field.
\end{itemize}

\begin{figure}[ht!]
    \centering
    \includegraphics[scale=0.7]{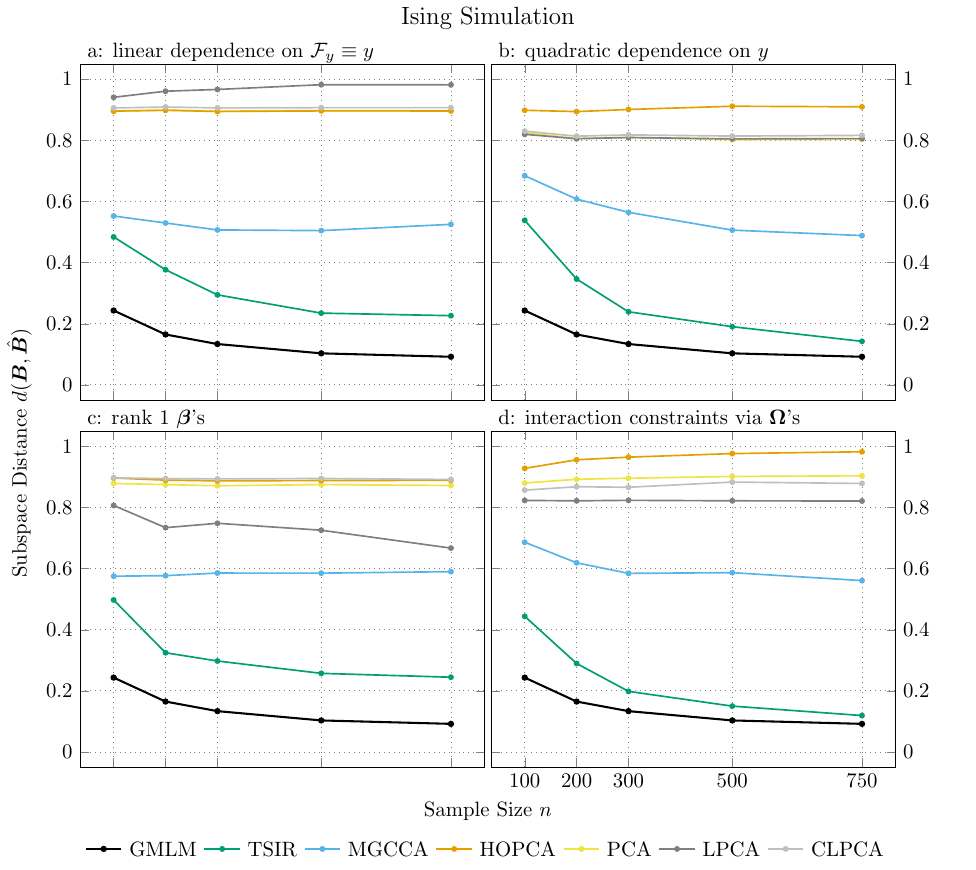}
    \caption{\label{fig:sim-ising}Visualization of the simulation results for Ising GMLM. Sample size on the $x$-axis and the mean of subspace distance $d(\mat{B}, \hat{\mat{B}})$ over $100$ replications on the $y$-axis. Described in \cref{sec:sim-ising}.}
\end{figure}

The subspace distances between the estimated and true parameter spaces for sample size $n=100, 200, 300, 500, 750$ are plotted in \cref{fig:sim-ising}. Regardless of the simulation setting, 2a-d), the comparative results are similar. We observe that PCA and HOPCA, both treating the response $\ten{X}$ as continuous, perform poorly. Not much better are LPCA and CLPCA. Similar to PCA and HOPCA they do not consider the relation to the response, but they are specifically created for binary predictors. MGCCA, which accounts for $y$ in reducing the predictors, clearly out-performs all the PCA variants. Even better is TSIR, even though it treats the predictors $\ten{X}$ as continuous. Finally, the Ising GMLM model is the best in all simulations.

\section{Data Analysis}\label{sec:data-analysis}

We apply GMLM and competing methods to two data sets.

\subsection{EEG Data}\label{sec:EEG}

We described the data in Section \ref{sec:introduction}. SDR methods for regression with matrix-valued predictors were developed by \cite{LiKimAltman2010} (folded SIR), \cite{PfeifferForzaniBura2012} (Longitudinal SIR), \cite{DingCook2014} (dimension folding PCA and PFC), \cite{PfeifferKaplaBura2021} (K-PIR (ls), K-PIR (mle)), and \cite{DingCook2015} (two-tensor SIR). In this data set, $p= p_1 p_2 = 16384$ is much larger than $n=122$. To deal with this issue, the predictor data were pre-screened via (2D)$^2$PCA \cite[]{ZhangZhou2005} which reduced the dimension of the predictor array to $(p_1, p_2) = (3, 4)$, $(15, 15)$ and $(20, 30)$ \cite[]{PfeifferKaplaBura2021, LiKimAltman2010, DingCook2014, DingCook2015}. \cite{PfeifferKaplaBura2021} also carried out simultaneous dimension reduction and variable selection using the fast POI-C algorithm of \cite{JungEtAl2019} (due to high computational burden, only 10-fold cross-validation was performed for fast POI-C).

In contrast to these methods, our GMLM model can be applied directly to the raw data of dimension $(256, 64)$ without pre-screening or variable selection. The high dimensionality issue in the GMLM model is resolved by the regularization trick used for numerical stability, as described in \cref{sec:tensor-normal-estimation}, without any change to the estimation procedure. In general, the sample size does not need to be large for maximum likelihood estimation in the multilinear normal model. For matrix normal models in particular, \cite{DrtonEtAl2020} proved that very small sample sizes, as little as $3$,\footnote{The required minimum sample size depends on non-trivial algebraic relations between the mode dimensions, while the magnitude of the dimensions has no specific role.} are sufficient to obtain unique MLEs for Kronecker covariance structures.

The discriminatory ability of the reduced predictors of the various approaches was assessed by the area under the receiver operator characteristics curve, AUC \cite[p. 67]{Pepe2003}. \cref{tab:eeg} reports the AUC values along with their standard deviation for the three compared methods. We use leave-one-out cross-validation to obtain unbiased AUC estimates. Then, we compare the GMLM model to the best performing methods from \cite{PfeifferKaplaBura2021}, namely K-PIR (ls) and LSIR from \cite{PfeifferForzaniBura2012} for $(p_1, p_2) = (3, 4)$, $(15, 15)$ and $(20, 30)$. For all applied pre-screening dimensions, K-PIR (ls) has an AUC of $78\%$. LSIR performs better at the price of some instability; it peaked at $85\%$ at $(3, 4)$, then dropped down to $81\%$ at $(15, 15)$, and then increased to $83\%$. In contrast, our GMLM method peaked at $(3, 4)$ with $85\%$ and stayed stable at $84\%$, even when no pre-processing was applied. In contrast, fast POI-C that carries out simultaneous feature extraction and feature selection is clearly outperformed by all other methods with an AUC of $63\%$. Folded SIR \cite[]{LiKimAltman2010} is excluded as it exhibits consistently lower AUC values, ranging from 0.61 to 0.70.

\begin{table}[!hpt]
    \centering
    \begin{tabular}{l l c c c}
        \hline
        Method  & \multicolumn{4}{c}{$(p_1, p_2)$} \\
                & \multicolumn{1}{c}{$(3, 4)$}
                & \multicolumn{1}{c}{$(15, 15)$}
                & \multicolumn{1}{c}{$(20, 30)$}
                & \multicolumn{1}{c}{$(256, 64)$} \\
        \hline
        K-PIR (ls) & 0.78 (0.04) & 0.78 (0.04) & 0.78 (0.04) & \\
        LSIR       & 0.85 (0.04) & 0.81 (0.04) & 0.83 (0.04) & \\
        TSIR       & 0.85 (0.04) & 0.83 (0.04) & 0.80 (0.04) & 0.69 (0.05) \\
        FastPOI-C  & & & & 0.63 (0.22)\makebox[0pt]{\ \ ${}^{*}$} \\
        GMLM       & 0.85 (0.04) & 0.84 (0.04) & 0.84 (0.04) & 0.84 (0.04) \\
        \hline
        \multicolumn{4}{r}{${}^{*}$\footnotesize based on 10-fold cross-validation.}
    \end{tabular}
    \caption{\label{tab:eeg}Mean AUC values and their standard deviation in parentheses based on leave-one-out cross-validation for the EEG imaging data (77 alcoholic and 45 control subjects)}
\end{table}

\subsection{Chess}\label{sec:chess}

The data set is provided by the \emph{lichess.org open database}.\footnote{\emph{lichess.org open database}. visited on December 8, 2023. \url{https://database.lichess.org}} We randomly selected the November of 2023 data that consist of more than $92$ million games. We removed all games without position evaluations. The evaluations, also denoted as scores, are from Stockfish,\footnote{The Stockfish developers [see \href{https://github.com/official-stockfish/Stockfish/blob/master/AUTHORS}{AUTHORS} file] (since 2008). \textit{Stockfish}. Stockfish is a free and strong UCI chess engine. \url{https://stockfishchess.org}} a free and strong chess engine. The scores take the role of the response $Y$ and correspond to a winning probability from the white pieces' point of view. Positive scores are good for white and negative scores indicate an advantage for black pieces. We ignore all highly unbalanced positions, which we set to be positions with absolute score above $5$. We also remove all positions with a mate score (one side can force checkmate). Furthermore, we only consider positions after $10$ half-moves to avoid oversampling the beginning of the most common openings including the start position which is in every game. Finally, we only consider positions with white to move. This leads to a final data set of roughly $64$ million positions, including duplicates.

A chess position is encoded as a set of $12$ binary matrices $\ten{X}_{\mathrm{piece}}$ of dimensions $8\times 8$. Every binary matrix encodes the positioning of a particular piece by containing a $1$ if the piece is present at the corresponding board position. The $12$ pieces derive from the $6$ types of pieces, namely pawns (\pawn), knights (\knight), bishops (\bishop), queens (\queen), and kings (\king) of two colors, black and white. See \cref{fig:fen2tensor} for a visualization.

\begin{figure}[!h]
    \centering
    \includegraphics[scale=0.8]{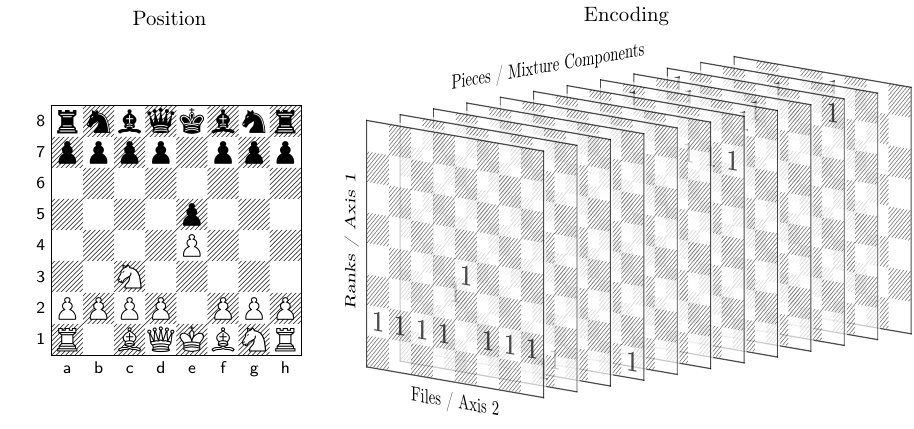}
    \caption{\label{fig:fen2tensor}The chess start position and its 3D binary tensor representation, empty entries are $0$.}
\end{figure}

We assume that $\ten{X}_{\mathrm{piece}}\mid Y = y$ follows an Ising GMLM model as in \cref{sec:ising_estimation} with different conditional piece predictors being independent. The simplifying assumption of independence results in a mixture model with log-likelihood
\begin{displaymath}
    l_n(\mat{\theta}) = \frac{1}{12}\sum_{\mathrm{piece}}l_n(\mat{\theta}_{\mathrm{piece}})
\end{displaymath}
where $l_n(\mat{\theta}_{\mathrm{piece}})$ is the Ising GMLM log-likelihood as in \cref{sec:ising_estimation} for $\ten{X}_{\mathrm{piece}}\mid Y = y$. For every component the same relation to the scores $y$ is modeled via a $2\times 2$ dimensional matrix-valued function $\ten{F}_y$ consisting of the monomials $1, y, y^2$, specifically $(\ten{F}_y)_{i j} = y^{i + j - 2}$.

Due to the volume of the data (millions of observations), it is computationally infeasible to compute the gradients on the entire data set. The dimension of the binary data is $12$ times a $8\times 8 = 768$ for every observation. The solution is to switch from a classic gradient-based optimization to a stochastic version; that is, every gradient update uses a new random subset of the entire data set. We draw independent random samples from the data consisting of $64$ million positions. The independence of samples is derived from the independence of games, and every sample is drawn from a different game.

\subsubsection{Validation}
Given the non-linear nature of the reduction, due to the quadratic matrix-valued function $\ten{F}_y$ of the score $y$, we use a \emph{generalized additive model}\footnote{using the function \texttt{gam()} from the \texttt{R} package \texttt{mgcv}.} (GAM) to predict position scores from reduced positions. The reduced positions are $48$ dimensional continuous values by combining the $12$ mixture components from the $2\times 2$ matrix-valued reductions per piece. The per-piece reduction is
\begin{displaymath}
    \ten{R}(\ten{X}_{\mathrm{piece}}) = \mat{\beta}_{1,\mathrm{piece}}(\ten{X}_{\mathrm{piece}} - \E\ten{X}_{\mathrm{piece}})\t{\mat{\beta}_{2, \mathrm{piece}}}
\end{displaymath}
which gives the complete $48$ dimensional vectorized reduction by stacking the piece-wise reductions
\begin{displaymath}
    \vec{\ten{R}(\ten{X}})
        = (\vec{\ten{R}(\ten{X}_{\text{white pawn}})}, \ldots, \vec{\ten{R}(\ten{X}_{\text{black king}})})
        = \t{\mat{B}}\vec(\ten{X} - \E\ten{X}).
\end{displaymath}
The second line encodes all the piece-wise reductions in a block diagonal full reduction matrix $\mat{B}$ of dimension $768\times 48$ which is applied to the vectorized 3D tensor $\ten{X}$ combining all the piece components $\ten{X}_{\mathrm{piece}}$ into a single tensor of dimension $8\times 8\times 12$. This is a reduction to $6.25\%$ of the original dimension. The $R^2$ statistic of the GAM fitted on $10^5$ new reduced samples is $R^2_{\mathrm{gam}}\approx 46\%$. A linear model on the reduced data achieves $R^2_{\mathrm{lm}}\approx 26\%$ which clearly shows the non-linear relation. On the other hand, the static evaluation of the \emph{Schach H\"ornchen}\footnote{Main author's chess engine.} engine, given the full position (\emph{not} reduced), achieves an $R^2_{\mathrm{hce}}\approx 52\%$. The $42\%$ are reasonably well compared to $51\%$ of the engine static evaluation which gets the original position and uses chess specific expert knowledge. Features the static evaluation includes, which are expected to be learned by the GMLM mixture model, are; \emph{material} (piece values) and \emph{piece square tables} (PSQT, preferred piece type positions). In addition, the static evaluation includes chess specific features like \emph{king safety}, \emph{pawn structure}, or \emph{rooks on open files}. This lets us conclude that the reduction captures most of the relevant features possible, given the oversimplified modeling we performed.

\subsubsection{Interpretation}
For a compact interpretation of the estimated reduction we construct PSQTs. To do so we use the linear model from the validation section. Then, we rewrite the combined linear reduction and linear model in terms of PSQTs. Let $\mat{B}$ be the $768\times 48$ full vectorized linear reduction. This is the block diagonal matrix with the $64\times 4$ dimensional per piece reductions $\mat{B}_{\mathrm{piece}} = \mat{\beta}^{\mathrm{piece}}_2\otimes\mat{\beta}^{\mathrm{piece}}_1$. Then, the linear model with coefficients $\mat{b}$ and intercept $a$ on the reduced data is given by
\begin{equation}\label{eq:chess-lm}
    y = a + \t{\mat{b}}\t{\mat{B}}\vec(\ten{X} - \E\ten{X}) + \epsilon
\end{equation}
with an unknown mean zero error term $\epsilon$ and treating the binary tensor $\ten{X}$ as continuous. Decomposing the linear model coefficients into blocks of $4$ gives per piece coefficients $\mat{b}_{\mathrm{piece}}$ which combine with the diagonal blocks $\mat{B}_{\mathrm{piece}}$ of $\mat{B}$ only. Rewriting \eqref{eq:chess-lm} gives
\begin{align*}
    y &= a + \sum_{\mathrm{piece}}\t{(\mat{B}_{\mathrm{piece}}\mat{b}_{\mathrm{piece}})}\vec(\ten{X}_{\mathrm{piece}} - \E\ten{X}_{\mathrm{piece}}) + \epsilon \\
    &= \tilde{a} + \sum_{\mathrm{piece}}\langle
        \mat{B}_{\mathrm{piece}}\mat{b}_{\mathrm{piece}},
        \vec(\ten{X}_{\mathrm{piece}})
    \rangle + \epsilon
\end{align*}
with a new intercept term $\tilde{a}$, which is of no interest to us. Finally, we enforce color symmetry, using known mechanisms from chess engines. Specifically, mirroring the position changes the sign of the score $y$. Here, mirroring reverses the rank (row) order, this is the image in a mirror behind a chess board. Let for every $\mat{C}_{\mathrm{piece}}$ be a $8\times 8$ matrix with elements $(\mat{C}_{\mathrm{piece}})_{i j} = (\mat{B}_{\mathrm{piece}}\mat{b}_{\mathrm{piece}})_{i + 8 (j - 1)}$. And denote with $\mat{M}(\mat{A})$ the matrix mirror operation which reverses the row order of a matrix. Using this new notation allows enforcing this symmetry, leading to the new approximate linear relation
\begin{align*}
    y &= \tilde{a} + \sum_{\mathrm{piece}}\langle
        \mat{C}_{\mathrm{piece}},
        \ten{X}_{\mathrm{piece}}
        \rangle + \epsilon \\
    &\approx \tilde{a} + \sum_{\text{piece type}}\frac{1}{2}\langle
        \mat{C}_{\text{white piece}} - \mat{M}(\mat{C}_{\text{black piece}}),
        \ten{X}_{\text{white piece}} - \mat{M}(\ten{X}_{\text{white piece}})
    \rangle + \epsilon
\end{align*}
If for every piece type ($6$ types, \emph{not} distinguishing between color) holds $\mat{C}_{\text{white piece}} = -\mat{M}(\mat{C}_{\text{black piece}})$, then we have equality. In our case, this is valid given that the estimates $\hat{\mat{C}}_{\mathrm{piece}}$ fulfill this property with a small error. The $6$ matrices $(\mat{C}_{\text{white piece}} - \mat{M}(\mat{C}_{\text{black piece}})) / 2$ are called \emph{piece square tables} (PSQT) which are visualized in \cref{fig:psqt}. The interpretation of those tables is straightforward. A high positive value (blue) means that it is usually good to have a piece of the corresponding type on that square while a high negative value (red) means the opposite. It needs to be considered that the PSQTs are for quiet positions only, which means all pieces are save in the sense that there is no legal capturing moves nor is the king in check.

\begin{figure}[t]
    \centering
    \includegraphics[scale=0.8]{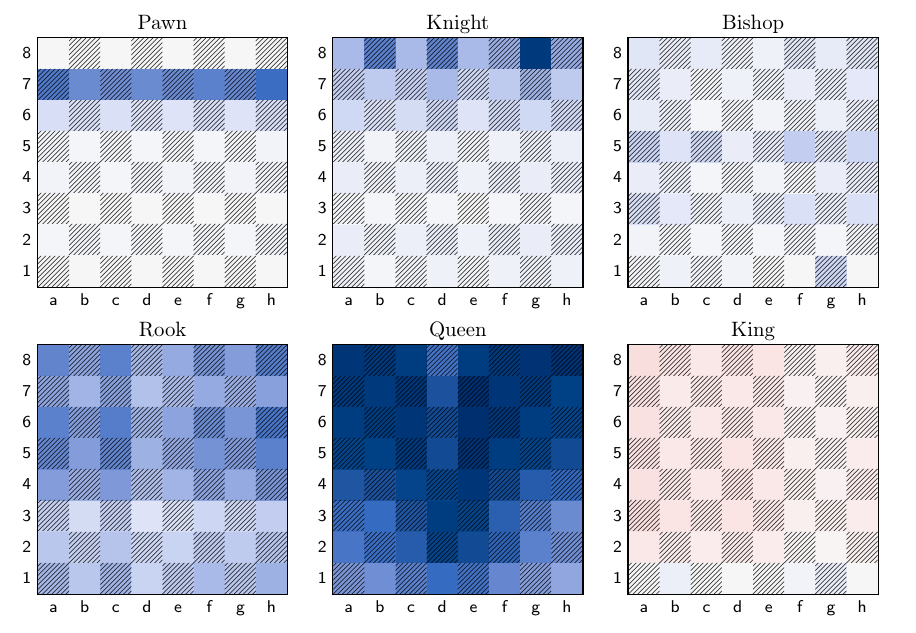}
    \caption{\label{fig:psqt}Extracted PSQTs (piece square tables) from the chess example GMLM reduction.}
\end{figure}

The first visual effect in \cref{fig:psqt} is the dark blue PSQT of the Queen followed by a not-so-dark Rook PSQT. This indicates that the Queen, followed by the Rook, are the most valuable pieces (after the king which is the most valuable, which also implies that assigning value to the king makes no sense). The next two are the Knight and Bishop which have higher value than the Pawns, ignoring the $6$th and $7$th rank as this makes the pawns potential queens. This is the classic piece value order known in chess. 

Next, going over the PSQTs one by one, a few words about the preferred positions for every piece type. The pawn positions are specifically good on the $6$th and especially on the $7$th rank as this threatens a promotion to a Queen (or Knight, Bishop, Rook). The Knight PSQT is a bit surprising, the most likely explanation for the knight being good in the enemy territory is that it got there by capturing an enemy piece for free. A common occurrence in low-rated games is a big chunk of the training data, ranging over all levels. The Bishops seem to have no specific preferred placement, only a slightly higher overall value than pawns, excluding pawns imminent for a promotion. Continuing with the rooks, we see that the rook is a good attacking piece, indicated by a save rook infiltration. \footnote{Rook infiltration is a strategic concept in chess that involves skillfully maneuvering your rook to penetrate deep into your opponent's territory.} The Queen is powerful almost everywhere, only the outer back rank squares (lower left and right) tend to reduce her value. This is rooted in the queen's presence there being a sign for being chased by enemy pieces. Leading to a lot of squares being controlled by the enemy hindering one own movement. Finally, given the goal of the game is to checkmate the king, a safe position for the king is very valuable. This is seen by the back rank (rank $1$) being the only non-penalized squares. Furthermore, the safest squares are the castling \footnote{Castling is a maneuver that combines king safety with rook activation.} target squares ($g1$ and $c1$) as well as the $b1$ square. Shifting the king over to $b1$ is quite common to protect the $a2$ pawn so that the entire pawn shield in front of the king is protected.

The results of our analysis in the previous paragraph agree with the configuration of the chess board most associated with observed chess game outcomes. This arrangement also aligns with the understanding of human chess players of an average configuration at any moment during the game.

\section{Discussion}\label{sec:discussion}

In this paper, we propose a generalized multi-linear model formulation for the inverse conditional distribution of a tensor-valued predictor given a response and derive a multi-linear sufficient reduction for the corresponding forward regression/classification problem. We also propose estimators for the sufficient reduction and show they are consistent and asymptotically normal. Obtaining the asymptotic results required leveraging manifolds as a basis for resolving the issue of unidentifiable parameters. This in turn led to an even more flexible modeling framework, which allows building complex and potentially problem-specific parameter spaces that incorporate additional domain-specific knowledge into the model.

We allude to this feature of our approach in \cref{sec:matrix-manifolds}, where we also tabulate different matrix manifolds that can be used as building blocks $\manifold{B}_k$ and $\manifold{O}_k$ of the parameter space in \cref{tab:matrix-manifolds}. For example, our formulation can easily accommodate longitudinal data tabulated in matrix format, where the rows are covariates and the columns are consecutive time points with discrete AR($k$) dependence structure.

Our multi-linear Ising model can be thought of as the extension of the Ising model-based approach of \cite{ChengEtAl2014}, where a $q$-dimensional binary vector is regressed on a $p$-dimensional continuous vector. Yet, our model leverages the inherent structural information of the tensor-valued covariates by assuming separable first and second moments. By doing so, it bypasses requiring sparsity assumptions or penalization, despite the tensor high-dimensional nature of the data. Moreover, it can accommodate a mixture of continuous and binary tensor-valued predictors, which is the subject of future work.

A special case of the Ising model is the one-parameter Ising model without an external field. Given our theory, it is possible to represent this special case as well by constructing a one-dimensional parameter manifold. More precisely, if $\mat{A}$ is a fixed symmetric $p\times p$ matrix with zero main diagonal elements, the parameter space for the one-dimensional Ising model has the form $\Theta = \{\mat{0}\}\times \{ \theta\mat{A} + 0.5\mat{I}_p : \theta\in\mathbb{R} \}$. As a result (a) the model is a vector-valued multi-variate binary model for predictors $\mat{X}\in\{0, 1\}^p$, and (b) the predictors $\mat{X}$ are independent of the response $y$. In the unsupervised setting, the usual question is to determine the single parameter $\theta$. This is the model considered by \cite{XuMukherjee2023}, whose finding that the asymptotic distribution of the MLE depends on the true value of the parameter and is not necessarily normal, at first glance, may seem contradictory to our asymptotic result that the MLE $\hat{\theta}_n$ is asymptotically normal. \cite{XuMukherjee2023} answer a different question in that they determine the limiting experiment distribution \cite[][Ch~9]{vanderVaart1998}, where they let the dimension $p$ of the binary vector-valued predictors go to infinity ($p\to\infty$). We, on the other hand, let the sample size; i.e., the number of observations go to infinity ($n\to\infty$).

An additional powerful extension of our model involves considering a sum of separable Kronecker predictors. This is motivated by the equivalence of a Kronecker product to a rank $1$ tensor. By allowing a sum of a few separable Kronecker predictors, we remove the implicit rank $1$ constraint. However, if this extension is to be applied to the SDR setting, as in this paper, it is crucial to ensure that the sum of Kronecker products form a parameter manifold.

\printbibliography[title={References}]

\newpage
\appendix

\section{Proofs}\label{app:proofs}

\begin{proof}[Proof of \cref{thm:sdr}]\label{proof:sdr}
    A direct implication of Theorem~1 from \cite{BuraDuarteForzani2016} is that, under the exponential family \eqref{eq:quadratic-exp-fam} with natural statistic \eqref{eq:t-stat},
    \begin{displaymath}
        \t{\mat{\alpha}}(\mat{t}(\ten{X}) - \E\mat{t}(\ten{X}))
    \end{displaymath}
    is a sufficient reduction, where $\mat{\alpha}\in\mathbb{R}^{(p + d)\times q}$ with $\Span(\mat{\alpha}) = \Span(\{\mat{\eta}_y - \E_{Y}\mat{\eta}_Y : y\in\mathcal{S}_Y\})$. Since $\E_Y\ten{F}_Y = 0$, $\E_Y\mat{\eta}_{1 Y} = \E[\vec\overline{\ten{\eta}} - \mat{B}\vec\ten{F}_Y] = \vec\overline{\ten{\eta}}$. Thus,
    \begin{displaymath}
        \mat{\eta}_y - \E_{Y}\mat{\eta}_Y = \begin{pmatrix}
            \mat{\eta}_{1 y} - \E_{Y}\mat{\eta}_{1 Y} \\
            \mat{\eta}_{2} - \E_{Y}\mat{\eta}_{2}
        \end{pmatrix} = \begin{pmatrix}
            \mat{B}\vec\ten{F}_y \\
            \mat{0}
        \end{pmatrix}.
    \end{displaymath}
as $\mat{\eta}_{2}$ does not depend on $y$. 
The set $\{ \vec{\ten{F}_y} : y\in\mathcal{S}_y \}$ is a subset of $\mathbb{R}^q$. Therefore,
    \begin{displaymath}
        \Span\left(\{\mat{\eta}_y - \E_{Y}\mat{\eta}_Y : y\in\mathcal{S}_Y\}\right) = \Span\left(\left\{\begin{pmatrix}
            \mat{B}\vec\ten{F}_Y \\ \mat{0}
        \end{pmatrix} : y\in\mathcal{S}_Y \right\}\right)
        \subseteq
        \Span\begin{pmatrix}
            \mat{B} \\ \mat{0}
        \end{pmatrix},
    \end{displaymath}
which obtains that 
    \begin{displaymath}
        \t{\begin{pmatrix}
            \mat{B} \\ \mat{0}
        \end{pmatrix}}(\mat{t}(\ten{X}) - \E\mat{t}(\ten{X}))
        =
        \t{\mat{B}}\vec(\ten{X} - \E\ten{X})
        = \vec\Bigl(\ten{F}_y\mlm_{k = 1}^{r}\mat{\beta}_k\Bigr)
    \end{displaymath}
is also a sufficient reduction, though not necessarily minimal, using $\mat{B} = \bigkron_{k = 1}^{r}\mat{\beta}_k$. When the exponential family is full rank, which in our setting amounts to all $\mat{\beta}_j$ being full rank matrices, $j=1,\ldots,r$, then Theorem~1 from \cite{BuraDuarteForzani2016} also obtains the minimality of the reduction.
\end{proof}

\begin{proof}[Proof of \cref{thm:grad}]\label{proof:grad}
    We first note that for any exponential family with density \eqref{eq:quad-density} the term $b(\mat{\eta}_{y_i})$ differentiated with respect to the natural parameter $\mat{\eta}_{y_i}$ is the expectation of the statistic $\mat{t}(\ten{X})$ given $Y = y_i$. In our case we get $\nabla_{\mat{\eta}_{y_i}}b = (\nabla_{\mat{\eta}_{1{y_i}}}b, \nabla_{\mat{\eta}_2}b)$ with components
    \begin{displaymath}
        \nabla_{\mat{\eta}_{1{y_i}}}b
            = \E[\mat{t}_1(\ten{X})\mid Y = y_i]
            = \vec\E[\ten{X}\mid Y = y_i]
            = \vec\ten{g}_1(\mat{\eta}_{y_i})
    \end{displaymath}
    and
    \begin{align*}
        \nabla_{\mat{\eta}_{2}}b
           &= \E[\mat{t}_2(\ten{X})\mid Y = y_i]
            = \E[\mat{T}_2\vech((\vec\ten{X})\t{(\vec\ten{X})})\mid Y = y_i] \\
           &= \E[\mat{T}_2\pinv{\mat{D}_p}\vec(\ten{X}\circ\ten{X})\mid Y = y_i]
            = \mat{T}_2\pinv{\mat{D}_p}\vec\ten{g}_2(\mat{\eta}_{y_i}).
    \end{align*}
    The gradients are related to their derivatives by transposition, $\nabla_{\mat{\eta}_{1{y_i}}}b = \t{\D b(\mat{\eta}_{1{y_i}})}$ and $\nabla_{\mat{\eta}_2}b = \t{\D b(\mat{\eta}_2)}$.
    Next, we provide the differentials of the natural parameter components from \eqref{eq:eta1} and \eqref{eq:eta2} in a quite direct form, without any further ``simplifications,'' because the down-stream computations will not benefit from re-expressing the following
    \begin{align*}
        \d\mat{\eta}_{1{y_i}}(\overline{\ten{\eta}})
            &= \d\vec{\overline{\ten{\eta}}}, \\
        \d\mat{\eta}_{1{y_i}}(\mat{\beta}_j)
            &= \vec\Bigl( \ten{F}_{y_i}\mlm_{\substack{k = 1\\k\neq j}}^{r}\mat{\beta}_k\times_j\d\mat{\beta}_j \Bigr), \\
        \d\mat{\eta}_2(\mat{\Omega}_j)
            &= \t{(\pinv{(\mat{T}_2\pinv{\mat{D}_p})})}\vec(c\,\d\mat{\Omega}) \\
            &= c\t{(\pinv{(\mat{T}_2\pinv{\mat{D}_p})})}\vec\Bigl(\,\bigkron_{k = r}^{j + 1}\mat{\Omega}_k\otimes\d\mat{\Omega}_j\otimes\bigkron_{l = j - 1}^{1}\mat{\Omega}_l \Bigr).
    \end{align*}
    All other combinations, namely $\d\mat{\eta}_{1{y_i}}(\mat{\Omega}_j)$, $\d\mat{\eta}_2(\overline{\ten{\eta}})$ and $\d\mat{\eta}_2(\mat{\beta}_j)$, are zero.
    Continuing with the partial differentials of $l_n$ from \eqref{eq:log-likelihood}
    \begin{multline*}
        \d l_n(\overline{\ten{\eta}})
            = \sum_{i = 1}^{n} (\langle \d\overline{\ten{\eta}}, \ten{X}_i \rangle - \D b(\mat{\eta}_{1{y_i}})\d\mat{\eta}_{1{y_i}}(\overline{\ten{\eta}}))
            = \sum_{i = 1}^{n} \t{(\vec{\ten{X}_i} - \vec\ten{g}_1(\mat{\eta}_{y_i}))}\d\vec{\overline{\ten{\eta}}} \\
            = \t{(\d\vec{\overline{\ten{\eta}}})}\vec\sum_{i = 1}^{n} (\ten{X}_i - \ten{g}_1(\mat{\eta}_{y_i})).
    \end{multline*}
    For every $j = 1, ..., r$ we get the differentials
    \begin{multline*}
        \d l_n(\mat{\beta}_j)
            = \sum_{i = 1}^{n} \biggl(\Bigl\langle \ten{F}_{y_i}\mlm_{\substack{k = 1\\k\neq j}}^{r}\mat{\beta}_k\times_j\d\mat{\beta}_j, \ten{X}_i \Bigr\rangle - \D b(\mat{\eta}_{1{y_i}})\d\mat{\eta}_{1{y_i}}(\mat{\beta}_j)\biggr) \\
            = \sum_{i = 1}^{n} \Bigl\langle \ten{F}_{y_i}\mlm_{\substack{k = 1\\k\neq j}}^{r}\mat{\beta}_k\times_j\d\mat{\beta}_j, \ten{X}_i - \ten{g}_1(\mat{\eta}_{y_i}) \Bigr\rangle
            = \sum_{i = 1}^{n} \tr\biggl( \d\mat{\beta}_j\Bigl(\ten{F}_{y_i}\mlm_{\substack{k = 1\\k\neq j}}^{r}\mat{\beta}_k\Bigr)_{(j)} \t{(\ten{X}_i - \ten{g}_1(\mat{\eta}_{y_i}))_{(j)}} \biggr) \\
            = \t{(\d\vec{\mat{\beta}_j})}\vec\sum_{i = 1}^{n} (\ten{X}_i - \ten{g}_1(\mat{\eta}_{y_i}))_{(j)} \t{\Bigl(\ten{F}_{y_i}\mlm_{\substack{k = 1\\k\neq j}}^{r}\mat{\beta}_k\Bigr)_{(j)}}
    \end{multline*}
    as well as
    {\allowdisplaybreaks
    \begin{align*}
        \d l_n(\mat{\Omega}_j)
            &= \sum_{i = 1}^{n} \biggl( c\Bigl\langle \ten{X}_i\mlm_{\substack{k = 1\\k\neq j}}^{r}\mat{\Omega}_k\times_j\d\mat{\Omega}_j, \ten{X}_i \Bigr\rangle - \D b(\mat{\eta}_2)\d\mat{\eta}_2(\mat{\Omega}_j) \biggr) \\
            &= c\sum_{i = 1}^{n} \biggl( \Bigl\langle \ten{X}_i\mlm_{\substack{k = 1\\k\neq j}}^{r}\mat{\Omega}_k\times_j\d\mat{\Omega}_j, \ten{X}_i \Bigr\rangle \\
                &\qquad\qquad - \t{(\mat{T}_2\pinv{\mat{D}_p}\vec\ten{g}_2(\mat{\eta}_{y_i}))}\t{(\pinv{(\mat{T}_2\pinv{\mat{D}_p})})}\vec\Bigl(\,\bigkron_{k = r}^{j + 1}\mat{\Omega}_k\otimes\d\mat{\Omega}_j\otimes\bigkron_{l = j - 1}^{1}\mat{\Omega}_l \Bigr) \biggr) \\
            &= c\sum_{i = 1}^{n} \biggl( \Bigl\langle \ten{X}_i\mlm_{\substack{k = 1\\k\neq j}}^{r}\mat{\Omega}_k\times_j\d\mat{\Omega}_j, \ten{X}_i \Bigr\rangle - \t{(\vec\ten{G}_2(\mat{\eta}_{y_i}))}\vec\Bigl(\,\bigkron_{k = r}^{j + 1}\mat{\Omega}_k\otimes\d\mat{\Omega}_j\otimes\bigkron_{l = j - 1}^{1}\mat{\Omega}_l \Bigr) \biggr) \\
            &= c\sum_{i = 1}^{n} \biggl( \t{\vec(\ten{X}_i\circ\ten{X}_i - \ten{G}_2(\mat{\eta}_{y_i}))}\vec\Bigl(\,\bigkron_{k = r}^{j + 1}\mat{\Omega}_k\otimes\d\mat{\Omega}_j\otimes\bigkron_{l = j - 1}^{1}\mat{\Omega}_l \Bigr) \biggr) \\
            &= c\sum_{i = 1}^{n} \K(\ten{X}_i\circ\ten{X}_i - \ten{G}_2(\mat{\eta}_{y_i}))\mlm_{\substack{k = 1\\k\neq j}}^{r}\t{(\vec{\mat{\Omega}_k})}\times_j\t{(\d\vec{\mat{\Omega}_j})} \\
            &= c\t{(\d\vec{\mat{\Omega}_j})}\sum_{i = 1}^{n} \Bigl((\ten{X}_i\otimes\ten{X}_i - \K(\ten{G}_2(\mat{\eta}_{y_i})))\mlm_{\substack{k = 1\\k\neq j}}^{r}\t{(\vec{\mat{\Omega}_k})}\Bigr)_{(j)} \\
            &= c\t{(\d\vec{\mat{\Omega}_j})}\vec\sum_{i = 1}^{n} (\ten{X}_i\otimes\ten{X}_i - \K(\ten{G}_2(\mat{\eta}_{y_i})))\mlm_{\substack{k = 1\\k\neq j}}^{r}\t{(\vec{\mat{\Omega}_k})}
    \end{align*}}
    Now, applying the identity $\d \ten{A}(\ten{B}) = \t{(\d\vec{\ten{B}})}\nabla_{\ten{B}}\ten{A}$ gives the required partial gradients.

    Finally, if $\mat{T}_2$ is the identify matrix, then
    \begin{displaymath}
        \vec{\ten{G}_2(\mat{\eta}_y)} = \pinv{(\mat{T}_2\pinv{\mat{D}_p})}\mat{T}_2\pinv{\mat{D}_p}\vec{\ten{g}_2(\mat{\eta}_y)}
        = \mat{D}_p\pinv{\mat{D}_p}\vec{\ten{g}_2(\mat{\eta}_y)}
        = \vec{\ten{g}_2(\mat{\eta}_y)}
    \end{displaymath}
    where the last equality holds because $\mat{N}_p = \mat{D}_p\pinv{\mat{D}_p}$ is the symmetrizer matrix \cite[][Ch.~11]{AbadirMagnus2005}. The symmetrizer matrix $\mat{N}_p$ satisfies $\mat{N}_p\vec{\mat{A}} = \vec{\mat{A}}$ if $\mat{A} = \t{\mat{A}}$, and
    \begin{displaymath}
        \vec{\ten{g}_2(\mat{\eta}_y)} = \vec\E[\ten{X}\circ\ten{X}\mid Y = y] = \vec\E[(\vec{\ten{X}})\t{(\vec{\ten{X}})}\mid Y = y]
    \end{displaymath}
    is the vectorization of a symmetric matrix.
\end{proof}

\begin{proof}[Proof of \cref{thm:kron-manifolds}]\label{proof:kron-manifolds}
    We start by considering the first case and assume that $\manifold{B}$ is spherical with radius $1$ w.l.o.g. We equip $\manifold{K} = \{ \mat{A}\otimes \mat{B} : \mat{A}\in\manifold{A}, \mat{B}\in\manifold{B} \}\subset\mathbb{R}^{p_1 p_2\times q_1 q_2}$ with the subspace topology.\footnote{Given a topological space $(\manifold{M}, \mathcal{T})$ and a subset $\mathcal{S}\subseteq\manifold{M}$ the \emph{subspace topology} on $\manifold{S}$ induced by $\manifold{M}$ consists of all open sets in $\manifold{M}$ intersected with $\manifold{S}$, that is $\{ O\cap\manifold{S} : O\in\mathcal{T} \}$ \cite[]{Lee2012,Lee2018,Kaltenbaeck2021}.} Define the hemispheres $H_i^{+} = \{ \mat{B}\in\manifold{B} : (\vec{\mat{B}})_i > 0 \}$ and $H_i^{-} = \{ \mat{B}\in\manifold{B} : (\vec{\mat{B}})_i < 0 \}$ for $i = 1, ..., p_2 q_2$. The hemispheres are an open cover of $\manifold{B}$ with respect to the subspace topology. Define for every $H_i^{\pm}$, where $\pm$ is a placeholder for ether $+$ or $-$, the function
    \begin{displaymath}
        f_{H_i^{\pm}} : \manifold{A}\times H_i^{\pm}\to\mathbb{R}^{p_1 p_2\times q_1 q_2}
            : (\mat{A}, \mat{B})\mapsto \mat{A}\otimes \mat{B}
    \end{displaymath}
    which is smooth. With the spherical property of $\manifold{B}$ the relation $\|\mat{A}\otimes \mat{B}\|_F = \|\mat{A}\|_F$ for all $\mat{A}\otimes \mat{B}\in\manifold{K}$ ensures that the function $f_{H_i^{\pm}}$, constrained to its image, is bijective with inverse function (identifying $\mathbb{R}^{p\times q}$ with $\mathbb{R}^{p q}$) given by
    \begin{displaymath}
        f_{H_i^{\pm}}^{-1} = \begin{cases}
            f_{H_i^{\pm}}(\manifold{A}\times H_i^{\pm})\to\manifold{A}\times H_i^{\pm} \\
            \mat{C}\mapsto \left(\pm\frac{\|\mat{C}\|_F}{\|\mat{R}(\mat{C})\mat{e}_i\|_2}\mat{R}(\mat{C})\mat{e}_i, \pm\frac{1}{\|\mat{C}\|_F\|\mat{R}(\mat{C})\mat{e}_i\|_2}\mat{R}(\mat{C})\t{\mat{R}(\mat{C})}\mat{e}_i\right)
        \end{cases}
    \end{displaymath}
    where $\pm$ is $+$ for a ``positive'' hemisphere $H_i^{+}$ and $-$ otherwise, $\mat{e}_i\in\mathbb{R}^{p_2 q_2}$ is the $i$th unit vector and $\mat{R}(\mat{C})$ is a ``reshaping'' permutation\footnote{Relating to $\K$ the operation $\mat{R}$ is basically its inverse as $\K(\mat{A}\circ\mat{B}) = \mat{A}\otimes\mat{B}$ with a mismatch in the shapes only.} which acts on Kronecker products as $\mat{R}(\mat{A}\otimes \mat{B}) = (\vec{\mat{A}})\t{(\vec{\mat{B}})}$. This makes $f_{H_i^{\pm}}^{-1}$ a combination of smooth functions ($\mat{0}$ is excluded from $\manifold{A}, \manifold{B}$ guarding against division by zero) and as such it is also smooth. This ensures that $f_{H_i^{\pm}} : \manifold{A}\times {H_i^{\pm}}\to f_{H_i^{\pm}}(\manifold{A}\times {H_i^{\pm}})$ is a diffeomorphism.

    Next, we construct an atlas\footnote{A collection of charts $\{ \varphi_i : i\in I \}$ with index set $I$ of a manifold $\manifold{A}$ is called an \emph{atlas} if the pre-images of the charts $\varphi_i$ cover the entire manifold $\manifold{A}$.} for $\manifold{K}$ which is equipped with the subspace topology. Let $(\varphi_j, U_j)_{j\in J}$ be a atlas of $\manifold{A}\times\manifold{B}$. Such an atlas exists and admits a unique smooth structure as both $\manifold{A}, \manifold{B}$ are embedded manifolds from which we take the product manifold. The images of the coordinate domains $f_H(U_j)$ are open in $\manifold{K}$, since $f_H$ is a diffeomorphism, with the corresponding coordinate maps
    \begin{displaymath}
        \phi_{H_i^{\pm},j} : f_{H_i^{\pm}}(U_j)\to \varphi_j(U_j)
            : \mat{C}\mapsto \varphi_j(f_{H_i^{\pm}}^{-1}(\mat{C})).
    \end{displaymath}
    By construction the set $\{ \phi_{H_i^{\pm},j} : i = 1, ..., p_2 q_2, \pm\in\{+, -\}, j\in J \}$ is an atlas if the charts are compatible. This means we need to check if the transition maps are diffeomorphisms, let $(\phi_{H, j}, V_j), (\phi_{\widetilde{H}, k}, V_k)$ be two arbitrary charts from our atlas, then the transition map $\phi_{\widetilde{H}, k}\circ\phi_{H,j}^{-1}:\phi_{H,j}^{-1}(V_j\cap V_k)\to\phi_{\widetilde{H},k}^{-1}(V_j\cap V_k)$ has the form
    \begin{displaymath}
        \phi_{\widetilde{H}, k}\circ\phi_{H,j}^{-1}
            = \varphi_k\circ f_{\widetilde{H}}^{-1}\circ f_{H}\circ\varphi_{j}^{-1}
            = \varphi_k\circ (\pm\mathrm{id})\circ\varphi_{j}^{-1}
    \end{displaymath}
    where $\pm$ depends on $H, \widetilde{H}$ being of the same ``sign'' and $\mathrm{id}$ is the identity. We conclude that the charts are compatible, which makes the constructed set of charts an atlas. With that we have shown the topological manifold $\manifold{K}$ with the subspace topology admit a smooth atlas that makes it an embedded smooth manifold with dimension equal to the dimension of the product topology $\manifold{A}\times\manifold{B}$; that is, $d = \dim\manifold{A} + \dim\manifold{B}$.

    It remains to show that the cone condition also admits a smooth manifold. $\manifold{K} = \{ \mat{A}\otimes \mat{B} : \mat{A}\in\manifold{A}, \mat{B}\in\widetilde{B} \}$, where $\widetilde{B} = \{ \mat{B}\in\manifold{B} : \|\mat{B}\|_F = 1 \}$, holds if both $\manifold{A}, \manifold{B}$ are cones. Since $g:\manifold{B}\to\mathbb{R}:\mat{B}\mapsto \|\mat{B}\|_F$ is continuous on $\manifold{B}$ with full rank $1$ everywhere, $\widetilde{\manifold{B}} = g^{-1}(1)$ is a $\dim{\manifold{B}} - 1$ dimensional embedded submanifold of $\manifold{B}$. An application of the spherical case proves the cone case.
\end{proof}

\begin{proof}[Proof of \cref{thm:param-manifold}]
    An application of \cref{thm:kron-manifolds} ensures that $\manifold{K}_{\mat{B}}$ and $\manifold{K}_{\mat{\Omega}}$ are embedded submanifolds.

    With $\mat{T}_2$ being a $d\times p(p + 1) / 2$ full rank matrix and the duplication matrix $\mat{D}_p$ is full rank of dimension $p(p + 1) / 2 \times p^2$ we have $\mat{T}_2\pinv{\mat{D}_p}$ to be $d\times p^2$ of full rank. This means that $\mat{P} = \pinv{(\mat{T}_2\pinv{\mat{D}_p})}\mat{T}_2\pinv{\mat{D}_p}$ is a $p^2\times p^2$ projection of rank $d$ and $\mat{I}_{p^2} - \mat{P}$ is then a projection of rank $p^2 - d$. This leads to
    \begin{displaymath}
        \manifold{CK}_{\mat{\Omega}}
            = \{ \mat{\Omega}\in\manifold{K}_{\mat{\Omega}} : (\mat{I}_{p^2} - \mat{P})\vec{\mat{\Omega}} = \mat{0} \}.
    \end{displaymath}
    To show that $\manifold{CK}_{\mat{\Omega}}$ is an embedded submanifold of $\manifold{K}_{\mat{\Omega}}$ we apply the ``Constant-Rank Level Set Theorem'' \cite[][Theorem~5.12]{Lee2012} which states (slightly adapted) the following; Let $\manifold{A}$, $\manifold{B}$ be smooth manifolds and $F:\manifold{A}\to\manifold{B}$ a smooth map such that $\nabla_{\mat{a}} F$ has constant matrix rank for all $\mat{a}\in\manifold{A}$. Then, for every $\mat{b}\in F(\mat{A})\subseteq\manifold{B}$, the preimage $F^{-1}(\{ \mat{b} \})$ is a smooth embedded submanifold of $\manifold{A}$.

    In our setting, we have $F:\manifold{K}_{\mat{\Omega}}\to\mathbb{R}^{p^2}$ defined as $F(\mat{\Omega}) = (\mat{I}_{p^2} - \mat{P})\vec{\mat{\Omega}}$ with gradient $\nabla_{\mat{\Omega}} F = \mat{I}_{p^2} - \mat{P}$ of constant rank. Therefore, $F^{-1}(\{\mat{0}\}) = \manifold{CK}_{\mat{\Omega}}$ is an embedded submanifold of $\manifold{K}_{\mat{\Omega}}$.

    Finally, the finite product manifold of embedded submanifolds is embedded in the finite product space of their ambient spaces, that is $\Theta = \mathbb{R}^p \times \manifold{K}_{\mat{B}}\times\manifold{CK}_{\mat{\Omega}} \subset \mathbb{R}^p\times\mathbb{R}^{p\times q}\times\mathbb{R}^{p\times p}$ is embedded.
\end{proof}

\begin{proof}[Proof of \cref{thm:exists-strong-M-estimator-on-subsets}]
    Let $\hat{\mat{\xi}}_n$ be a (weak/strong) M-estimator for the unconstrained problem. This gives by definition, in any case, that
    \begin{displaymath}
        \sup_{\mat{\xi}\in\Xi} M_n(\mat{\xi}) \leq M_n(\hat{\mat{\xi}}_n) + o_P(1).
    \end{displaymath}
    Since $\emptyset\neq\Theta\subseteq\Xi$, $\sup_{\mat{\theta}\in\Theta} M_n(\mat{\theta}) \leq \sup_{\mat{\xi}\in\Xi} M_n(\mat{\xi})$ and with $M_n(\mat{\xi}) < \infty$ for any $\mat{\xi}\in\Xi$
    \begin{displaymath}
        P\Bigl( \sup_{\mat{\theta}\in\Theta} M_n(\mat{\theta}) < \infty \Bigr)
        \geq
        P\Bigl( \sup_{\mat{\xi}\in\Xi} M_n(\mat{\xi}) < \infty \Bigr)
        \xrightarrow{n\to\infty}
        1.
    \end{displaymath}
    If $\sup_{\mat{\theta}\in\Theta} M_n(\mat{\theta}) < \infty$, then, for any $0 < \epsilon_n$ there exists $\hat{\mat{\theta}}_n\in\Theta$ such that $\sup_{\mat{\theta}\in\Theta} M_n(\mat{\theta}) - \epsilon_n \leq M_n(\hat{\mat{\theta}}_n)$. Therefore, we can choose $\epsilon_n\in o(n^{-1})$, which yields
    \begin{displaymath}
        P\Bigl( M_n(\hat{\mat{\theta}}_n) \geq \sup_{\mat{\theta}\in\Theta} M_n(\mat{\theta}) - o(n^{-1}) \Bigr)
        \geq
        P\Bigl( \sup_{\mat{\theta}\in\Theta} M_n(\mat{\theta}) < \infty \Bigr)
        \xrightarrow{n\to\infty}
        1.
    \end{displaymath}
    The last statement is equivalent to
    \begin{displaymath}
        M_n(\hat{\mat{\theta}}_n) \geq \sup_{\mat{\theta}\in\Theta} M_n(\mat{\theta}) - o_P(n^{-1}),
    \end{displaymath}
    which is the definition of $\hat{\mat{\theta}}_n$ being a strong M-estimator over $\Theta$.
\end{proof}

\begin{proof}[Proof of \cref{thm:M-estimator-consistency-on-subsets}]
    The proof follows the proof of Proposition~2.4 in \cite{BuraEtAl2018} with the same assumptions, except that we only require $\Theta$ to be a subset of $\Xi$. This is accounted for by replacing Lemma~2.3 in \cite{BuraEtAl2018} with \cref{thm:exists-strong-M-estimator-on-subsets} to obtain the existence of a strong M-estimator on $\Theta$.
\end{proof}

\begin{proof}[Proof of \cref{thm:M-estimator-asym-normal-on-manifolds}]
    Let $\varphi:U\to\varphi(U)$ be a coordinate chart\footnote{By \cref{def:manifold}, the chart $\varphi : U\to\varphi(U)$ is bi-continuous, infinitely often continuously differentiable, and has a continuously differentiable inverse $\varphi^{-1} : \varphi(U) \to U$. Furthermore, the domain $U$ is open with respect to the trace topology on $\Theta$, which means that there exists an open set $O\subseteq\mathbb{R}^p$ such that $U = \Theta\cap O$.} with $\mat{\theta}_0\in U\subseteq\Theta$. As $\varphi$ is continuous, $\varphi(\hat{\mat{\theta}}_n)\xrightarrow{p}\varphi(\mat{\theta}_0)$ by the continuous mapping theorem on metric spaces \cite[][Theorem~18.11]{vanderVaart1998}, which implies $P(\varphi(\hat{\mat{\theta}}_n)\in\varphi(U))\xrightarrow{n\to\infty}1$.

    The next step is to apply Theorem~5.23 in \cite{vanderVaart1998} to $\hat{\mat{s}}_n = \varphi(\hat{\mat{\theta}}_n)$. Therefore, assume that $\hat{\mat{s}}_n\in\varphi(U)$. Denote with $\mat{s} = \varphi(\mat{\theta})\in\varphi(U)\subseteq\mathbb{R}^d$ the coordinates of the parameter $\mat{\theta}\in U\subseteq\Theta$ of the $d = \dim(\Theta)$ dimensional manifold $\Theta\subseteq\mathbb{R}^p$. Since $\varphi : U\to\varphi(U)$ is bijective, we can express $m_{\mat{\theta}}$ in terms of $\mat{s} = \varphi(\mat{\theta})$ as $m_{\mat{\theta}} = m_{\varphi^{-1}(\mat{s})}$ for every $\mat{\theta}\in U$. Furthermore, we let
    \begin{displaymath}
        M(\mat{\theta}) = \E[m_{\mat{\theta}}(Z)] \qquad\text{and}\qquad M_{\varphi}(\mat{s}) = \E[m_{\varphi^{-1}(\mat{s})}(Z)] = M(\varphi^{-1}(\mat{s})).
    \end{displaymath}

    \begin{figure}[ht!]
        \centering
        \includegraphics{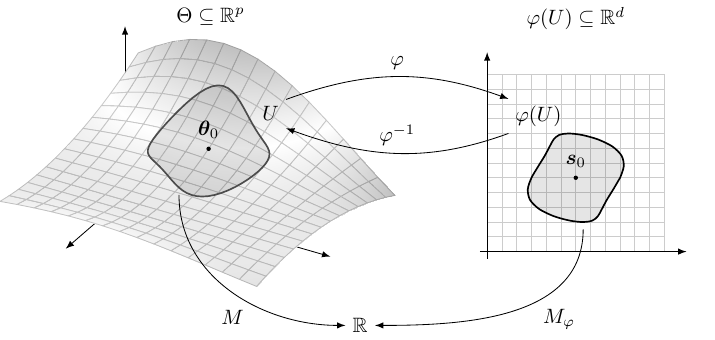}
        \caption{\label{fig:proof:M-estimator-asym-normal-on-manifolds}Depiction of the notation used in the proof of \cref{thm:M-estimator-asym-normal-on-manifolds}. Example with $p = 3$ and $d = \dim(\Theta) = 2$.}
    \end{figure}

    By assumption, the function $M(\mat{\theta})$ is twice continuously differentiable in a neighborhood\footnote{A set $N$ is called a neighborhood of $u$ if there exists an open set $O$ such that $u\in O\subseteq N$.} of $\mat{\theta}_0$. Without loss of generality, we can assume that $U$ is contained in that neighborhood. Then, using the chain rule, we compute the gradient of $M_{\varphi}$ at $\mat{s}_0$ to be $\mat{0}$,
    \begin{displaymath}
        \nabla M_{\varphi}(\mat{s}_0) = {\nabla\varphi^{-1}(\mat{s}_0)}{\nabla M(\varphi^{-1}(\mat{s}_0))} = {\nabla\varphi^{-1}(\mat{s}_0)}{\nabla M(\mat{\theta}_0)} = {\nabla\varphi^{-1}(\mat{s}_0)}\mat{0} = \mat{0}
    \end{displaymath}
    because $\mat{\theta}_0 = \varphi^{-1}(\mat{s}_0)$ is a maximizer of $M$. The second-derivative of $M$, evaluated at $\mat{s}_0 = \varphi(\mat{\theta}_0)$ and using $\nabla M_{\varphi}(\mat{s}_0) = \mat{0}$, is
    \begin{displaymath}
        \nabla^2 M_{\varphi}(\mat{s}_0)
        = \nabla\varphi^{-1}(\mat{s}_0)\nabla^2 M(\varphi^{-1}(\mat{s}_0))\t{\nabla\varphi^{-1}(\mat{s}_0)}
        = \nabla\varphi^{-1}(\mat{s}_0)\mat{H}_{\mat{\theta}_0}\t{\nabla\varphi^{-1}(\mat{s}_0)}.
    \end{displaymath}
    This gives the second-order Taylor expansion of $M_{\varphi}$ at $\mat{s}_0$ as
    \begin{displaymath}
        M_{\varphi}(\mat{s}) = M_{\varphi}(\mat{s}_0) + \frac{1}{2}\t{(\mat{s} - \mat{s}_0)} \nabla^2 M_{\varphi}(\mat{s}_0) (\mat{s} - \mat{s}_0) + \mathcal{O}(\|\mat{s} - \mat{s}_0\|^3)
    \end{displaymath}
    We also need to check the local Lipschitz condition of $m_{\varphi^{-1}(\mat{s})}$. Let $V_{\epsilon}(\mat{s}_0)$ be the open $\epsilon$-ball with center $\mat{s}_0$; i.e., $V_{\epsilon}(\mat{s}_0) = \{ \mat{s}\in\mathbb{R}^d : \|\mat{s} - \mat{s}_0\| < \epsilon \}$. Since $\varphi(U)$ contains $\mat{s}_0$, and is open in $\mathbb{R}^d$, there exists an $\epsilon > 0$ such that $V_{\epsilon}(\mat{s}_0)\subseteq\varphi(U)$. The closed $\epsilon/2$ ball $\overline{V}_{\epsilon / 2}(\mat{s}_0)$ is a neighborhood of $\mat{s}_0$ and $\sup_{\mat{s}\in \overline{V}_{\epsilon / 2}(\mat{s}_0)}\|\nabla\varphi^{-1}(\mat{s})\| < \infty$ due to the continuity of $\nabla\varphi^{-1}$ on $\varphi(U)$ with $\overline{V}_{\epsilon / 2}(\mat{s}_0)\subset V_{\epsilon}(\mat{s}_0)\subseteq\varphi(U)$. Then, for almost every $z$ and every $\mat{s}_1 = \varphi(\mat{\theta}_1), \mat{s}_2 = \varphi(\mat{\theta}_2)\in\overline{V}_{\epsilon / 2}(\mat{s}_0)$,
    \begin{multline*}
        | m_{\varphi^{-1}(\mat{s}_1)}(z) - m_{\varphi^{-1}(\mat{s}_2)}(z) |
        = | m_{\mat{\theta}_1}(z) - m_{\mat{\theta}_2}(z) |
        \overset{(a)}{\leq} u(z) \| \mat{\theta}_1 - \mat{\theta}_2 \| \\
        = u(z) \| \varphi^{-1}(\mat{s}_1) - \varphi^{-1}(\mat{s}_2) \|
        \overset{(b)}{\leq} u(z) \sup_{\mat{s}\in \overline{V}_{\epsilon / 2}(\mat{s}_0)}\|\nabla\varphi^{-1}(\mat{s})\| \|\mat{s}_1 - \mat{s}_2\|
        =: v(z) \|\mat{s}_1 - \mat{s}_2\|.
    \end{multline*}
    Here, $(a)$ holds by assumption and $(b)$ is a result of the mean value theorem. Now, $v(z)$ is measurable and square integrable as a scaled version of $u(z)$. Finally, since $\varphi$ is one-to-one, $\hat{\mat{s}}_n = \varphi(\hat{\mat{\theta}}_n)$ is a strong M-estimator for $\mat{s}_0 = \varphi(\mat{\theta}_0)$ of the objective $M_{\varphi}$. We next apply Theorem~5.23 in \cite{vanderVaart1998} to obtain the asymptotic normality of $\hat{\mat{s}}_n$,
    \begin{displaymath}
        \sqrt{n}(\hat{\mat{s}}_n - \mat{s}_0) \xrightarrow{d} \mathcal{N}_{d}(0, \mat{\Sigma}_{\mat{s}_0}),
    \end{displaymath}
    where the $d\times d$ variance-covariance matrix $\mat{\Sigma}_{\mat{s}_0}$ is given by
    \begin{align*}
        \mat{\Sigma}_{\mat{s}_0} &= (\nabla^2 M_{\varphi}(\mat{s}_0))^{-1}\E[\nabla_{\mat{s}} m_{\varphi^{-1}(\mat{s}_0)}(Z)\t{(\nabla_{\mat{s}} m_{\varphi^{-1}(\mat{s}_0)}(Z))}](\nabla^2 M_{\varphi}(\mat{s}_0))^{-1}.
    \end{align*}
    {%
        \def\PP{\mat{\varPhi}_{\mat{\theta}_0}}%
        \def\EE#1#2{\E[\nabla_{#2} m_{#1}(Z)\t{(\nabla_{#2} m_{#1}(Z))}]}%
        An application of the delta method yields
        \begin{displaymath}
            \sqrt{n}(\hat{\mat{\theta}}_n - \mat{\theta}_0)= \sqrt{n}(\varphi^{-1}(\hat{\mat{s}}_n) - \varphi^{-1}(\mat{s}_0))
            \xrightarrow{d} \mathcal{N}_p(0, \t{\nabla\varphi^{-1}(\mat{s}_0)}\mat{\Sigma}_{\mat{s}_0}{\nabla\varphi^{-1}(\mat{s}_0)}).
        \end{displaymath}
        We continue by reexpressing the $p\times p$ asymptotic variance-covariance matrix of $\hat{\mat{\theta}}_n$ in terms of $\mat{\theta}_0$ instead of $\mat{s}_0 = \varphi(\mat{\theta}_0)$. We let $\PP = \t{\nabla\varphi^{-1}(\varphi(\mat{\theta}_0))} = \t{\nabla\varphi^{-1}(\mat{s}_0)}$ and observe that for all $\mat{s}\in\varphi(U)$, the gradient of $\mat{s}\mapsto m_{\varphi^{-1}(\mat{s})}(z)$ evaluated at $\mat{s}_0 = \varphi(\mat{\theta}_0)$ has the form
        \begin{displaymath}
            \nabla_{\mat{s}}m_{\varphi^{-1}(\mat{s}_0)}(z)= \nabla\varphi^{-1}(\mat{s}_0)\nabla_{\mat{\theta}}m_{\mat{\theta}_0}(z)= \t{\PP}\nabla_{\mat{\theta}}m_{\mat{\theta}_0}(z).
        \end{displaymath}
        Then,
        \begin{multline*}
            \t{\nabla\varphi^{-1}(\mat{s}_0)}\mat{\Sigma}_{\mat{s}_0}{\nabla\varphi^{-1}(\mat{s}_0)}= \PP\mat{\Sigma}_{\mat{s}_0}\t{\PP} \\
            = \PP(\nabla^2 M_{\varphi}(\mat{s}_0))^{-1}\EE{\varphi^{-1}(\mat{s}_0)}{}(\nabla^2 M_{\varphi}(\mat{s}_0))^{-1}\t{\PP} \\
            = {\PP}(\t{\PP}\mat{H}_{\mat{\theta}_0}\PP)^{-1}\t{\PP}\EE{\mat{\theta}_0}{\mat{\theta}}{\PP}(\t{\PP}\mat{H}_{\mat{\theta}_0}\PP)^{-1}\t{\PP} \\
            = \mat{\Pi}_{\mat{\theta}_0}\EE{\mat{\theta}_0}{\mat{\theta}}\mat{\Pi}_{\mat{\theta}_0}
        \end{multline*}
        where the last equality holds because $\Span\PP = T_{\mat{\theta}_0}\Theta$, by \cref{def:tangent-space} of the tangent space $T_{\mat{\theta}_0}\Theta$.

        It remains to show that $\mat{\Pi}_{\mat{\theta}_0} = \mat{P}_{\mat{\theta}_0}\pinv{(\t{\mat{P}_{\mat{\theta}_0}}\mat{H}_{\mat{\theta}_0}\mat{P}_{\mat{\theta}_0})}\t{\mat{P}_{\mat{\theta}_0}}$ for any $p\times k$ matrix $\mat{P}_{\mat{\theta}_0}$ such that $k\geq d$ and $\Span{\mat{P}_{\mat{\theta}_0}} = T_{\mat{\theta}_0}\Theta$. This also ensures that the final result is independent of the chosen chart $\varphi$, since the tangent space does not depend on a specific chart. Therefore, let $\PP = {\mat{Q}}{\mat{R}}$ and $\mat{P}_{\mat{\theta}_0} = \widetilde{\mat{Q}}\widetilde{\mat{R}}$ be their thin QR decompositions, respectively. Both $\mat{Q}, \widetilde{\mat{Q}}$ have dimension $p\times d$, $\mat{Q}$ is semi-orthogonal, and $\mat{R}$ is invertible of dimension $d\times d$ while $\widetilde{\mat{R}}$ is a $d\times k$ full row-rank matrix. Since $\mat{Q}$ is semi-orthogonal, the $p\times p$ matrix $\mat{Q}\t{\mat{Q}}$ is an orthogonal projection onto $\Span\mat{Q} = \Span\mat{P}_{\mat{\theta}_0} = T_{\mat{\theta}_0}\Theta$. This allows to express $\mat{P}_{\mat{\theta}_0}$ in terms of $\mat{Q}$ as
        \begin{displaymath}
            \mat{P}_{\mat{\theta}_0} = \mat{Q}\t{\mat{Q}}\mat{P}_{\mat{\theta}_0}
                = \mat{Q}\t{\mat{Q}}\widetilde{\mat{Q}}\widetilde{\mat{R}} =: {\mat{Q}}\mat{M}.
        \end{displaymath}
        From $\Span\mat{Q} = \Span\mat{P}_{\mat{\theta}_0}$, it follows that the $d\times k$ matrix $\mat{M}$ is also of full row-rank. Also, $\mat{M}\pinv{\mat{M}} = \mat{I}_d = \mat{R}\mat{R}^{-1}$ as a property of the Moore-Penrose pseudo inverse with $\mat{M}$ being of full row-rank. Another property of the pseudo inverse is that for matrices $\mat{A}, \mat{B}$, where $\mat{A}$ has full column-rank and $\mat{B}$ has full row-rank, $\pinv{(\mat{A}\mat{B})} = \pinv{\mat{B}}\pinv{\mat{A}}$. This enables the computation
        \begin{multline*}
            \mat{P}_{\mat{\theta}_0}\pinv{(\t{\mat{P}_{\mat{\theta}_0}}\mat{H}_{\mat{\theta}_0}\mat{P}_{\mat{\theta}_0})}\t{\mat{P}_{\mat{\theta}_0}}= \mat{Q} \mat{M} \pinv{\mat{M}} (\t{\mat{Q}} \mat{H}_{\mat{\theta}_0} \mat{Q})^{-1} \t{(\mat{M} \pinv{\mat{M}})} \t{\mat{Q}} \\
            = \mat{Q} {\mat{R}} {\mat{R}}^{-1} (\t{\mat{Q}} \mat{H}_{\mat{\theta}_0} \mat{Q})^{-1} \t{({\mat{R}} {\mat{R}}^{-1})} \t{\mat{Q}}
            = \PP(\t{\PP}\mat{H}_{\mat{\theta}_0}\PP)^{-1}\t{\PP}
            = \mat{\Pi}_{\mat{\theta}_0}.
        \end{multline*}
    }
\end{proof}

Next, we rewrite the log-likelihood \eqref{eq:log-likelihood} in a simpler form to simplify the proof of \cref{thm:asymptotic-normality-gmlm} and to provide the notation to express the regularity conditions of \cref{thm:asymptotic-normality-gmlm} in a compact form.

Rewriting the first natural parameter component $\mat{\eta}_{1y}$ defined in \eqref{eq:eta1-manifold} gives
\begin{align*}
    \mat{\eta}_{1y}
    &= \vec{\overline{\ten{\eta}}} + \mat{B}\vec{\ten{F}_y}= \mat{I}_p\vec{\overline{\ten{\eta}}} + (\t{(\vec{\ten{F}_y})}\otimes\mat{I}_p)\vec{\mat{B}} \\
    &= \begin{pmatrix}
    \mat{I}_p & \t{(\vec{\ten{F}_y})}\otimes\mat{I}_p
    \end{pmatrix}\begin{pmatrix}
    \vec{\overline{\ten{\eta}}} \\
    \vec{\mat{B}}
    \end{pmatrix}.
\end{align*}
For the second natural parameter component $\mat{\eta}_2$, modeled in \eqref{eq:eta2-manifold}, we have
\begin{displaymath}
    \langle \mat{\eta}_2, \mat{T}_2\vech((\vec{\ten{X}})\t{(\vec{\ten{X}})}) \rangle
        = \langle \t{(\mat{T}_2\pinv{\mat{D}_p})}\mat{\eta}_2, \vec(\ten{X}\circ\ten{X}) \rangle
        = \langle c\,\mat{\Omega}, \ten{X}\circ\ten{X} \rangle
\end{displaymath}
which means that
\begin{displaymath}
    c \vec{\mat{\Omega}} = \t{(\mat{T}_2\pinv{\mat{D}_p})}\mat{\eta}_2.
\end{displaymath}
The inverse relation is
\begin{displaymath}
    \mat{\eta}_2 = c\t{(\pinv{(\mat{T}_2\pinv{\mat{D}_p})})}\vec\mat{\Omega} = c\t{(\pinv{(\mat{T}_2\pinv{\mat{D}_p})})}\mat{D}_p\vech\mat{\Omega},
\end{displaymath}
describing the linear relation between $\mat{\eta}_2$ and $\vech{\mat{\Omega}}$. This gives the following relation between $\mat{\eta}_y = (\mat{\eta}_{1y}, \mat{\eta}_2)$ and $\mat{\xi} = (\vec{\overline{\ten{\eta}}}, \vec{\mat{B}}, \vech{\mat{\Omega}})\in\Xi$,
\begin{equation}
    \mat{\eta}_y = \begin{pmatrix}
        \mat{I}_p & \t{(\vec{\ten{F}_y})}\otimes\mat{I}_p & 0 \\
        0 & 0 & c\t{(\pinv{(\mat{T}_2\pinv{\mat{D}_p})})}\mat{D}_p
    \end{pmatrix}\begin{pmatrix}
        \vec{\overline{\ten{\eta}}} \\
        \vec{\mat{B}} \\
        \vech{\mat{\Omega}}
    \end{pmatrix} =: \mat{F}(y)\mat{\xi} \label{eq:eta-to-xi-linear-relation}
\end{equation}
where $\mat{F}(y)$ is a $(p + d)\times p (p + 2 q + 3) / 2$ dimensional matrix-valued function in $y$. Moreover, for every $y$ the matrix $\mat{F}(y)$ is of full rank $p + d$.
The log-likelihood of model \eqref{eq:quad-density} for the unconstrained parameters $\xi\in\Xi$ is
\begin{displaymath}
    l_n(\mat{\xi})
        = \frac{1}{n}\sum_{i = 1}^{n} (\langle \mat{t}(\ten{X}), \mat{\eta}_{y} \rangle - b(\mat{\eta}_y))
        =: \frac{1}{n}\sum_{i = 1}^{n} m_{\mat{\xi}}(Z_i)
\end{displaymath}
where $Z_i = (\ten{X}_i, Y_i)$. Using \eqref{eq:eta-to-xi-linear-relation} we can write
\begin{displaymath}
    m_{\mat{\xi}}(z) = \langle\mat{t}(\ten{X}), \mat{F}(y)\mat{\xi}\rangle - b(\mat{F}(y)\mat{\xi}).
\end{displaymath}

The following are the regularity conditions for the log-likelihood required by \cref{thm:asymptotic-normality-gmlm}.

\begin{condition}\label{cond:differentiable-and-convex}
    The mapping $\mat{\xi}\mapsto m_{\mat{\xi}}(z)$ is twice continuously differentiable for almost every $z$ and $z\mapsto m_{\mat{\xi}}(z)$ is measurable. Moreover, $\mat{\eta}\mapsto b(\mat{\eta})$ is strictly convex. Furthermore, for every $\widetilde{\mat{\eta}}$, $P(\mat{F}(Y)\mat{\xi} = \widetilde{\mat{\eta}}) < 1$.
\end{condition}

\begin{condition}\label{cond:moments}
    $\E\|\t{\mat{t}(\ten{X})}\mat{F}(Y)\| < \infty$, and $\E\|\t{\mat{t}(\ten{X})}\mat{F}(Y)\|^2 < \infty$.
\end{condition}

\begin{condition}\label{cond:finite-sup-on-compacta}
    The mapping $\mat{\eta}\mapsto b(\mat{\eta})$ is twice continuously differentiable and for every non-empty compact $K\subseteq\Xi$,
    \begin{gather*}
        \E\sup_{\mat{\xi}\in K}\|b(\mat{F}(Y)\mat{\xi})\| < \infty, \qquad
        \E\sup_{\mat{\xi}\in K}\|\t{\nabla b(\mat{F}(Y)\mat{\xi})}\mat{F}(Y)\|^2 < \infty, \\
        \E\sup_{\mat{\xi}\in K}\| \t{\mat{F}(Y)}\nabla^2 b(\mat{F}(Y)\mat{\xi})\mat{F}(Y) \| < \infty.
    \end{gather*}
\end{condition}

The following are technical Lemmas needed for the proof of \cref{thm:asymptotic-normality-gmlm}.

\begin{lemma}\label{thm:kron-perm}
    Given $r \geq 2$ matrices $\mat{A}_k$ of dimension $p_j\times q_j$ for $k = 1, \ldots, r$, then there exists a unique permutation matrix $\mat{S}_{\mat{p}, \mat{q}}$ such that
    \begin{equation}\label{eq:kron-to-outer-perm}
        \vec\bigkron_{k = r}^{1}\mat{A}_k = \mat{S}_{\mat{p}, \mat{q}}\vec\bigouter_{k = 1}^{r}\mat{A}_k.
    \end{equation}
    The permutation $\mat{S}_{\mat{p}, \mat{q}}$ with indices $\mat{p} = (p_1, \ldots, p_r)$ and $\mat{q} = (q_1, \ldots, q_r)$ is defined recursively as
    \begin{equation}\label{eq:S_pq}
        \mat{S}_{\mat{p}, \mat{q}} = \mat{S}_{\bigl( \prod_{k = 1}^{r - 1}p_k, p_r \bigr), \bigl( \prod_{k = 1}^{r - 1}q_k, q_r \bigr)} \bigl(\mat{I}_{p_r q_r}\otimes\mat{S}_{(p_1, \ldots, p_{r-1}), (q_1, \ldots, q_{r-1})}\bigr)
    \end{equation}
    with initial value
    \begin{displaymath}
        \mat{S}_{(p_1, p_2), (q_1, q_2)} = \mat{I}_{q_2}\otimes\mat{K}_{q_1, p_2}\otimes\mat{I}_{p_1}
    \end{displaymath}
    where $\mat{K}_{p, q}$ is the \emph{commutation matrix} \cite[][Ch.~11]{AbadirMagnus2005};  that is, the permutation such that $\vec{\t{\mat{A}}} = \mat{K}_{p, q}\vec{\mat{A}}$ for every $p\times q$ dimensional matrix $\mat{A}$.
\end{lemma}

\begin{proof}
    Lemma~7 in \cite{MagnusNeudecker1986} states that
    \begin{align}
        \vec(\mat{A}_2\otimes\mat{A}_1)
        &= (\mat{I}_{q_2}\otimes\mat{K}_{q_1, p_2}\otimes\mat{I}_{p_1})(\vec{\mat{A}_2}\otimes\vec{\mat{A}_1}) \label{eq:MagnusNeudecker1986-vec-kron-identity} \\
        &= (\mat{I}_{q_2}\otimes\mat{K}_{q_1, p_2}\otimes\mat{I}_{p_1})\vec(\mat{A}_1\circ \mat{A}_2). \nonumber
    \end{align}
    This proves the statement for $r = 2$. The general statement for $r > 2$ follows by induction. Assuming \eqref{eq:kron-to-outer-perm} holds for $r - 1$, the induction step is 
    \begin{multline*}
        \vec{\bigkron_{k = r}^{1}}\mat{A}_k
        = \vec\Bigl(\mat{A}_r\otimes\bigkron_{k = r - 1}^{1}\mat{A}_k\Bigr) \\
        \overset{\eqref{eq:MagnusNeudecker1986-vec-kron-identity}}{=} \Bigl( \mat{I}_{q_r}\otimes\mat{K}_{\prod_{k = 1}^{r - 1}q_k, p_r}\otimes\mat{I}_{\prod_{k = 1}^{r - 1}p_k} \Bigr)\vec\Bigl((\vec\mat{A}_r)\otimes\vec\bigkron_{k = r - 1}^{1}\mat{A}_k\Bigr) \\
        = \mat{S}_{\bigl( \prod_{k = 1}^{r - 1}p_k, p_r \bigr), \bigl( \prod_{k = 1}^{r - 1}q_k, q_r \bigr)}\vec\Bigl[\Bigl(\vec\bigkron_{k = r - 1}^{1}\mat{A}_k\Bigr)\t{(\vec\mat{A}_r)}\Bigr] \\
        \overset{\eqref{eq:kron-to-outer-perm}}{=} \mat{S}_{\bigl( \prod_{k = 1}^{r - 1}p_k, p_r \bigr), \bigl( \prod_{k = 1}^{r - 1}q_k, q_r \bigr)}\vec\Bigl[\mat{S}_{(p_1, \ldots, p_{r-1}), (q_1, \ldots, q_{r-1})}\Bigl(\vec\bigouter_{k = 1}^{r - 1}\mat{A}_k\Bigr)\t{(\vec\mat{A}_r)}\Bigr] \\
        \overset{(a)}{=} \mat{S}_{\bigl( \prod_{k = 1}^{r - 1}p_k, p_r \bigr), \bigl( \prod_{k = 1}^{r - 1}q_k, q_r \bigr)} \bigl(\mat{I}_{p_rq_r}\otimes\mat{S}_{(p_1, \ldots, p_{r-1}), (q_1, \ldots, q_{r-1})}\bigr)\vec\Bigl[\Bigl(\vec\bigouter_{k = 1}^{r - 1}\mat{A}_k\Bigr)\t{(\vec\mat{A}_r)}\Bigr] \\
        =\mat{S}_{\mat{p},\mat{q}}\vec\bigouter_{k = 1}^{r}\mat{A}_k.
    \end{multline*}
    Equality $(a)$ uses the relation $\vec(\mat{C}\mat{a}\t{\mat{b}}) = (\mat{I}_{\dim(\mat{b})}\otimes\mat{C})\vec(\mat{a}\t{\mat{b}})$ for a matrix $\mat{C}$ and vectors $\mat{a}, \mat{b}$.
\end{proof}

\begin{lemma}\label{thm:kron-manifold-tangent-space}
    Let $\manifold{A}_k\subseteq\mathbb{R}^{p_k\times q_k}\backslash\{\mat{0}\}$ for $k = 1, \ldots, r$ be smooth embedded submanifolds as well as ether a sphere or a cone. Then
    \begin{displaymath}
        \manifold{K} = \Bigl\{ \bigkron_{k = r}^{1}\mat{A}_k : \mat{A}_k\in\manifold{A}_k \Bigr\}
    \end{displaymath}
    is an embedded manifold in $\mathbb{R}^{p\times q}$ for $p = \prod_{k = 1}^{r} p_k$ and $q = \prod_{k = 1}^{r} q_k$.
    Furthermore, define for $j = 1, \ldots, r$ the matrices
    \begin{equation}\label{eq:kron-differential-span}
        \mat{\Gamma}_j
            = \bigkron_{k = r}^{1}(\mat{I}_{p_k q_k}\mathrm{\ if\ } j = k \mathrm{\ else\ }\vec{\mat{A}_k})
            = \bigkron_{k = r}^{j + 1}(\vec{\mat{A}_k})\otimes\mat{I}_{p_j q_j}\otimes\bigkron_{k = j - 1}^{1}(\vec{\mat{A}_k})
    \end{equation}
    and let $\gamma_j$ be $p_j q_j\times d_j$ matrices with $d_j \geq\dim\manifold{A}_j$ which span the tangent space $T_{\mat{A}_j}\manifold{A}_j$ of $\manifold{A}$ at $\mat{A}_j\in\manifold{A}_j$, that is $\Span\gamma_j = T_{\mat{A}_j}\manifold{A}_j$.
    Then, with the permutation matrix $\mat{S}_{\mat{p}, \mat{q}}$ defined in \eqref{eq:S_pq}, the $p q \times \sum_{k = 1}^{r} d_j$ dimensional matrix
    \begin{displaymath}
        \mat{P}_{\mat{A}} = \mat{S}_{\mat{p}, \mat{q}}\left[\mat{\Gamma}_1\mat{\gamma}_1, \mat{\Gamma}_2\mat{\gamma}_2, \ldots, \mat{\Gamma}_r\mat{\gamma}_r\right]
    \end{displaymath}
    spans the tangent space $T_{\mat{A}}\manifold{K}$ of $\manifold{K}$ at $\mat{A} = \bigkron_{k = r}^{1}\mat{A}_k\in\manifold{K}$.
\end{lemma}

\begin{proof}
    The statement that $\manifold{K}$ is an embedded manifold follows via induction using \cref{thm:kron-manifolds}.
    We compute the differential of the vectorized Kronecker product using \cref{thm:kron-perm} where $\mat{S}_{\mat{p}, \mat{q}}$ is the permutation \eqref{eq:S_pq} defined therein.
    \begin{multline*}
        \d\vec\bigotimes_{k = r}^{1}\mat{A}_k
        = \vec\sum_{j = 1}^{r}\bigkron_{k = r}^{1}(\ternary{k = j}{\d\mat{A}_j}{\mat{A}_k}) \\
        =\mat{S}_{\mat{p},\mat{q}}\vec\sum_{j = 1}^{r}\Bigl(\bigouter_{k = 1}^{r}(\ternary{k = j}{\d\mat{A}_j}{\mat{A}_k})\Bigr)
        =\mat{S}_{\mat{p}, \mat{q}}\sum_{j = 1}^{r}\bigkron_{k = r}^{1}(\ternary{k = j}{\vec\d\mat{A}_j}{\vec\mat{A}_k}) \\
        =\mat{S}_{\mat{p}, \mat{q}}\sum_{j =1}^{r}\Bigl(\bigkron_{k = r}^{1}(\ternary{k = j}{\mat{I}_{p_j q_j}}{\vec\mat{A}_k})\Bigr)\vec\d\mat{A}_j= \mat{S}_{\mat{p}, \mat{q}}\sum_{j = 1}^{r}\mat{\Gamma}_j\vec\d\mat{A}_j \\
        =\mat{S}_{\mat{p}, \mat{q}}[\mat{\Gamma}_1, \ldots, \mat{\Gamma}_r]\begin{pmatrix}
            \vec\d\mat{A}_1 \\ \vdots \\ \vec\d\mat{A}_r
        \end{pmatrix}
    \end{multline*}
    Due to the definition of the manifold, this differential provides the gradient of a surjective map into the manifold. The span of the gradient then spans the tangent space.

    Now, we take a closer look at the differentials $\vec{\d\mat{A}_j}$ for $j = 1, \ldots, r$. Let $\varphi_j$ be a chart of $\manifold{A}_j$ in a neighborhood of $\mat{A}_j$. Then, $\mat{A}_j = \varphi_j^{-1}(\varphi_j(\mat{A}_j))$ which gives
    \begin{displaymath}
        \vec{\d\mat{A}_j} = \t{\nabla\varphi_j^{-1}(\varphi_j(\mat{A}_j))}\vec\d\varphi_j(\mat{A}_j).
    \end{displaymath}
    Therefore, for every matrix $\mat{\gamma}_j$ such that $\Span{\mat{\gamma}_j} = T_{\mat{A}_j}\manifold{A}_j$ holds $\Span{\t{\nabla\varphi_j^{-1}(\varphi_j(\mat{A}_j))}} = \Span{\mat{\gamma}_j}$ by \cref{def:tangent-space} of the tangent space. We get
    \begin{displaymath}
        \Span\mat{S}_{\mat{p}, \mat{q}}[\mat{\Gamma}_1, \ldots, \mat{\Gamma}_r]\begin{pmatrix}
        \vec\d\mat{A}_1 \\ \vdots \\ \vec\d\mat{A}_r
    \end{pmatrix}
        =\Span\mat{S}_{\mat{p}, \mat{q}}[\mat{\Gamma}_1\mat{\gamma}_1, \ldots, \mat{\Gamma}_r\mat{\gamma}_r]
        =\Span\mat{P}_{\mat{A}}
    \end{displaymath}
    which concludes the proof.
\end{proof}

\begin{proof}[Proof of \cref{thm:asymptotic-normality-gmlm}]
    The proof consists of three parts. First, we show the existence of a consistent strong M-estimator by applying \cref{thm:M-estimator-consistency-on-subsets}. Next, we apply \cref{thm:M-estimator-asym-normal-on-manifolds} to obtain its asymptotic normality. We conclude by computing the missing parts of the asymtotic covariance matrix $\mat{\Sigma}_{\mat{\theta}_0}$ provided by \cref{thm:M-estimator-asym-normal-on-manifolds}.

    We check whether the conditions of \cref{thm:M-estimator-consistency-on-subsets} are satisfied. On $\Xi$, the mapping $\mat{\xi}\mapsto m_{\mat{\xi}}(z) = m_{\mat{\xi}}(\ten{X},y) = \langle \mat{F}(y)\mat{\xi}, \mat{t}(\ten{X}) \rangle - b(\mat{F}(y)\mat{\xi})$ is strictly concave for every $z$ because $\mat{\xi}\mapsto\mat{F}(y)\mat{\xi}$ is linear and $b$ is strictly convex by \cref{cond:differentiable-and-convex}. Since $\ten{X} \mid Y$ is distributed according to \eqref{eq:quadratic-exp-fam}, the function $M(\mat{\xi}) = \E m_{\mat{\xi}}(Z)$ is well defined by \cref{cond:moments}. Let $\mat{\xi}_k = (\vec{\overline{\ten{\eta}}_k}, \vec{\mat{B}_k}, \vech{\mat{\Omega}_k})$, and $f_{\mat{\xi}_k}$ be the pdf of $\ten{X} \mid Y$ indexed by $\mat{\xi}_k$, for $k = 1, 2$. If $\mat{\xi}_1\ne \mat{\xi}_2$, then $f_{\mat{\xi}_1} \neq f_{\mat{\xi}_2}$, which obtains that the true $\mat{\theta}_0$ is a unique maximizer of $\mat{\theta}_0\in\Theta\subseteq\Xi$ by applying Lemma~5.35 from \cite{vanderVaart1998}. Finally, under \cref{cond:finite-sup-on-compacta}, all assumptions of \cref{thm:M-estimator-consistency-on-subsets} are fulfilled yielding the existence of an consistent strong M-estimator over $\Theta\subseteq\Xi$.

    Next, let $\hat{\mat{\theta}}_n$ be a strong M-estimator on $\Theta\subseteq\Xi$, whose existence and consistency was shown in the previous step. Since $z\mapsto m_{\mat{\xi}}(z)$ is measurable for all $\mat{\xi}\in\Xi$, it is also measurable in a neighborhood of $\mat{\theta}_0$. The differentiability of $\mat{\theta}\mapsto m_{\mat{\theta}}(z)$ is stated in \cref{cond:differentiable-and-convex}. For the Lipschitz condition, let $K\subseteq\Xi$ be a compact neighborhood of $\mat{\theta}_0$, which exists since $\Xi$ is open. Then,
    \begin{multline*}
        \left| m_{\mat{\theta}_1}(z) - m_{\mat{\theta}_2}(z) \right|
        = \left| \langle \mat{t}(\ten{X}), \mat{F}(y)(\mat{\theta}_1 - \mat{\theta}_2) \rangle - b(\mat{F}(z)\mat{\theta}_1) + b(\mat{F}(z)\mat{\theta}_2) \right| \\
        \leq (\| \t{\mat{F}(y)}\mat{t}(\ten{X}) \|_2 + \sup_{\mat{\theta}\in K}\| \nabla b(\mat{F}(y)\mat{\theta}) \mat{F}(y)\| ) \| \mat{\theta}_1 - \mat{\theta}_2 \|_2
        =: u(z)\| \mat{\theta}_1 - \mat{\theta}_2 \|_2
    \end{multline*}
    with $u(z)$ being measurable and square integrable derives from \cref{cond:finite-sup-on-compacta}. The existence of a second-order Taylor expansion of $\mat{\theta}\mapsto M(\mat{\theta}) = \E m_{\mat{\theta}}(Z)$ in a neighborhood of $\mat{\theta}_0$ holds by \cref{cond:finite-sup-on-compacta}. Moreover, the Hessian $\mat{H}_{\mat{\theta}_0}$ is non-singular by the strict convexity of $b$ stated in \cref{cond:differentiable-and-convex}. Now, we can apply \cref{thm:M-estimator-asym-normal-on-manifolds} to obtain the asymptotic normality of $\sqrt{n}(\hat{\mat{\theta}}_n - \mat{\theta}_0)$ with variance-covariance structure 
    \begin{equation}\label{eq:asymptotic-covariance-gmlm}
        \mat{\Sigma}_{\mat{\theta}_0} = \mat{\Pi}_{\mat{\theta}_0} \E[\nabla m_{\mat{\theta}_0}(Z)\t{(\nabla m_{\mat{\theta}_0}(Z))}]\mat{\Pi}_{\mat{\theta}_0}
    \end{equation}
    where $\mat{\Pi}_{\mat{\theta}_0} = \mat{P}_{\mat{\theta}_0}(\t{\mat{P}_{\mat{\theta}_0}}\mat{H}_{\mat{\theta}_0}\mat{P}_{\mat{\theta}_0})^{-1}\t{\mat{P}_{\mat{\theta}_0}}$ and $\mat{P}_{\mat{\theta}_0}$ is any $p\times \dim(\Theta)$ matrix such that it spans the tangent space of $\Theta$ at $\mat{\theta}_0$. That is, $\Span \mat{P}_{\mat{\theta}_0} = T_{\mat{\theta}_0}\Theta$.

    Finally, we compute a matrix $\mat{P}_{\mat{\theta}_0}$ such that $\Span{\mat{P}_{\mat{\theta}_0}} = T_{\mat{\theta}_0}\Theta$ for $\Theta = \mathbb{R}^p\times\manifold{K}_{\mat{B}}\times\manifold{CK}_{\mat{\Omega}}$ as in \cref{thm:param-manifold}. Since the manifold $\Theta$ is a product manifold we get a block diagonal structure for $\mat{P}_{\mat{\theta}_0}$ as
    \begin{displaymath}
        \mat{P}_{\mat{\theta}_0} = \begin{pmatrix}
            \mat{I}_p & 0 & 0 \\
            0 & \mat{P}_{\mat{B}_0} & 0 \\
            0 & 0 & \mat{P}_{\mat{\Omega}_0}
        \end{pmatrix}
    \end{displaymath}
    where $\mat{I}_p$ is the identity matrix spanning the tangent space of $\mathbb{R}^p$, which is identified with $\mathbb{R}^p$ itself. The blocks $\mat{P}_{\mat{B}_0}$ and $\mat{P}_{\mat{\Omega}_0}$ need to span the tangent spaces of $\manifold{K}_{\mat{B}}$ and $\manifold{CK}_{\mat{\Omega}}$, respectively. Both $\manifold{K}_{\mat{B}}$ and $\manifold{CK}_{\mat{\Omega}}$ are manifolds according to \cref{thm:kron-manifolds} under the cone condition. The constraint manifold $\manifold{CK}_{\mat{\Omega}}$ is the intersection of $\manifold{K}_{\mat{\Omega}}$ with the span of the projection $\pinv{(\mat{T}_2\pinv{\mat{D}_p})}\mat{T}_2\pinv{\mat{D}_p}$ meaning that the differential $\vec{\d\mat{\Omega}}$ on $\manifold{CK}_{\mat{\Omega}}$ fulfills $\vec{\d\mat{\Omega}} = \pinv{(\mat{T}_2\pinv{\mat{D}_p})}\mat{T}_2\pinv{\mat{D}_p}\vec{\d\mat{\Omega}}$. Now, we can apply \cref{thm:kron-manifold-tangent-space} for $\manifold{K}_{\mat{B}}$ and $\manifold{K}_{\mat{\Omega}}$ which give
    \begin{align*}
        \mat{P}_{\mat{B}_0} &= \mat{S}_{\mat{p}, \mat{q}}[\mat{\Gamma}_{\mat{\beta}_1}\mat{\gamma}_{\mat{\beta}_1}, \ldots, \mat{\Gamma}_{\mat{\beta}_r}\mat{\gamma}_{\mat{\beta}_r}], \\
        \mat{P}_{\mat{\Omega}_0} &= \pinv{(\mat{T}_2\pinv{\mat{D}_p})}\mat{T}_2\pinv{\mat{D}_p}\mat{S}_{\mat{p}, \mat{p}}[\mat{\Gamma}_{\mat{\Omega}_1}\mat{\gamma}_{\mat{\Omega}_1}, \ldots, \mat{\Gamma}_{\mat{\Omega}_r}\mat{\gamma}_{\mat{\Omega}_r}]
    \end{align*}
    where the matrices $\mat{S}_{\mat{p}, \mat{q}}$, $\mat{\Gamma}_{\mat{\beta}_j}$, $\mat{\gamma}_{\mat{\beta}_j}$, $\mat{\Gamma}_{\mat{\Omega}_j}$ and $\mat{\gamma}_{\mat{\Omega}_j}$ are described in \cref{thm:kron-manifold-tangent-space} for the Kronecker manifolds $\manifold{K}_{\mat{B}}$ and $\manifold{K}_{\mat{\Omega}}$. Leading to
    \begin{equation}\label{eq:param-manifold-span}
        \mat{P}_{\mat{\theta}_0} = \begin{pmatrix}
            \mat{I}_p & 0 & 0 \\
            0 & \mat{S}_{\mat{p}, \mat{q}}[\mat{\Gamma}_{\mat{\beta}_1}\mat{\gamma}_{\mat{\beta}_1}, \ldots, \mat{\Gamma}_{\mat{\beta}_r}\mat{\gamma}_{\mat{\beta}_r}] & 0 \\
            0 & 0 & \pinv{(\mat{T}_2\pinv{\mat{D}_p})}\mat{T}_2\pinv{\mat{D}_p}\mat{S}_{\mat{p}, \mat{p}}[\mat{\Gamma}_{\mat{\Omega}_1}\mat{\gamma}_{\mat{\Omega}_1}, \ldots, \mat{\Gamma}_{\mat{\Omega}_r}\mat{\gamma}_{\mat{\Omega}_r}]
        \end{pmatrix}.
    \end{equation}
\end{proof}

\section{Examples}\label{app:examples}

\begin{example}[Vectorization]\label{ex:vectorization}
    Given a matrix
    \begin{displaymath}
        \mat{A} = \begin{pmatrix}
            1 & 4 & 7 \\
            2 & 5 & 8 \\
            3 & 6 & 9
        \end{pmatrix}
    \end{displaymath}
    its vectorization is $\vec{\mat{A}} = \t{(1, 2, 3, 4, 5, 6, 7, 8, 9)}$ and its half vectorization $\vech{\mat{A}} = \t{(1, 2, 3, 5, 6, 9)}$. Let a $\ten{A}$ be a tensor of dimension $3\times 3\times 3$ given by
    \begin{displaymath}
        \ten{A}_{:,:,1} = \begin{pmatrix}
            1 & 4 & 7 \\
            2 & 5 & 8 \\
            3 & 6 & 9
        \end{pmatrix},
        \qquad
        \ten{A}_{:,:,2} = \begin{pmatrix}
            10 & 13 & 16 \\
            11 & 14 & 17 \\
            12 & 15 & 18
        \end{pmatrix},
        \qquad
        \ten{A}_{:,:,3} = \begin{pmatrix}
            19 & 22 & 25 \\
            20 & 23 & 26 \\
            21 & 24 & 27
        \end{pmatrix}.
    \end{displaymath}
    Then the vectorization of $\ten{A}$ is given by
    \begin{displaymath}
        \vec{\ten{A}} = \t{(1, 2, 3, 4, ..., 26, 27)}\in\mathbb{R}^{27}.
    \end{displaymath}
\end{example}

\begin{example}[Matricization]
    Let $\ten{A}$ be the $3\times 4\times 2$ tensor given by
    \begin{displaymath}
        \ten{A}_{:,:,1} = \begin{pmatrix}
            1 & 4 & 7 & 10 \\
            2 & 5 & 8 & 11 \\
            3 & 6 & 9 & 12
        \end{pmatrix},
        \ten{A}_{:,:,2} = \begin{pmatrix}
            13 & 16 & 19 & 22 \\
            14 & 17 & 20 & 23 \\
            15 & 18 & 21 & 24
        \end{pmatrix}.
    \end{displaymath}
    Its matricizations are then
    \begin{gather*}
        \ten{A}_{(1)} = \begin{pmatrix}
            1 & 4 & 7 & 10 & 13 & 16 & 19 & 22 \\
            2 & 5 & 8 & 11 & 14 & 17 & 20 & 23 \\
            3 & 6 & 9 & 12 & 15 & 18 & 21 & 24
        \end{pmatrix},
        \qquad
        \ten{A}_{(2)} = \begin{pmatrix}
             1 &  2 &  3 & 13 & 14 & 15 \\
             4 &  5 &  6 & 16 & 17 & 18 \\
             7 &  8 &  9 & 19 & 20 & 21 \\
            10 & 11 & 12 & 22 & 23 & 24
        \end{pmatrix}, \\
        \ten{A}_{(3)} = \begin{pmatrix}
             1 &  2 &  3 &  4 &  5 &  6 &  7 &  8 &  9 &  10 &  11 &  12 \\
            13 & 14 & 15 & 16 & 17 & 18 & 19 & 20 & 21 &  22 &  23 &  24
        \end{pmatrix}.
    \end{gather*}
\end{example}

\end{document}